\documentclass[11pt,reqno]{amsart}
\usepackage{amsmath,amsxtra,amssymb,amsthm,amsfonts,bm}

\usepackage[latin1]{inputenc}
\usepackage[cal=cm]{mathalfa}
  \usepackage[foot]{amsaddr}

\usepackage{bbold}
\usepackage{bbm}
\usepackage{nonfloat}
\usepackage{braket}
\usepackage{dsfont}
\usepackage{mathdots}
\usepackage{mathtools}
\usepackage{enumerate}
\usepackage{csquotes}
\usepackage{stmaryrd}
\usepackage{graphicx}
\usepackage{stackengine}
\usepackage{scalerel}
\usepackage{array}
\usepackage{makecell}
\newcolumntype{x}[1]{>{\centering\arraybackslash}p{#1}}
\usepackage{tikz}
\usetikzlibrary{shapes.geometric, shapes.misc, positioning, arrows, decorations.pathreplacing, angles, quotes}
\usepackage{booktabs}
\usepackage{xfrac}
\usepackage{siunitx}
\usepackage{centernot}
\usepackage{comment}
\usepackage{chngcntr}

\newtheorem{theorem}{Theorem}
\numberwithin{theorem}{section} % important bit

\newtheorem*{theorem*}{Theorem}
\newtheorem{proposition}[theorem]{Proposition}
\newtheorem*{proposition*}{Proposition}
\newtheorem{lemma}[theorem]{Lemma}
\newtheorem*{lemma*}{Lemma}
\newtheorem{corollary}[theorem]{Corollary}
\newtheorem*{cor*}{Corollary}

\newtheorem*{cj*}{Conjecture}
\newtheorem{Def}[theorem]{Definition}
\newtheorem*{Def*}{Definition}

\newtheorem{problem}[theorem]{Problem}

\makeatletter
\def\thmhead@plain#1#2#3{%
  \thmname{#1}\thmnumber{\@ifnotempty{#1}{ }\@upn{#2}}%
  \thmnote{ {\the\thm@notefont#3}}}
\let\thmhead\thmhead@plain
\makeatother

\theoremstyle{definition}
\newtheorem{rem}[theorem]{Remark}

\newtheorem{example}[theorem]{Example}

\newcommand{\vertiii}[1]{{\left\vert\kern-0.25ex\left\vert\kern-0.25ex\left\vert #1
    \right\vert\kern-0.25ex\right\vert\kern-0.25ex\right\vert}}

\newcommand{\ten}{\otimes}

\newcommand{\pl}{\hspace{.1cm}}

\newcommand{\ran}{\rangle}
\newcommand{\lan}{\langle}
\newcommand{\al}{\alpha}
\renewcommand{\si}{\sigma}

\newcommand{\la}{\lambda}
\newcommand{\eps}{\varepsilon}

\newcommand{\bb}{\begin{equation}}
\newcommand{\bbb}{\begin{equation*}}
\newcommand{\ee}{\end{equation}}
\newcommand{\eee}{\end{equation*}}

\newcommand{\M}{{\mathcal M}}

\newcommand{\cT}{\mathcal{T}}

\newcommand{\norm}[2]{\parallel \! #1 \! \parallel_{#2}}
\newcommand{\id}{\operatorname{id}}
\newcommand{\tr}{\operatorname{tr}}

\newcommand{\Id}{\mathds{1}}

\newcommand{\cD}{\mathcal{D}}
\newcommand{\cN}{\mathcal{N}}
\newcommand{\cM}{\mathcal{M}}
\newcommand{\cF}{\mathcal{F}}
\newcommand{\cH}{\mathcal{H}}
\newcommand{\cX}{\mathcal{X}}
\newcommand{\cB}{\mathcal{B}}

\newcommand{\cL}{\mathcal{L}}
\newcommand{\cP}{\mathcal{P}}
\newcommand{\cE}{\mathcal{E}}

\newcommand{\cZ}{\mathcal{Z}}
\newcommand{\cY}{\mathcal{Y}}

\newcommand{\CC}{\mathbb{C}}

\DeclareMathAlphabet{\pazocal}{OMS}{zplm}{m}{n}

\DeclareMathOperator{\supp}{supp}

\newcommand{\cK}{\mathcal{K}}

\newcommand{\lsmatrix}{\left(\begin{smallmatrix}}
\newcommand{\rsmatrix}{\end{smallmatrix}\right)}

\stackMath

\stackMath

\makeatletter
\newcommand*\rel@kern[1]{\kern#1\dimexpr\macc@kerna}
\newcommand*\widebar[1]{%
  \begingroup
  \def\mathaccent##1##2{%
    \rel@kern{0.8}%
    \overline{\rel@kern{-0.8}\macc@nucleus\rel@kern{0.2}}%
    \rel@kern{-0.2}%
  }%
  \macc@depth\@ne
  \let\math@bgroup\@empty \let\math@egroup\macc@set@skewchar
  \mathsurround\z@ \frozen@everymath{\mathgroup\macc@group\relax}%
  \macc@set@skewchar\relax
  \let\mathaccentV\macc@nested@a
  \macc@nested@a\relax111{#1}%
  \endgroup
}

\usepackage[pdftex]{hyperref}
\hypersetup{colorlinks=true,linkcolor=blue,filecolor=magenta,urlcolor=blue}
\usepackage{cleveref}

\title[Complete entropic inequalities for Quantum Markov chains]{Complete entropic inequalities \\for Quantum Markov chains
}

\author{Li Gao}
 \email[Li Gao]{li.gao@tum.de}
\author{Cambyse Rouz\'{e}}
 \address{Zentrum Mathematik, Technische Universit\"{a}t M\"{u}nchen, 85748 Garching, Germany}
\email[Cambyse Rouz\'{e}]{cambyse.rouze@tum.de}

\begin{document}
\begin{abstract}
We prove that every GNS-symmetric quantum Markov semigroup on a finite dimensional matrix algebra satisfies a modified log-Sobolev inequality. In the discrete time setting, we prove that every finite dimensional GNS-symmetric quantum channel satisfies a strong data processing inequality with respect to its decoherence free part. Moreover, we establish the first general approximate tensorization property of relative entropy. This extends the famous strong subadditivity of the quantum entropy (SSA) of two subsystems to the general setting of two subalgebras. All the three results are independent of the size of the environment and hence satisfy the tensorization property. They are obtained via a common, conceptually simple method for proving entropic inequalities via spectral or $L_2$-estimates. As applications, we combine our results on the modified log-Sobolev inequality and approximate tensorization to derive bounds for examples of both theoretical and practical relevance, including representation of sub-Laplacians on $\operatorname{SU}(2)$ and various classes of local quantum Markov semigroups such as quantum Kac generators and continuous time approximate unitary designs. For the latter, our bounds imply the existence of local continuous time Markovian evolutions on $nk$ qudits forming $\eps$-approximate $k$-designs in relative entropy for times scaling as $\widetilde{\mathcal{O}}(n^2 \operatorname{poly}(k))$.
\end{abstract}
\maketitle

\tableofcontents

\section{Introduction and main results}
The relative entropy is a fundamental information measure that has been widely used in probability, statistics and information theory. It was first introduced by Kullback and Leibler \cite{kullback1951information} for probability distributions (also called KL-divergence), and later extended by Umegaki \cite{umegaki1962conditional} to the noncommutative setting for quantum states. For two quantum states with density matrices $\rho$ and $\si$, %with $\operatorname{supp}(\rho)\subseteq \operatorname{supp}(\si)$,
the relative entropy of $\rho$ with respect to $\si$ is defined as
\begin{align}\label{eq:umegaki}D(\rho\|\si)=\tr(\rho\ln \rho-\rho\ln \si)\pl, \end{align}
where $\tr$ is the matrix trace. When $\rho$ and $\si$ share a same eigenbasis, \eqref{eq:umegaki} recovers the KL-divergence for two (discrete) probability densities. In both classical and quantum cases, $D(\rho\|\si)$ measures how well the classical or quantum state $\rho$ can be distinguished from $\sigma$ by statistical or quantum-mechanical experiments \cite{blahut1974hypothesis,hiai1991proper,ogawa2005strong}. In this work, we study several related inequalities of quantum relative entropy which have direct applications in quantum information theory and quantum many-body systems. Some of our results also yield new insights in the classical cases for probability distributions.\\

{\bf \noindent Modified logarithmic Sobolev inequality.}
The logarithmic Sobolev inequality is a functional inequality that was first introduced by Gross in his study of quantum field theory \cite{gross1975hypercontractivity} as an equivalent formulation of hypercontractivity \cite{nelson1973free}. Over the past decades, logarithmic Sobolev inequalities have been intensively studied for their applications in analysis, probability and information theory (see e.g.~the  \textit{} \cite{ledoux2011analytic,gross2014hypercontractivity} and the references therein). One of their variant formulation, called \emph{modified logarithmic Sobolev inequality}, is directly related to entropy.
Let $(\Omega,\mu)$ be a
probability space and $(\cT_t:L_\infty(\Omega,\mu)\to L_\infty(\Omega,\mu))_{t\ge0}$ be a Markov semigroup with the unique invariant measure $\mu$. $(\cT_t)_{t\ge 0}$ is said to satisfy the $\al$-modified logarithmic Sobolev inequality (in short, $\al$-MLSI ) for $\al>0$ if for any probability density $f\ge 0, \,\int f d\mu =1$,
\begin{align} \label{eq:classicalMLSI}\al \int f\ln f d\mu \le -\int L (f) \ln f d\mu, \end{align}
where $L$ is the generator of the semigroup, i.e.~$\cT_t=e^{Lt}$. The left hand side is the (classical) entropy functional $\text{Ent}(f):=\int f\ln f d\mu$. It is well known that $\al$-MLSI \eqref{eq:classicalMLSI} is equivalent to
\[ \text{Ent}(\cT_tf)\le e^{-\al t}\text{Ent}(f)\ ,\]
which means that the entropy of the system decays exponentially. This entropic convergence property is a powerful tool to derive mixing times of the semigroup.

The main purpose of this work is to study modified logarithmic Sobolev inequalities for quantum Markov semigroups. Quantum Markov semigroups are noncommutative generalizations of Markov semigroups where the underlying function spaces are replaced by matrix algebras or operator algebras. Let $\cH$ be a finite dimensional Hilbert space and let $\cB(\cH)$ be the bounded operators on $\cH$.
A quantum Markov semigroup (QMS) $(\cP_t: \cB(\cH)\to  \cB(\cH))_{t\ge 0}$ is a continuous semigroup of completely positive trace preserving maps. Such continuous time families of quantum channels model the Markovian evolution of dissipative open quantum systems. In recent years, the connection between logarithmic Sobolev inequalities and other functional inequalities, such as hypercontractivity, Poincar\'e inequality and transport cost inequality, have been largely extended to quantum Markov semigroup (see \cite{olkiewicz1999hypercontractivity,kastoryano2013quantum,datta2020relating,rouze2019concentration,carlen2017gradient,carlen2020non}). Some of them found direct applications in quantum information and quantum computational complexity (see e.g. \cite{muller2016entropy,beigi2020quantum,brandao2016local}).

Despite the rich connections to many aspects of quantum Markov processes, logarithmic Sobolev inequalities in the quantum framework are missing one key property---the tensorization property. For two classical Markov semigroups $(\mathcal{S}_t)_{t\ge 0}$ and $(\cT_t)_{t\ge 0}$, if
each semigroup satisfies $\al$-MLSI, then $(S_t\ten T_t)_{t\ge 0}$ also satisfies $\al$-MLSI \cite{bobkov2003modified} with the same constant $\al$. Tensorization is a powerful property that allow us to obtain MLSI for large, composite systems in terms of the dynamics on smaller subsystems, which is a technique that was already used by Gross in his very first work on the logarithmic Sobolev inequality. Nevertheless, the tensor stability of MLSI fails for general (non-ergodic) quantum Markov semigroups. The lack of tensorization property is a common difficulty in quantum information (see e.g.~the super-additivity of the channel capacity \cite{hastings2009superadditivity,smith2008quantum}). On the other hand, it was discovered in \cite{gao2020fisher} that the tensorization property is satisfied with a stronger definition of MLSI:
a quantum Markov semigroup $(\cP_t: \cB(\cH)\to  \cB(\cH))_{t\ge 0}$ is said to satisfy the $\al$-\textit{complete modified logarithmic Sobolev inequality} (in short, $\al$-CMLSI) if for any $n\ge 1$, the amplification $\cP_t\ten \id_{n}$
satisfies $\al$-MLSI, where $\id_{n}$ is the identity map on a $n$-dimensional quantum system.
Our first main result shows that such tensor stable modified log-Sobolev inequality generically holds in finite dimensions.

\begin{theorem} \label{thm:1}
Let $(\cP_t)_{t\ge 0}$ be a quantum Markov semigroup and denote by $\displaystyle E_*=\lim_{t\to\infty }\cP_t$ the projection onto its fixed point space. Suppose $(\cP_t)_{t\ge 0}$ is GNS-symmetric to some full-rank invariant state $\si$. Then for all $n\in\mathbb{N}$ and all states $\rho\in \cB(\cH\ten \mathbb{C}^n)$,
 \begin{align}\label{MLSI1}
 D(\cP_t\ten \id_{n}(\rho)\| E_*\ten \id_{n}(\rho))\le e^{-\al t}D(\rho\| E_*\ten \id_{n}(\rho))\pl .\pl  \tag{CMLSI}
 \end{align}
where $D(\cdot\|\cdot)$ denotes the relative entropy and the constant $\al$ satisfies
\[ \frac{\la}{C_{\operatorname{cb}}(E_*)}\le \al\le 2\la \pl. \]
Here $\la$ is the spectral gap and $C_{\operatorname{cb}}(E_*)$ is the complete Pimsner-Popa index of the map $E_*$.
\end{theorem}
We refer to Section \ref{sec:CMLSI} for details on the definition of the spectral gap and Section \ref{sec:keylemma} for the index $C_{\operatorname{cb}}(E_*)$. We remark that Theorem \ref{thm:1} asserts that not only the semigroup $(\cP_t)_{t\ge 0}$ itself but also its amplifications
$(\cP_t\ten \id_{n})_{t\ge 0}$, coupling an environment system $\mathbb{C}^n$, admit exponential decay of relative entropy for a uniform rate $\al$ for all dimension $n$. This definition was introduced in \cite{gao2020fisher}, and proved to satisfy the tensorization property: whenever two quantum Markov semigroups satisfy $\al$-CMLSI, their tensor product satisfies $\al$-CMLSI. Later, Li, Junge and LaRacuente \cite{li2020graph} proved that the heat semigroup of Riemannian manifolds of positive curvature and all classical (continuous-time) finite Markov chains satisfy CMLSI. Using the noncommutative curvature lower bound introduced in \cite{carlen2017gradient,datta2020relating}, CMLSI was obtained for heat semigroup on all compact Riemannian manifolds and some examples from operator algebras \cite{brannan2020complete,brannan2020complete2}.
Despite the constant progress on this topic in the recent years, the problem of the positivity of the CMLSI constant for finite dimensional QMS has been left open. Here, our Theorem \ref{thm:1} finally provides a positive answer to the question via a relatively simple proof.\\

{\bf \noindent Strong Data processing inequality.}
One key property behind
the widespread applications of the quantum relative entropy is the data processing inequality. It states that the relative entropy is non-increasing under the action of a quantum channel $\Phi$ (complete positive trace perserving map). Namely, for all states $\rho$ and $\sigma$,
\begin{align}\label{DPI}
    D(\Phi(\rho)\|\Phi(\sigma))\le\,D(\rho\|\sigma)\,.
\end{align}
As the relative entropy is a measure of distinguishability, the data processing inequality
asserts that two states can not become more distinguishable after applying a
same channel to them. First proved by Lindblad \cite{lindblad1975completely} and Uhlmann \cite{uhlmann1977relative}, the data processing inequality for the relative entropy has been largely
refined and improved in recent years (e.g.\cite{muller2017monotonicity,junge2018universal,carlen2020recovery}).
As discussed in \cite{lesniewski1999monotone,muller2016entropy,hiai2016contraction,berta2019quantum},
one natural direction is to ask when the contraction of relative entropy observed in \eqref{DPI} can be strict, i.e. there exists a constant $c<1$ such that
\begin{align}
\label{SDPI}
D(\Phi(\rho)\|\Phi(\si))\le c\,D(\rho\|\si)\,.
\end{align}
This question has been intensively studied for classical channels and more general entropies (see e.g. \cite{AG76,Dobrushin56,Dobrushin56II,Csiszar67,Liese06,polyanskiy2017strong,Raginsky16} and the references therein) under the name \emph{strong data processing inequality} (SDPI). In the quantum setting,
despite progresses on some special cases \cite{muller2016entropy,hiai2016contraction}, the existence of a contractive coefficient for general channels in \eqref{SDPI} remains open.
Our second main result is the following strong data processing inequality as a discrete time analog of Theorem \ref{thm:1}.

\begin{theorem}[(c.f. Corollary \ref{cor:SDPI})]\label{thm:2}
Let $\Phi:\cB(\cH)\to \cB(\cH)$ be a quantum channel. Suppose $\Phi$ is $\operatorname{GNS}$-symmetric to a full-rank invariant state $\si=\Phi(\si)$. Then there exists an explicit constant $c<1$ such that for any $n\in\mathbb{N}$ and all bipartite states $\rho\in\cD(\cH\otimes \mathbb{C}^n)$,
 \begin{align}\label{CSDPI1}\tag{CSDPI}
D((\Phi\otimes\id_n)(\rho)\|(\Phi\circ E_*\otimes \id_n)(\rho))\le c\,D(\rho\|(E_*\otimes\id_n)(\rho))\,,
\end{align}
where $E_*$ is the projection onto the decoherence-free space of $\Phi$.
\end{theorem}
We refer to Section \ref{sec:CSDPI} for the definition of $E_*$ and remark that the constant $c$ explicitly depends on the index $C_{\operatorname{cb}}(E_*)$ and an $L_2$-condition $\la:=\norm{\Phi-E_*:L_2\to L_2}{}$. The above inequality \eqref{CSDPI1} implies a discrete time entropy decay. Moreover, the inequality \eqref{CSDPI1} gives again a uniform control for all amplifications $\Phi\otimes\id_n$, which is the reason why we call it \emph{complete strong data processing inequality} (CSDPI).  These improvements over the standard data processing inequality have applications to quantum state preparation and quantum channel capacities \cite{bardet2019group,capel2020modified}. For instance, similar to CMLSI, CSDPI admits tensorization: if two quantum channels $\Phi$ and $\Psi$ satisfy CSDPI with contraction coefficient $c<1$, so does $\Phi\ten \Psi$. Also thanks to ``completeness'', Theorem \ref{thm:2} implies a concrete estimate on the convergence $\Phi^n \to \Phi^n\circ E_*$ in terms of the diamond norm.\\

{\bf \noindent Approximate tensorization of relative entropy.}
The data processing inequality is closely related to another celebrated inequality in quantum information theory, namely the \textit{strong subadditivity} (SSA). SSA can be equivalently stated in terms of relative entropies as follows: for any tripartite state $\rho^{ABC}$,
\[ D\Big(\rho^{ABC}\Big\| \frac{\Id_{AB}}{d_{AB}}\ten \rho^{C}\Big)\le D\Big(\rho^{ABC}\Big\| \frac{\Id_{A}}{d_{A}}\ten \rho^{BC}\Big)+D\Big(\rho^{ABC}\Big\| \frac{\Id_{B}}{d_{B}}\ten \rho^{AC}\Big)\pl.\]
Here $\frac{\Id_{AB}}{d_{AB}}$ is the completely mixed state on $AB$ whereas $\rho^C$ denotes the reduced density on $C$ (and similarly for the other terms). SSA was long known in classical information theory, and proved by Lieb and Ruskai \cite{lieb1973proof} for the quantum entropy. %using Lieb's concavity theorem \cite{lieb1973convex}).
Later Petz \cite{Petz91} proved SSA in a very general setting: given any four matrix subalgebras $\cN\subset \cN_1,\cN_2\subset \cM$, and corresponding projections $E_1,\,E_2,\,E_\cN$ from $\cM$ onto $\cN_1,\,\cN_2$ and $\cN$, for all states $\rho$ on $\cM$, the following inequality holds
\begin{align}\label{tensorization}
    D(\rho\|E_{\cN*}(\rho))\le \,D(\rho\|E_{1*}(\rho))+D(\rho\|E_{2*}(\rho))\,
\end{align}
as long as $E_1\circ E_2=E_2\circ E_1=E_\cN$. This last commutation relation is usually referred to as a ``\textit{commuting square}'' condition and was introduced by Popa \cite{popa83}.

Although the commuting square gives a nice characterization of SSA,
SSA-type inequalities are also desired when the ``commuting square'' condition is not fully satisfied.
For instance, in the context of classical lattice spin systems, where the projections are conditional expectations onto different regions of the lattice with respect to a given Gibbs measure, the commuting square condition corresponds to the infinite temperature regime \cite{bardet2020approximate}. To assess the finite temperature regime, \eqref{tensorization} has to be modified in the following way \cite{cesi2001quasi,dai2002entropy}: there exists a constant $c >1$ such that for all states $\rho$,
\begin{align}\label{AT}
    D(\rho\|E_{\cN*}(\rho))\le \,c\,\big(D(\rho\|E_{1*}(\rho))+D(\rho\|E_{2*}(\rho))\big)\,,
\end{align}
where the constant $c$ is some measure of the violation of the commutation relation $\|E_1\circ E_2-E_\cN\|$ in some appropriate norm. This inequality, called \textit{approximate tensorization of the relative entropy}, was used in the classical case (i.e. when all algebras
are commutative) in the study of logarithmic Sobolev inequalities for lattice spin system \cite{cesi2001quasi}. In the quantum setting,
 a weaker bound to \eqref{AT} was derived in \cite{bardet2020approximate} with a further additive error term vanishing on classical states. However, the question of finding general bounds like \eqref{AT} without additive error term was left unresolved. Our third main theorem answers this question.
\begin{theorem}
\label{thm:3}
Let  $\cN\subset \cN_1,\cN_2\subset \cM$ be four finite dimensional von Neumann algebras. Let  $E_1,\,E_2,\,E_\cN$ be the corresponding projections from $\cM$ onto $\cN_1,\,\cN_2$ and $\cN$ such that $E_\cN\circ E_{1}=E_\cN\circ E_{2}=E_\cN$. Then there exists an explicit constant $c_{\operatorname{cb}}$ such that
any $n\in\mathbb{N}$ and all states $\rho\in \cM\otimes \cB(\mathbb{C}^n)$, we have
\begin{align}\label{completeATq1}
    D(\rho\|(E_{\cN*}\otimes\id)(\rho))\le \,c_{\operatorname{cb}}\big(D(\rho\|(E_{1*}\otimes \id)(\rho))+D(\rho\|(E_{2*}\otimes \id)(\rho))\big).
\end{align}
\end{theorem}
We refer to Theorems \ref{thm:AT} \& \ref{theo:lambdaindexcontrolAT} and Corollary  \ref{Coro:sharperAT} for concrete estimates on the constant $c_{\operatorname{cb}}$. All of the three results above rely on a common conceptually simple tool, namely a two-sided estimate of the relative entropy via the so-called \textit{Bogoliubov-Kubo-Mori Fisher information} (see Lemma \ref{lemma:keylemma} in Section \ref{sec:keylemma}  for more details).
The Bogoliubov-Kubo-Mori Fisher information is closely related to a special case of monotone Riemannian metric on state space studied in \cite{petz1996monotone,lesniewski1999monotone} and quantum $\chi_2$-divergence in \cite{wolf2012quantum}.
It allows us to approach each of the three above entropy inequalities via corresponding spectral gap conditions. Given the simplicity of our approach, we believe it will also prove useful in the study of other entropic inequalities. \\

{\bf \noindent Applications and Examples.}
In the second part of this article, namely \Cref{sec:symmetric,sec:localsemigroups}, we exploit the  approximate tensorization estimate from Theorem \ref{thm:3} to get tighter bounds on the optimal CMLSI constant for quantum Markov semigroups (QMS) relevant to the communities of mathematical physics and quantum information theory. For a QMS $(\cP_t=e^{t\cL})_{t\ge 0}$ with the generator $\cL$,
we denote by $\alpha_{\operatorname{CMLSI}}(\cL)$ the largest constant $\alpha$ satsisfying \eqref{MLSI1} in Theorem \ref{thm:1}. In \Cref{sec:symmetric}, we restrict our analysis to the class of symmetric QMS, that is QMS symmetric to the trace inner product or equivalently the maximally mixed state. The generators of these semigroups admit a simple form as a sum of double commutators with self-adjoint operators $\{a_k\}$:
\begin{align}
    \cL(\rho)=-\sum_{k=1}^l[a_k,[a_k,\rho]]\,.
\end{align}
Using approximate tensorization, we obtain the following improved CMLSI constant for
symmetric QMS.
\begin{theorem}[(c.f. Corollary \ref{cor:symmetric2})]\label{thm:4}
For a symmetric generator $\cL$ given as above,
\[\alpha_{\operatorname{CMLSI}}(\cL)\ge \Omega\big(\la m^{-2}\operatorname{polylog}(d_{\cH})^{-1} \big)\]
where $d_\cH$ is the dimension of the underlying Hilbert space, $m$ denotes the maximal number of $a_k$ that do not commute with any single one of them, and $\la:=\min_k \la(\cL_{a_k})$ is the minimum spectral gap of any of the generators $\cL_{a_k}(\rho)=[a_k,[a_k,\rho]]$.
\end{theorem}
Note that the above bound is asymptotically better than Theorem \ref{thm:1} because the index is $C_{\operatorname{cb}}(E_*)=d^2$ for primitive semigroups.
\begin{example}{
Consider the quantum Markov semigroups induced by sub-Laplacians of the special unitary group $\operatorname{SU}(2)$ on its irreducible representations: \begin{align*}
    \mathcal{L}^{H}_m(\rho):=-[X_m,[X_m,\rho]]-[Y_m,[Y_m,\rho]]
    \end{align*}
where $X_m$ (resp. $Y_m$) is the spin-$\frac{m-1}{2}$ representation of the Pauli $X$ matrix (resp. $Y$-matrix). In contrast to the induced semigroup of the standard Laplace-Beltrami operator $\Delta=X^2+Y^2+Z^2$ the CMLSI constant of $\cL^H_m$ is not accessible from the corresponding classical Markov semigroup due to the lack of curvature lower bound in the sub-Riemannian setting. With help of numerics, we obtain that $$\al_{\operatorname{CMLSI}}(\mathcal{L}^{H}_m)>0.18$$ uniformly for all $m\ge 2$. %Here the approximate tensorization we use is a new entropic uncertainty relations in the presence of quantum memory as an application of our Theorem \ref{thm:3}.
We note that the existence of such dimension independent CMLSI constant for general quantum Markov semigroups induced by sub-Lalpacian were independently obtained by the first author, Junge and Li \cite{GJL21} using a completely different method.}
\end{example}

In \Cref{sec:localsemigroups}, we focus on symmetric semigroups which bare a locality structure inherited from a graph. More precisely, given a finite graph $G=(V,E)$, we consider the $n$-fold tensor product $\cH_V:=\bigotimes_{v\in V}\cH_v$ of a finite dimensional local Hilbert space $\cH$, namely, a $n$-qudit system for $d=\dim(\cH)$. The Lindblad operators are supported on the edges $e\in E$ of the graph:
\begin{align}\label{eq:subsystem}
    \mathcal{L}_{G}:=\sum_{e\in E}\,\mathcal{L}_e\,,\qquad\text{ where }\qquad \mathcal{L}_{e}(\rho):=\sum_{j\in J^{(e)}} L^{(e)}_j \rho L^{(e)}_j -\frac{1}{2}\{ L^{(e)}_j L^{(e)}_j,\,\rho\}\,,
\end{align}
where for any edge $e\in (v,w)\in E$ and any $j\in J^{(e)}$, the local Lindblad operator $L^{(e)}_j$ acts trivially on subsystems other than $\cH_{v}\otimes \cH_{w}$. We call \eqref{eq:subsystem} a \textit{subsystem Lindbladian}, which means that the global dynamics consists of local interaction on subsystems of adjacent vertices. This gives a general model of 2-local interacting quantum lattice spin systems.
Using approximate tensorization again, we provide a lower bounds on the CMLSI constant for the global Lindbladian $\cL_G$ based on the local Lindbladians $\cL_e$.
\begin{theorem}[(c.f. \Cref{thm:graphs})]\label{thm:5}
Let $G=(V,E)$ be a finite, connected graph of maximum degree $\gamma$ and let $\mathcal{L}_G$ be a symmetric subsystem Lindbladian of the form \eqref{eq:subsystem}. Denote by $E_e$ the projection onto the kernel of the local Lindbladian $\cL_e$. Then
\begin{align*}
\al_{\operatorname{CMLSI}}(\cL_e)\ge \Omega\left(\frac{\ln\big(\frac{\lambda(\widetilde{\cL}_G)}{4(\gamma-1)^2}+1\big)}{\ln(C)+1}\right) \min_{e\in E}\al_{\operatorname{CMLSI}}(\cL_e) \,
 \end{align*}
where $\al_{\operatorname{CMLSI}}(\cL_e)$ is the $\operatorname{CMLSI}$ constant of $\cL_e$, and
$\lambda(\widetilde{\cL}_G)$ is the spectral gap of the generator $\widetilde{\cL}_G:=\sum_{e\in E}E_e-\id$.
\end{theorem}
Here the index $C$ can be chosen as either the complete Pimsner-Popa index \cite{popapimser} of the algebra $\cN$ of fixed points of the evolution, or the inverse  minimal eigenvalue of the Choi state of the projection map $\displaystyle E_G:=\lim_{t\to \infty}e^{t\cL_G}$. The index $C$ can be thought of as what replaces the size of the graph in the case of classical graph Laplacians. In particular, for expander graphs, our bound gives
 \begin{align*}
     \alpha_{\operatorname{CMLSI}}(\widetilde{\cL}_{\operatorname{G}})\ge \Omega\left(\frac{1}{\ln(C)} \right)\,.
 \end{align*}

%In \Cref{thm:graphs}, we provide a general strategy in order to derive lower bounds on the CMLSI constant of a subsystem Lindbladian. In the case of expander graphs, or bound reduces to
 %\begin{align*}
  %   \alpha_{\operatorname{CMLSI}}=\Omega\Big(\frac{1}{\ln(C)} \Big)\,,
 %\end{align*}
%where $C$ can be chosen as either the so-called completely bounded Pimsner-Popa index \cite{popapimser} of the algebra $\cN$ of fixed points of the evolution, or the inverse minimal eigenvalue $\mu_{\min}(J_\cN)^{-1}$ of the Choi state of $E_G:=\lim_{t\to \infty}e^{t\cL_G}$:
%\begin{align*}
 %   &C:=\min\{C_{\operatorname{cb}}(\cB(\cH_V):\cN),\mu_{\min}(J_{E_G})^{-1}\}\,,\\ &C_{\operatorname{cb}}(\cB(\cH_V):\cN):=\inf\{c>0:\,\rho\le c\,E_{G}(\rho)\,\forall\rho\in\cB(\cH_V\otimes \mathbb{C}^m),\,m\in\mathbb{N}\}\,.
%\end{align*}
%The index $C$ can be thought of as what replaces the size of the graph in the case of classical graph Laplacians on expander graphs.
We exemplify our bound on three classes of subsystem Lindbladians.
\begin{example}[(Random transposition)]{
Motivated by the classical random transposition model in \cite{bobkov2003modified,gao2003exponential},
 we introduce in \Cref{sec:randomperturb} the quantum nearest neighbor random transposition. More precisely, we consider the local Lindbladian on an edge $(i,j)\in E$ given by
\begin{align}
  \mathcal{L}_{(i,j)}(\rho):=\frac{1}{2}(S_{i,j}\rho S_{i,j}-\rho )\,, \,\, S_{i,j}(|\psi\rangle \otimes |\varphi\rangle)=|\varphi\rangle\otimes |\psi\rangle
\end{align}
where $S_{i,j}:{\cH_{i}\ten \cH_{j}}\to \cH_{i}\ten \cH_{j}$ is the swap unitary gate between vertex $i$ and $j$. Then the global
Lindbladian $\cL_G^{\operatorname{NNRT}}:= \sum_{e\in E}\cL_{e}$ is generated by local random swaps on $|V|=n$ qudits.
 In this case, we find %\textcolor{red}{not really, gap unknown...}
\begin{align*}
    \alpha_{\operatorname{CMLSI}}(\cL_G^{\operatorname{NNRT}})\ge \la(\cL_G^{\operatorname{NNRT}}) \,\Omega((\ln n!)^{-1})\,,
\end{align*}
where $\la(\cL_G^{\operatorname{NNRT}})$ is the spectral gap and the factorial $n!$ is the size of the permutation group $\mathcal{S}_n$. This presents an exponential improvement over the bounds from Theorem \ref{thm:1}, where the constant was controlled by the inverse size of the group $(n!)^{-1}$. }\end{example}

\begin{example}
[(Approximate unitary design)]{ Another class of examples we are concerned with are continuous time approximate $k$-designs previously studied in the literature  \cite{brandao2016local,onorati2017mixing}. These are quantum Markov semigroups that locally converge to the Haar unitary $k$-design over the unitary group $\operatorname{U}(d^2)$.
Namely, for each vertex $i\in V$, $\cH_i=(\mathbb{C}^d)^{\ten k}$, we consider the local Lindbladian given by
\begin{align}
\cL_e^{(k)}(\rho)=\mathcal{D}^{(k)}_{\operatorname{Haar}}(\rho)-\rho\ , \ \ \mathcal{D}^{(k)}_{\operatorname{Haar}}(\rho):=\int_{\operatorname{U}(d^2)} U^{\otimes k}\rho (U^\dagger)^{\otimes k}\ d\mu_{\operatorname{Haar}}(U)
\end{align}
where $\mu_{\operatorname{Haar}}$ is the Haar measure.  Previous works \cite{brandao2016local,onorati2017mixing,hunter2019unitary} studied the spectral gap of  $\cL_G^{(k)}:=\sum_{e} \cL^{(k)}_{e}$ for a linear graph which, combined with an equivalence of norms, gave a convergence time of order $\mathcal{\widetilde{O}}(n^2k^{6+3.1/\ln(d)})$ as measured in diamond norm. Here in $\widetilde{O}$ we hide further dependence on the local dimension $d$ as well as sublinear factors.
%In \cite{nakata2017efficient}, another random quantum circuit-like construction of a time-dependent Hamiltonian with varying couplings over discrete time steps was constructed, which forms $k$-designs in a depth $\mathcal{O}(n^2k)$ up to $k = o(\sqrt{ n})$.
In \Cref{sec:designs}, combining Theorem \ref{thm:5} with Pinsker's inequality, we find that
\begin{align}\label{mixingtime}
& \al_{\operatorname{CMLSI}}(\cL_G^{(k)})\ge \mathcal{\widetilde{O}}(nk^{6+3.1/\ln(d)})
\\
   & t^{\operatorname{linear}}_\eps:=\min\{t\ge 0:\,\|  e^{\frac{t}{n}\cL_G^{(k)}}(\rho)- E_*\|_\diamond\le \eps\}=\,\mathcal{\widetilde{O}}(n^2k^{6+3.1/\ln(d)})\,
\end{align}
for any moment $k$ and any local dimension $d$. Note that we renormalized the time parameter in \eqref{mixingtime} in order to compare our bound to the ones found for discrete time random circuits, since $\mathcal{O}(n)$ gates per time unit are effectively being implemented in continuous-time. This result also extends to
other physically motivated generators, whose local Landbladian satifies
\begin{align}\label{eq:asymptoticsdesignsintro}
    e^{t\cL_e^{(k)}}(\rho)\underset{t\to \infty}{\to} \mathcal{D}^{(k)}_{\operatorname{Haar}}(\rho)
\end{align}
We recall that a universal lower bound of $\widetilde{\Omega}(nk)$ (up to logarithmic factors) was found in (see \cite[Proposition 8]{brandao2016local}) for any $\eps$-approximate $k$-design.
}
\end{example}
\begin{example}[(Quantum Kac model)]{
Finally, in \Cref{sec:quantumkac}, we consider the recently introduced quantum extensions of the Kac generator which models the evolution of the  velocity distributions of $n$ particles undergoing elastic collisions \cite{carlen2019chaos}. More precisely, for each vertex $i\in V$, $\cH_i=\mathbb{C}^d$ and the local Lindbladian is
\[\cL_{e}(\rho)=\Phi_\mu(\rho)-\rho\ ,\ \Phi_\mu(\rho)=\int_{\operatorname{U}(d^2)} U \rho U^\dag d\mu(U)\ ,\]
where $\mu$ is some probability measure specifying the collision model. The global semigroup can be understood as a continuous time approximate $1$-design over the complete graph $K_n$. We obtain the following bound on the CMLSI constant for the quantum Kac model
\begin{align*}
\al_{\operatorname{CMLSI}}(\cL_{K_n})\ge  \mathcal{\widetilde{O}}\left(\frac{1}{\ln d}\right)\,  \al_{\operatorname{CMLSI}}(\cL_{e})\,.
\end{align*}
For more general $k$-designs as in the previous example, it was argued in \cite{onorati2017mixing} that the spectral gap of subsystem Lindbladian for the complete graph is always larger than that for nearest neighbour graphs, so that the bound $\la(\cL_{K_n}^{(k)})\ge \mathcal{\widetilde{O}}(n^2k^{6+3.1/\ln(d)})$ mentioned above still holds. Here, we refine their argument in two directions: first, we estimate the CMLSI constant instead of the spectral gap. Second, using graph theoretic arguments we improve the constant by a factor $n$. This results in an improvement of the mixing time in \eqref{mixingtime}
\begin{align*}
&\al_{\operatorname{CMLSI}}(\cL_{K_n}^{(k)})\ge  \mathcal{\widetilde{O}}\ (k^{6+3.1/\ln(d)}),\\
&t^{K_n}_\eps:=\min\{t\ge 0:\,\|  e^{\frac{ t}{n}\cL_{K_n}^{(k)}}(\rho)- E_*\|_\diamond\le \eps\}=\,\mathcal{\widetilde{O}}(nk^{6+3.1/\ln(d)})\,.
\end{align*}
We note that the $n$-independence of our CMLSI bound matches that of the classical setting proved by Villani in \cite{villani2003cercignani} using similar tensorization techniques. }\end{example}

The rest of the paper is organized as follows. In the
next section, we review some preliminary definitions and prove our key lemma. Section 3 is devoted to the proof of Theorem \ref{thm:1}, which is our first main result on the complete modified log-Sobolev inequality. In Section 4, we prove the complete strong data processing inequality of \Cref{thm:2}. The approximate tensorization results are discussed in Section 5. Section 6 provides the improved CMSLI constant of \Cref{thm:4} for symmetric quantum Markov semigroup. In Section 7, we discuss examples from subsystem Lindbladians. We end the paper with some discussion on questions that remain opens. We remark that although we restrict our discussion to finite dimensions, the general results in \cref{sec:CMLSI}, \cref{sec:CSDPI}, and \cref{sec:AT} can be extended to (trace) symmetric maps in the setting of finite von Neumann algebras, as long as the index $C_{\operatorname{cb}}(E_*)$ and corresponding spectral gap condition are satisfied.

\medskip

{\bf \noindent Acknowledgements.}
CR is supported by a Junior Researcher START Fellowship from the MCQST. CR is grateful to Daniel Stilck Fran\c{c}a, Angela Capel and Ivan Bardet for stimulating discussions. CR and LG particularly thank Daniel Stilck Fran\c{c}a for very useful comments on a preliminary version of the paper. LG thanks Marius Junge and Haojian Li for helpful discussions.

\section{Preliminaries}\label{sec:keylemma}
\subsection{Relative Entropy and Conditional expectation}
Throughout the paper, we will consider $\cH$ to be a finite dimensional Hilbert space, $\cB(\cH)$ to be the bounded operators, and $\cM\subset \cB(\cH)$ to be a von Neumann subalgebra. We write "$\tr$" for the standard matrix trace, $\lan \cdot,\cdot \ran_{\operatorname{HS}}$ for the trace inner product and $\norm{\cdot}{2}$ for the Hilbert-Schmidt norm. The corresponding Hilbert-Schmidt space (resp. trace class operators) is denoted by $\cT_2(\cH)$ (resp. $\cT_1(\cH)$). Operators will be denoted by capital letters, and sometimes also by lowercase letters, in order to emphasize their belonging to a subalgebra. We write $A^\dagger$ for the adjoint of an operator $A\in \cB(\cH)$, and  $\Phi^*$ (or $\Phi_*$) for the adjoint (or preadjoint) of a map $\Phi:\cB(\cH)\to \cB(\cH)$. The identity operator on $\cH$ is denoted as $\Id_{\cH}$ and the identity map on a von Neumann subalgebra $\cM\subseteq \cB(\cH)$ is $\id_{\cM}$. We also denote the dimension of $\cH$ by $d_\cH=\text{dim}(\cH)$. Given two maps $\Phi,\Psi:\cM\to \cM$ on a von Neumann subalgebra $\cM\subseteq \cB(\cH)$, we say that $\Phi\le_{\operatorname{cp}} \Psi$ if $\Psi-\Phi$ is completely positive.

We say that an operator $\rho$ is a state (or density operator) if $\rho\ge 0$ and $\tr(\rho)=1$. We denote by $\cD(\cH)$ the set of states on $\cH$. A quantum channel $\Phi: \cT_1(\cH)\to \cT_1(\cH)$ (or more generally, $\Phi:\cM_*\to \cM_*$) is a completely positive trace preserving map. With slight abuse of notation, we will often write $\Psi(\rho):=(\Psi\otimes  \id) (\rho)$ for a bipartite state $\rho\in\cD(\cH\otimes \mathbb{C}^n)$ and a quantum channel $\Psi:\cT_1(\cH)\to \cT_1(\cH)$. For two states $\rho$ and $\si$, their relative entropy is defined as \begin{align*}D(\rho\|\si)=\begin{cases}
                             \tr(\rho\ln\rho-\rho\ln\si), & \mbox{if } \supp(\rho)\subseteq \supp(\si) \\
                             +\infty, & \mbox{otherwise},
                           \end{cases}
\end{align*}
where $\supp(\rho)$ (resp. $\supp(\si)$) is the support projection of $\rho$ (resp. $\si$).

Let $\cN\subseteq \cM\subseteq  \cB(\cH)$ be two von Neumann subalgebras. Recall that a conditional expectation onto $\cN$ is a completely positive unital map $E_\cN:\cM\to \cN$ satisfying \begin{enumerate}
\item[i)]for all $a \in \cN$, $E_\cN(a)=a$
\item[ii)]for all  $a,b\in\cN,X\in \cB(\cH)$, $E_\cN(aXb)=aE_\cN(X)b$.
 \end{enumerate}
We denote by $E_{\cN*}$ its adjoint map with respect to the trace inner product, i.e.
\[\tr(E_{\cN*}(X)Y)=\tr(XE_\cN(Y))\pl. \]
For a state $\rho$, the relative entropy with respect to $\cN$ is defined as follows
\[D(\rho\|\mathcal{N}):=D(\rho\|E_{\cN*}(\rho))=\inf_{E_{\cN*}(\si)=\si} D(\rho\|\si)\pl,\]
where the infimum is always attained by $E_{\cN*}(\rho)$. Indeed, for any $\si$ satisfying $E_{\cN*}(\si)=\si$,  we have the chain rule (see \cite[Lemma 3.4]{junge2019stability})
\begin{align}\label{chain rule}D(\rho\|\si)=D(\rho\|E_{\cN*}(\rho))+D(E_{\cN*}(\rho)\|\si)\pl.\end{align}
Hence the infimum is attained if and only if $D(E_{\cN*}(\rho)\|\si)=0$. More explicitly, a finite dimensional von Neumann (sub)algebra can always be expressed as a direct sum of matrix algebras with multiplicity, i.e.
\[\cN=\bigoplus_{i=1}^n \cB(\cH_i)\ten \mathbb{C}\Id_{\cK_i}\pl, ~~~~~~~ \cH=\bigoplus_{i=1}^n \cH_i\ten \cK_i\pl.\]
Denote $P_i$ as the projection onto $\cH_i\ten \cK_i$. There exists a family of density operators $\tau_i\in \cD(\cK_i)$ such that
\begin{align}E_\cN(X)=\bigoplus_{i=1}^n \tr_{\cK_i}(P_iXP_i(\Id_{\cK_i}\ten \tau_i))\ten \Id_{\cK_i}\pl, ~~~~E_{\cN*}(\rho)=\bigoplus_{i=1}^n \tr_{\cK_i}(P_i\rho P_i)\ten \tau_i\pl, \label{cd}\end{align}
where $\tr_{\cK_i}$ is the partial trace with respect to $\cK_i$.
A state $\si$ satisfies $E_{\cN*}(\si)=\si$ if and only if
\[ \si=\bigoplus_{i=1}^n p_i\,\si_i\ten \tau_i\pl\]
for some density operators $\si_i\in \cD(\cH_i)$ and a probability distribution $\{p_i\}_{i=1}^n$. We denote $\cD(E_\cN):=\{\si\in \cD(\cH) | \si=E_{\cN*}(\si)\}$ as the subset of states that are invariant under $E_{\cN*}$. For any $\si\in \cD(E_\cN)$ and all $X\in\cM$,
\begin{align*}
E_{\cN*}(\si^{\frac12}X\si^{\frac12})= \si^{\frac12}E_{\cN}(X)\si^{\frac12}\pl. \end{align*}

%As already seen in \cite{bardet,CLSI,junge2019stability,bardet2018hypercontractivity}, the relative entropy $D(\rho\|E_{\cN*}(\rho))$ is crucial in functional inequalities for non-ergodic Markov semigroups.

\subsection{Subalgebra index and Max-relative entropy}
Let $\cM\subset \cB(\cH)$ be a finite dimensional von Neumann algebra and let $\cN\subset \cM$ be a subalgebra of $\cM$. The trace preserving conditional expectation $E_{\cN,\tr}:\cM\to\cN$ is defined so that for any $ X\in \cM$ and $ Y\in\cN$,
\[ \tr(XY)=\tr(E_{\cN,\tr}(X)Y)\pl.\]
$E_{\cN,\tr}$ is self-adjoint and corresponds to taking $\displaystyle \tau_i=d_{\cK_i}^{-1}\Id_{\cK_i}$ in \eqref{cd}. We recall the definition of the index associated to the algebra inclusion $\cN\subset\cM$,
\begin{align*}&C(\cM:\cN)=\inf\{ c>0\pl | \pl \rho\le  c\,E_{\cN,\tr}(\rho) \text{ for all states $\rho\in \M$} \}\pl, \\ &C_{\operatorname{cb}}(\cM:\cN)=\sup_{n\in\mathbb{N}}C(\cM\ten \mathbb{M}_n:\cN\ten \mathbb{M}_n)\pl,\end{align*}
where the supremum in $C_{\operatorname{cb}}(\cM:\cN)$ is taken over all finite dimensional matrix algebras $\mathbb{M}_n$.
The index $C(\cM:\cN)$ was first introduced by Pimsner and Popa in \cite{popapimser} for the connection to
subfactor index and Connes entropy, and the completely bounded version $C_{\operatorname{cb}}(\cM:\cN)$ was studied in \cite{gaoindex}.
These indices are closely related to the notion of maximal relative entropy. Recall that for two states $\rho,\omega$, their maximal relative entropy is \cite{datta2009min} \[D_{\max}(\rho\|\omega)=\ln \inf\{\pl c>0 \pl   | \pl \rho\le c \,\omega \pl \}\pl.\]
Indeed, $$\displaystyle\ln C(\cM:\cN)=\sup_{\rho\in \cD(E_{\cM,\tr})} D_{\max}(\rho\|E_{\cN,\tr}(\rho))\,.$$ For all finite dimensional inclusion $\cN\subset\cM$, the index $C(\cM:\cN)$ is explicitly calculated in \cite[Theorem 6.1]{popapimser} (hence also for $C_{\operatorname{cb}}(\cM:\cN)$). In particular, for $\cM=\cB(\cH)$ and $\cN=\bigoplus_{i=1}^n \cB(\cH_i)\ten \mathbb{C}\Id_{\cK_i}$,
\begin{align}
    C(\cB(\cH):\cN)=\sum_{i=1}^n \min\{d_{\cH_i},d_{\cK_i}\}\,d_{\cK_i}\,,~~~~~~ C_{\operatorname{cb}}(\cB(\cH):\cN)=\sum_{i=1}^n d_{\cK_i}^2\pl. \label{formula}\end{align}
For example, if we take $\cD\subset \cB(\cH)$ to be the subalgebra of diagonal matrices and $\mathbb{C}$ as the multiple of identity
\begin{align}&C(\cB(\cH):\cD)=C_{\operatorname{cb}}(\cB(\cH):\cD)=d_\cH\pl, \nonumber\\ &C(\cB(\cH):\mathbb{C})=d_\cH\,,~~~ C_{\operatorname{cb}}(\cB(\cH):\mathbb{C})=d_\cH^2\pl.\label{eq:indicesex}
\end{align}
In this paper, we will also consider the index for a general conditional expectation $E_{\cN}:\cM\to\cN$ (see e.g \cite{kosaki1998type} for more information). For a conditional expectation $E_{\cN}:\cM\to \cN$ onto $\cN$, we define
\begin{align}\label{eq:indexE}&C(E_{\cN})=\inf\{ c>0\pl | \pl \rho\le  c\,E_{\cN*}(\rho) \text{ for all states $\rho\in \M$} \}\pl, \\ &C_{\operatorname{cb}}(E_{\cN})=\sup_{n\in\mathbb{N}}C(E_{\cN}\ten \id_{\mathbb{M}_n})\pl \nonumber .\end{align}
Here, we recall that $\mathbb{M}_n$ is the $n$-dimensional matrix algebra and $E_{\cN}\ten \id_{\mathbb{M}_n}$ is a conditional expectation from $\cM\ten \mathbb{M}_n\to\cN\ten \mathbb{M}_n$.
Note that given the subalgebra $\cN$, $E_{\cN}$ and $E_{\cN*}$ are uniquely determined by any invariant state $\si\in \cD(E_{\cN})$, or equivalently the densities $\{\tau_i\}$ in \eqref{cd}.
Indeed, denoting \begin{align}\tau=\bigoplus_{i=1}^n \Id_{\cH_i}\ten \tau_i \,,\label{tau}\end{align}We have
\begin{align}\label{eqENENtr}
E_{\cN}(X)= E_{\cN,\tr}(\tau^{\frac12}X\tau^{\frac12})\,,\quad\quad E_{\cN*}(\rho)= \tau^{\frac12}E_{\cN,\tr}(\rho)\tau^{\frac12}.
\end{align}
In particular, $E_{\cN}$ is faithful if and only if $\tau$ is full-rank. By definition, the Pimsner-Popa index $C(\cM:\cN)$ is the special case for the trace perserving condition expectation $C(E_{\cN,\tr})$.
In the later discussion, we will often use the alternative notation
\[C_\tau(\cM:\cN):=C(E_{\cN})\pl,\quad\quad \pl C_{\tau,\operatorname{cb}}(\cM:\cN):=C_{\operatorname{cb}}(E_{\cN})\,. \]
Since $\tau$ commutes with $\cN$,
\begin{align}\label{fromtracialtonontracial} C_\tau(\cM:\cN)\le \mu_{\operatorname{min}}(\tau)^{-1}C(\cM:\cN)\,,\quad \quad C_{\tau,\operatorname{cb}}(\cM:\cN)\le \mu_{\operatorname{min}}(\tau)^{-1}C_{\operatorname{cb}}(\cM:\cN)\end{align}
where $\mu_{\operatorname{min}}(\tau)=\min_{i}\mu_{\min}(\tau_i)$ is the minimal eigenvalue of $\tau$.
Hence in finite dimensions, both $C(E_{\cN})$ and $C_{\operatorname{cb}}(E_{\cN})$ are finite if and only if $E_{\cN}$ is faithful. Moreover, for any invariant state $\sigma\in\cD(E_\cN)$, by the obvious bound $\sigma\le \tau$, we also have
\begin{align} C_\tau(\cM:\cN)\le \mu_{\operatorname{min}}(\sigma)^{-1}C(\cM:\cN)\,,\quad\quad C_{\tau,\operatorname{cb}}(\cM:\cN)\le \mu_{\operatorname{min}}(\sigma)^{-1}C_{\operatorname{cb}}(\cM:\cN)\,.\end{align}
\subsection{A key lemma}\label{subsec:key}
We shall now discuss the key lemma that will be repeatedly used in the later sections. Given a density operator $\rho\in \cD(\cH)$, we define the multiplication operator
\[\Gamma_{\rho}(X):=\int_{0}^1 \rho^{s} \,X\,\rho^{1-s} ds \pl.\]
$\Gamma_\rho$ is a positive operator on the Hilbert-Schmidt space $\cT_2(\cH):=L_2(\cB(\cH),\tr)$ and hence induces a weighted $L_2$-norm (semi-norm if $\rho$ is not full-rank) defined for $X\in \cB(\cH)$ as
\[ \norm{X}{\rho}^2:=\lan X,\Gamma_{\rho}(X)\ran_{\small{\operatorname{HS}}}= \int_{0}^1 \tr(X^\dagger\rho^{s} X\rho^{1-s}) \, ds\pl.\]
We denote by $L_2(\rho)$ the corresponding $L_2$-space. For a full-rank density $\rho$, the inverse operator of $\Gamma_\rho$ is given by
\[\Gamma_{\rho}^{-1}(X):=\int_{0}^\infty (\rho+r)^{-1} X(\rho+r)^{-1} \,dr\pl,\]
which is the double operator integral for the function $f(t)=\ln t$ and operator $\rho$ (see e.g. \cite{carlen2017gradient}). We denote by slight abuse of notations the corresponding weighted $L_2$-norm as
\[\norm{X}{\rho^{-1}}^2:=\lan X,\Gamma_{\rho}^{-1} (X)\ran_{\small{\operatorname{HS}}}= \int_{0}^\infty \tr(X^\dagger(\rho+r)^{-1} X(\rho+r)^{-1}) dr\pl. \]
and the corresponding $L_2$ space as $L_2(\rho^{-1})$. This is a special case of the quantum $\chi^2$-divergence introduced in \cite[Defnition 1]{TKRWV} for the logarithmic function.
It is easy to see that
\[ \norm{\Gamma_{\rho}(X)}{\rho^{-1}}=\norm{X}{\rho}\,,\quad\quad \norm{\Gamma_{\rho}^{-1}(X)}{\rho}=\norm{X}{\rho^{-1}}\,.\]

\begin{lemma}\label{lemma:compare}
If $\rho\le c\,\si$ for any two states $\rho,\sigma$ and some $c>0$, then for any $X\in\cB(\cH)$ and all $\mu_1,\mu_2>0$,
\begin{align*}
\int_{0}^\infty \operatorname{tr}(X^\dagger(\mu_1\,\sigma+r)^{-1} X(\mu_2\sigma+r)^{-1})\, dr\le c\,\int_{0}^\infty \operatorname{tr}(X^\dagger(\mu_1 \rho+r)^{-1} X(\mu_2\rho+r)^{-1})\, dr\, .
\end{align*}
In particular, $\norm{X}{\si^{-1}}\le c\norm{X}{\rho^{-1}}$.
\end{lemma}
\begin{proof}
This is a standard comparison. Using cyclicity of the trace and the fact that $t\mapsto  t^{-1}$ is operator anti-monotone,
\begin{align*}\int_{0}^\infty \tr(X^\dagger(\mu_1\rho+r)^{-1} X(\mu_2\rho+r)^{-1}) dr
%\\ =\int_{0}^\infty \tr((\rho+r)^{-1/2}X^\dagger(\rho+r)^{-1} X(\rho+r)^{-1/2}) dr
&\ge\int_{0}^\infty \tr(X^\dagger(c\mu_1\si+r)^{-1} X(\mu_2\rho+r)^{-1}) dr
%\\ =\int_{0}^\infty \tr((c\si+r)^{-1/2} X(\rho+r)^{-1}X^\dagger(c\si+r)^{-1/2}) dr
\\ &\ge\int_{0}^\infty \tr(X^\dagger(c\mu_1 \si+r)^{-1}X(c\mu_2 \si+r)^{-1}) dr
\\ &=\int_{0}^\infty \frac{1}{c^2}\,\tr( X^\dagger(\mu_1 \si+\frac{r}{c})^{-1}X(\mu_2 \si+\frac{r}{c})^{-1}) dr
\\ &=\frac{1}{c}\int_{0}^\infty \,\tr( X^\dagger(\mu_1 \si+r)^{-1}X(\mu_2 \si+r)^{-1}) dr\,.
\end{align*}
In the last equality, we used the change of variable $r\to \frac{r}{c}$.
\end{proof}
Our key lemma is a two-sided estimate of $D(\rho\|\si)$ via the inverse weighted norm.
\begin{lemma}\label{lemma:keylemma}Let $\rho$ and $\si$ be two full-rank density operators and suppose $\rho\le c\,\si$ for some $c>0$.
Then
\begin{align}k(c)\norm{\rho-\si}{\si^{-1}}^2\,\le D(\rho\| \si)\,\le\, \norm{\rho-\si}{\si^{-1}}^2\label{keyinequality}\end{align}
 where $\displaystyle k(c)=\frac{c\ln c-c+1}{(c-1)^2}$. Note that $k(c)\le 1/2$ for $c\ge 1$.
\end{lemma}
\begin{proof} For the lower bound, we consider $\rho_t:=(1-t)\si+t\rho, t\in [0,1]$ and the function $f(t)=D(\rho_t\|\si)$. We have $f(0)=0$, $f(1)=D(\rho\|\si)$ and the derivatives
\begin{align*}
&f'(t)=\tr((\rho-\si)\ln \rho_t-(\rho-\si)\ln\si)\pl, \\  &f''(t)=\int_{0}^\infty \tr\Big((\rho-\si)\frac{1}{\rho_t+r}(\rho-\si)\frac{1}{\rho_t+r}\Big)\,dr=\norm{\rho-\si}{\rho_t^{-1}}^2\pl.
\end{align*}
Note that $f'(0)= 0$ and $\rho_t\le (ct+(1-t))\si$. We have for the lower bound
\begin{align*}
 D(\rho\|\si)=&\int_0^1\Big(\int_0^s f''(t)dt\Big)ds\\
 = &\int_0^1\int_0^s \norm{\rho-\si}{\rho_t^{-1}}^2\,dt ds\\
 \ge &\int_0^1\int_0^s \frac{1}{1+(c-1)t}\,dt ds \norm{\rho-\si}{\si^{-1}}^2\\
 \ge & \ k(c) \norm{\rho-\si}{\si^{-1}}^2\pl,
\end{align*}
where we used Lemma \ref{lemma:compare} and
\[k(c)=\int_0^1\int_0^s \frac{1}{1+(c-1)t}\,dtds=\frac{c\ln c-c+1}{(c-1)^2}\pl.\]
The upper bound is a special case of \cite[Proposition 6]{TKRWV}. Here we present a different proof using a method similar to our lower bound. Note that $\rho_t=(1-t)\si+t\rho\ge (1-t)\si$. Then,
\begin{align*}
 D(\rho\|\si)
= &\int_0^1\int_0^s \norm{\rho-\si}{\rho_t^{-1}}^2\,dt ds\\
 \le &\int_0^1\int_0^s \frac{1}{1-t}\norm{\rho-\si}{\si^{-1}}^2 \,dtds\\
 = &\int_0^1\int_0^s \frac{1}{1-t}\,dtds \norm{\rho-\si}{\si^{-1}}^2=\norm{\rho-\si}{\si^{-1}}^2\pl.
\end{align*}
\end{proof}
\begin{rem}
Note that the upper bound does not require the assumption $\rho\le c\,\si$.
\end{rem}
Now given a conditional expectation $E_\cN:{\cM}\to\cN$, it follows immediately from the above that
 for any state $\rho$ and $\rho_{\cN}=E_{\cN*}(\rho)$,
\begin{align} k(C(E_{\cN}))\norm{\rho-\rho_{\cN}}{\rho_{\cN}^{-1}}^2\,\le D(\rho\| \rho_{\cN})\le\,  \norm{\rho-\rho_{\cN}}{\rho_{\cN}^{-1}}^2\pl,\label{keyinequality2}\end{align}
where $C(E_{\cN})$ is the index defined in \eqref{eq:indexE}. We also have an variant of the lower bound with another weighting state.
\begin{lemma}\label{lemma:keylemma2}Let $\rho$, $\si$ and $\omega$ be three full-rank density operators and suppose $\rho,\si\le c\,\omega$ for some $c>0$.
Then
\begin{align}\norm{\rho-\si}{\omega^{-1}}^2\le 2c\,D(\rho\| \si)\,. \end{align}
\end{lemma}
\begin{proof}Take $\rho_t=(1-t)\si+t\rho, t\in [0,1]$. By the assumption and Lemma \ref{lemma:compare}, we have $\rho_t\le c\,\omega$ and hence
\[ c\norm{\rho-\si}{\rho_t^{-1}}^2\,\ge\, \norm{\rho-\si}{\omega^{-1}}^2\]
for each $t$. Therefore,
\begin{align*}
 D(\rho\|\si)=
  &\int_0^1\int_0^s \norm{\rho-\si}{\rho_t^{-1}}^2\,dt ds\\
 \ge &\int_0^1\int_0^s \frac{1}{c}\, \norm{\rho-\si}{\omega^{-1}}^2\,dt ds\\
 \ge &\  \frac{1}{2c} \norm{\rho-\si}{\omega^{-1}}^2\pl.
\end{align*}
\end{proof}

\subsection{Detailed balance}
We shall now discuss the detailed balance condition and its connection to the spectral gap.
Given a full-rank state $\si$ and $0\le s\le 1$, we define the multiplication operator
\[\Gamma_{\si,s}(X)=\si^{1-s}X\si^s\pl.\]
$\Gamma_{\si,s}$ is a positive operator on the Hilbert-Schmidt space and induces
the following weighted inner product
\begin{align*}
\lan X,Y \ran_{\si,s}:=\tr(X^\dagger\si^{1-s}Y\si^s)\,,\qquad\quad \norm{X}{\si,s}^2=\lan X, X\ran_{\si,s}\,.
\end{align*}
We denote by $L_2(\sigma,s)$ the corresponding $L_2$ space. A map $\Phi^*:\cM\to \cM$ is self-adjoint with respect to $\lan \cdot,\cdot\ran_{\si,s}$ if
\[\Phi\circ \Gamma_{\si,s} =\Gamma_{\si,s}\circ \Phi^* \pl,\]
where $\Phi$ is the adjoint of $\Phi^*$ for the trace inner product. Denote
\[H=-\ln\si,\,\qquad \Delta_\si(X)=\si X\si^{-1},\, \qquad \al_t(X)=e^{itH}Xe^{-itH}, \quad t\in \mathbb{C}\]
as the modular generator, modular operator, and modular automorphism group of $\si$ respectively. It was proved in \cite[Theorem 2.9]{carlen2017gradient} that under the assumption $\Phi^*(a^\dagger)=(\Phi^*(a))^\dagger$, $\Phi^*$ is self-adjoint with respect to $\lan \cdot,\cdot\ran_{\si,s}$ for some $s\neq 1/2$ if and only if $\Phi^*$ commutes with $\Delta_\si$ and is self-adjoint for $s=1/2$, and hence $\Phi^*$ is self-adjoint with respect to $\lan \cdot,\cdot\ran_{\si,s}$ for all $s\in [0,1]$. We say that a map $\Phi^*$ satisfies  $\si$-DBC (detailed balance condition) if $\Phi^*$ is self-adjoint with respect to $\lan \cdot,\cdot\ran_{\si,1}$. Note that
$$\Gamma_{\si}=\displaystyle \int_{0}^1 \Gamma_{\si,s} \,ds\,.$$
Thus, we also have $\Gamma_{\si}\circ \Phi^*=\Phi\circ \Gamma_{\si} $ and hence
$\Gamma_{\si}^{-1}\circ \Phi=\Phi^{*}\circ \Gamma_{\si}^{-1} $ if $\Phi^*$ satisfies the $\si$-DBC.

Let $E_{\cN}:\cM\to \cN$ be a conditional expectation. It can be readily seen that $E_{\cN}$ satisfies the $\si$-DBC condition for all $\si\in \cD(E_\cN)$ (invariant state satisfying $\si=E_{\cN*}(\si)$). Hence
\[\forall s\in [0,1],\,\Gamma_{\si,s}\circ E_{\cN}=E_{\cN*} \circ \Gamma_{\si,s}\pl \text{ and } \pl \Gamma_{\si}\circ E_{\cN}=E_{\cN*} \circ \Gamma_{\si}\pl.\]
In particular, $E_{\cN}$ is the projection onto $\cN$ for the $L_2$-norms $\norm{\cdot}{\si,s}$ for any $s\in [0,1]$ and $\norm{\cdot}{\si}$, for all $\si\in \cD(E_\cN)$. Indeed, for any $X\in \cM$,
\begin{align*}\lan E_\cN(X), X-E_\cN(X) \ran_{\si,s}=&\lan \Gamma_{\si,s}\circ E_\cN(X), X-E_{\cN}(X)\ran_{\small{\operatorname{HS}}}
\\= &\lan  E_{\cN*}\circ \Gamma_{\si,s} (X), X-E_{\cN}(X)\ran_{\small{\operatorname{HS}}}\\ =&\lan  \Gamma_{\si,s} (X), E_{\cN}(X-E_{\cN}(X))\ran_{\small{\operatorname{HS}}}=0\,.
\end{align*}
Now, let $\Phi:\cM_*\to\cM_*$ be a quantum channel and $\cN$ be the multiplicative domain of $\Phi^*$:
\[\cN:=\{a\in \cM\pl | \pl \Phi^{*}(aa^{\dagger})=\Phi^{*}(a)\Phi^{*}(a^{\dagger})\pl,\Phi^*(a^\dagger a)=\Phi^*(a^\dagger)\Phi^*(a) \}\pl.\]
There always exists an invariant state $\si$ such that $\Phi(\si)=\si$.  The next lemma shows that if $\Phi^*$ satisfy $\si$-DBC, then $\Phi^*$ restricted to $\cN$ is a $*$-involution.
% There always exists an invariant state $\si$ such that $\Phi(\si)=\si$.  We will assume that $\Phi^*$ satisfy $\si$-DBC.
% Let $E_\cN:\cB(\cH)\to\cN$ be the conditional expectation such that $E_{\cN*}(\si)=\si$. \textcolor{red}{does it always exist even for $\si$ not full-rank?}

\begin{lemma}\label{md}
Let $\Phi:\cM_*\to \cM_*$ be a quantum channel and let $\cN$ be the multiplicative domain of $\Phi^*$. Then,
% , i.e.
% \[\cN=\{a\in \cB(\cH)\pl | \pl \Phi^*(aa^*)=\Phi^*(a)\Phi^*(a^*)\pl,\Phi^*(a^*a)=\Phi^*(a^*)\Phi^*(a) \}\pl.\]
\begin{enumerate}
\item[i)]There exists an invariant state $\si$ such that $\Phi(\si)=\si$
\end{enumerate}
If in addition $\sigma$ is full-rank and $\Phi^*$ satisfies $\si$-$\operatorname{DBC}$,
\begin{enumerate}
\item[ii)] $\Phi^*$ is a contraction on $L_2(\si,s)$ for any $s\in[0,1]$ and $L_2(\si)$. $\Phi^*$
restricted to $\cN$ is a $*$-isomorphism and an $L_2$-isometry on $L_2(\si,s)$ for all $s\in[0,1]$, as well as on $L_2(\si)$.
\item[iii)] Let $E_\cN:\cM\to\cN$ be the conditional expectation such that $E_{\cN*}(\si)=\si$. Then
\[\Phi^*\circ E_{\cN}= E_{\cN} \circ\Phi^*\pl, (\Phi^*)^2\circ E_{\cN}=E_{\cN}\circ(\Phi^*)^2= E_{\cN} \,.\]
\end{enumerate}
\end{lemma}
\begin{proof}
i) Viewing $\Phi$ as a linear map, $\Phi$ has eigenvalue $1$ because $\Phi^*(\Id)=\Id$. Since $\Phi$ preserves self-adjointness, we have an operator $a=a^\dagger$ such that $\Phi(a)=a$. Let $a_+$ (resp. $a_-$) be the positive (resp. negative) part of $a$. We have $\Phi(a)=\Phi(a_+)-\Phi(a_-)=a$. Because $\Phi$ is positive and trace preserving, $\Phi(a_+)$ and $\Phi(a_-)$ are positive and
\[ \tr(\Phi(a_+))+\tr(\Phi(a_-))=\tr(a_+)+\tr(a_-)=\norm{a}{1}\pl.\]
This implies $\Phi(a_+)=a_+$ and $\Phi(a_-)=a_-$, which proves i). For any $X\in\cM$,
\begin{align*}
\norm{\Phi^*(X)}{\si,s}^2=&\tr\Big(\Phi^*(X^\dagger)\si^{1-s}\Phi^*(X)\si^{s}\Big)
\\=&\tr\Big(\Phi^*(\al_{i\frac{1-s}{2}}(X)^\dagger)\Phi^*(\al_{i\frac{1-s}{2}}(X)) \si\Big)
\\ \le&\tr\Big(\Phi^*(\al_{i\frac{1-s}{2}}(X)^\dagger\al_{i\frac{1-s}{2}}(X) ) \si\Big)
\\ =&\tr\Big(\al_{i\frac{1-s}{2}}(X)^\dagger\al_{i\frac{1-s}{2}}(X) \si\Big)
\\ =&\norm{X}{\si,s}^2
\end{align*}
In the above inequality, we used the Kadison-Schwarz inequality and the second to last equality follows from $\Phi(\si)=\si$. Note that $\al_{s}(\cN)=\cN$ for any $s\in \mathbb{C}$. Then for any $X\in \cN$,
$\Phi^*(\al_{i\frac{1-s}{2}}(X)^\dagger)\Phi^*(\al_{i\frac{1-s}{2}}(X))=\Phi^*(\al_{i\frac{1-s}{2}}(X)^\dagger\al_{i\frac{1-s}{2}}(X))$ and the above inequality becomes an equality. This proves ii) for $L_2(\si,s)$ for all $s\in[0,1]$. The assertion for $L_2(\si)$ follows by integration. For iii),
we first note that for any $X\in \cN$, $(\Phi^*)^2(X)=X$. Indeed,
\[ \lan (\Phi^*)^2(X),X\ran_{\si,s}=\lan \Phi^*(X), \Phi^*(X)\ran_{\si,s}=\norm{X}{\si,s}^2\pl.\]
This further implies $\Phi^*(X)\in \cN$ is in the multiplicative domain because \[\Phi^*(\Phi^*(X^\dagger)\Phi^*(X))= \Phi^*(\Phi^*(X^\dagger X))=X^\dagger X=(\Phi^*)^2(X^\dagger)(\Phi^*)^2(X^\dagger)\pl.\]
Also, $\Phi^*$ is invariant on the orthogonal complement of $\cN$ because for any $Y\in\cM$,
\[ \lan X,\Phi^*\circ (\id-E_\cN)(Y)\ran_{\si,s}=\lan \Phi^*(X), (\id-E_\cN)(Y)\ran_{\si,s}=0\pl.\]
That completes the proof.
\end{proof}

We see from the above lemma that under $\si$-DBC, $\Phi^*$ is a self-adjoint contraction on $L_2(\si,s)$ (also $L_2(\si)$), and $\cN$ is the union of the eigenspace of $\Phi^*$ for eigenvalue $1$ and $-1$. The eigenspace for eigenvalue $1$ is the fixed point space of $\Phi^*$, which is a subalgebra $\cF\subset\cN$. For each invariant state $\si=\Phi(\si)$, we have $\si=E_{\cF*}(\si)$. In finite dimensions, there always exists $0<\epsilon<1$ such that
\[ \norm{\Phi^*(\id-E_{\cN}):L_2(\si,s)\to L_2(\si,s)}{}\le (1-\epsilon)\pl,\]
which is a spectral gap condition. The next lemma shows that this spectral gap condition is independent of $s\in[0,1]$ and of the choice of invariant state $\si$.

\begin{lemma}\label{lemma:independence}
Let $\Phi:\cM_*\to \cM_*$ be a quantum channel and $\Phi^*$ be its adjoint. Suppose $\Phi^*$ satisfy $\si$-$\operatorname{DBC}$ for some full-rank invariant state $\si$ such that $\Phi(\si)=\si$. %{\color{red}: Assume that $\Phi^*$ is invariant on its multiplicative domain $\cN$.}
Then,
\begin{enumerate}
\item[i)] $(\Phi^*)^2$ satisfies $\rho$-$\operatorname{DBC}$ for all states $\rho\in \cD(E_\cN)$ and $\Phi^*$ satisfies $\rho$-$\operatorname{DBC}$ for all invariant states $\rho$.
\item[ii)]For each full-rank state $\rho\in \cD(E_\cN)$, denote $\la(\rho,s)=\norm{\Phi^*(\id-E_{\cN}):L_2(\rho,s)\to L_2(\Phi(\rho),s)}{}^2$. Then for all $s\in [0,1]$
    \[\la(\rho,s)=\la(\si,1)\pl. \]
\item[iii)] For each full-rank state $\rho\in \cD(E_\cN)$,  denote $\la(\rho):=\norm{\Phi^*(\id-E_{\cN}):L_2(\rho)\to L_2(\Phi(\rho))}{}^2$. Then \begin{align*}\la(\rho):=&\norm{\Phi(\id-E_{\cN*}):L_2(\rho^{-1})\to L_2(\Phi(\rho)^{-1})}{}\\ =& \norm{\Phi^*(\id-E_{\cN}):L_2(\rho)\to L_2(\Phi(\rho))}{}= \la(\si,1)=\la(\si)\pl.\end{align*}
\end{enumerate}

\end{lemma}
\begin{proof}
By Lemma \ref{md}, $(\Phi^*)^2|_{\cN}$ is the identity map and we have the module property
\[ (\Phi^*)^2(aXb)=a (\Phi^*)^2(X)b \pl,\pl \forall a,b\in \cN\pl.\]
Note that for any two states $\rho,\si\in \cD(E_\cN)$, $\rho^{-s}\si^{s}\in\cN$ for any $s\in \mathbb{C}$.
Therefore, we have for all $s\in[0,1]$,
\[ \Gamma_{\rho,s}\circ(\Phi^*)^2\circ \Gamma_{\rho,s}^{-1}= \Gamma_{\si,s}\circ (\Phi^*)^2\circ \Gamma_{\si,s}^{-1}=\Phi^2\pl.\]
This shows $(\Phi^*)^2$ satisfies $\rho$-DBC. Now consider a state $\rho$ such that $\Phi(\rho)=\rho$. Because both $\rho,\si\in \cD(E_\cF)$, we have $\rho^{-s}\si^{s}\in \cF$ for any $s\in \mathbb{C}$. Then it follows from the same argument above that $\Phi^*$ satisfies $\rho$-DBC.
For ii), we denote $\iota=\Phi^*|_{\cN}$ to be the involution $\Phi^*$ restricted to $\cN$. Note that for any $s\in \mathbb{C}$, it can be verified by the finite dimensional direct sum structure in \eqref{cd} that
\begin{align}\label{eqPhipassing}
\iota(\rho^{-s}\si^{s})=\Phi(\rho)^{-s}\si^{s}
\end{align}
where $\rho\circ\iota=\Phi(\rho)$. For a mean zero element $Y=X-E_\cN(X)$,
\begin{align*}\norm{Y}{\rho,s}^2=&\norm{\Gamma_{\rho,s}^{1/2}(Y)}{2}^2
\\=&\norm{\Gamma_{\si,s}^{1/2}\Gamma_{\si,s}^{-1/2}\Gamma_{\rho,s}^{1/2}(Y)}{2}^2
\\=&\norm{\Gamma_{\si,s}^{1/2}(Y_0)}{2}^2
\\=&\norm{Y_0}{\si,s}^2
\end{align*}
where $Y_0= \Gamma_{\si,s}^{-1/2}\Gamma_{\rho,s}^{1/2}(Y)$ is also a mean zero element in $\cN^\perp$.
Moreover, \begin{align*}\norm{\Phi^*(Y_0)}{\si,s}^2=&\norm{\Gamma_{\si,s}^{1/2}\Gamma_{\si,s}^{-1/2}\Gamma_{\Phi(\rho),s}^{1/2}\Phi^*(Y_0)}{2}^2
\\=&\norm{\Gamma_{\Phi(\rho),s}^{1/2}\Phi^*(Y_0)}{2}^2
\\=&\norm{\Phi^*(Y_0)}{\Phi(\rho),s}^2
\end{align*}
where we used \eqref{eqPhipassing} in the first line. This proves $\la(\rho,s)=\la(\si,s)$ for each $s$. For the independence of $s$, we have for $r\in[0,1]$:
\begin{align*}\norm{\Phi^*(Y)}{\si,s}^2=\tr\big[\Phi^*(Y)^\dagger\si^{1-s}\Phi^*(Y)\si^{s}\big]
=&\tr\big[ \Phi^*(\al_{i\frac{r-s}{2}}(Y))^\dagger\si^{1-r}\Phi^*(\al_{i\frac{r-s}{2}}(Y))\si^{r}\big]\\
=&\norm{\Phi^*\big(\al_{i\frac{r-s}{2}}(Y)\big)}{\si,r}^2
\end{align*}
where $\al_{i\frac{r-s}{2}}(Y)=\al_{i\frac{r-s}{2}}(X-E_{\cN}(X))=\al_{i\frac{r-s}{2}}(X)-E_{\cN}(\al_{i\frac{r-s}{2}}(X))$ is also in $\cN^\perp$. Moreover,
\begin{align*}\norm{Y}{\si,s}^2=\norm{\al_{i\frac{r-s}{2}}(Y)}{\si,r}^2\pl.
\end{align*}
For iii), the inequality $\la(\rho)\le \la(\si,1)$ follows from integrating the $\lan \cdot,\cdot \ran_{\rho,s}$ inner product to obtain $\lan \cdot,\cdot \ran_{\rho}$. The equality $\la(\si,1)=\la(\si)$ follows from the fact that the map $\Phi^*(\id-E_{\cN})$ is self-adjoint with respect to both $\lan \cdot,\cdot\ran_\si$ and $\lan \cdot,\cdot\ran_{\si,s}$ for any $s\in[0,1]$. Then the quantity $\norm{\Phi^*(\id-E_{\cN})}{}$, which is equal to the maximal eigenvalue of $\Phi^*(\id-E_{\cN})$, is independent of the choice of Hilbert space norm $\|\cdot \|$. We note that by \eqref{eqPhipassing}
\[ \Gamma_{\Phi(\rho),s}\circ \Phi^*\circ\Gamma_{\rho,s}^{-1}=\Gamma_{\Phi(\rho),s}\circ \Phi^*\circ\Gamma_{\rho,s}^{-1}\Gamma_{\si,s}\Gamma_{\si,s}^{-1}=\Gamma_{\Phi(\rho),s} \Gamma_{\Phi(\rho),s}^{-1}\Gamma_{\si,s} \circ\Phi^*\circ\Gamma_{\si,s}^{-1}=\Phi\]
and
\[ \Gamma_{\Phi(\rho)}\circ \Phi^*\circ \Gamma_{\rho}^{-1}=\Phi\pl, \]
This implies $\Gamma_{\Phi(\rho)}\circ\Phi^*(\id-E_{\cN})\circ \Gamma_{\rho}^{-1}=\Phi(\id-E_{\cN*})
$ and hence
\begin{align*}
\norm{\Phi(\id-E_{\cN*}):L_2(\rho^{-1})\to L_2(\Phi(\rho)^{-1})}{} =& \norm{\Gamma_{\Phi(\rho)}\circ\Phi^*(\id-E_{\cN})\circ \Gamma_{\rho}^{-1}:L_2(\rho^{-1})\to L_2(\Phi(\rho)^{-1})}{}
\\
=& \norm{\Phi^*(\id-E_{\cN}):L_2(\rho)\to L_2(\Phi(\rho))}{} \,.
\end{align*}
Moreover, since both $\si$ and $\rho$ are invariant to $\Phi^2$, we have by ii)
\begin{align*}
\la(\si)=&\norm{(\Phi^2)^*(\id-E_{\cN}):L_2(\si)\to L_2(\si)}{}=\norm{(\Phi^2)^*(\id-E_{\cN}):L_2(\rho)\to L_2(\rho)}{}\\ \le&
\norm{\Phi^*(\id-E_{\cN}):L_2(\rho)\to L_2(\Phi(\rho))}{}\norm{\Phi^*(\id-E_{\cN}):L_2(\Phi(\rho))\to L_2(\rho)}{} \\ \le& \norm{\Phi^*(\id-E_{\cN}):L_2(\si)\to L_2(\si)}{}^2=\la(\si).
\end{align*}
That verifies iii).
\end{proof}

\section{Modified logarithmic Sobolev inequalities}\label{sec:CMLSI}
In this section, we prove the complete modified logarithmic Sobolev inequality (CMLSI) for quantum Markov semigroups on finite dimensional matrix algebras. The argument is a simple application of the key estimates in Subsection \ref{subsec:key}. Let $\cM\subset \cB(\cH)$ be a finite dimensional von Neumann algebra.
A quantum Markov semigroup (QMS) $(\cP_t)_{t\ge0}:\cM\to \cM$ is a continuous parameter semigroup of completely positive, unital maps such that $\cP_0=\id_\cM$ and $\cP_s\circ\cP_t =\cP_{s+t}$ for all $s,t\ge 0$.
Such a semigroup is characterised by its generator, called the Lindbladian $\cL$, which is defined as $$\cL(X)={\lim}_{t\to 0}\,\frac{1}{t}\,(\cP_t(X)-X)\pl,\quad  \pl \forall \pl X\in\cM,$$ so that $\cP_t=\operatorname{e}^{t\cL}$ for all $t\ge 0$. A QMS is said to be \textit{primitive} if it admits a unique full-rank invariant state $\sigma$. In this section, we assume $\cP_t:\cM\to \cM$ admits an extension to a semigroup on $\cB(\cH)$ and exclusively study QMS that satisfy the following \textit{detailed balance condition} with respect to some (possibly non-unique) full-rank invariant state $\sigma$: if for any $X,Y\in\cM$ and any $t\ge 0$:
\begin{align}\label{eq:DBC}\tag{$\sigma$-DBC}
\tr(\sigma\,X^\dagger\cP_t(Y))=\tr(\sigma\,\cP_t(X)^\dagger Y)\,.
\end{align}
We say a semigroup $\cP_t$ is GNS-symmetric if $\cP_t$ satisfies $\sigma$-DBC for a full-rank invariant state $\sigma$. Under this condition,  the generator $\cL$ can be written as \cite{carlen2017gradient}
\begin{align}\label{eqlindbladsigma}
\cL(X)=\sum_{j}\,\Big(\operatorname{e}^{-\omega_j/2}\,A_j^\dagger [X,A_j]+\,\operatorname{e}^{\omega_j/2}[A_j,X]A_j^\dagger \Big)\,,
\end{align}
Here $A_j\in \cB(\cH)$ and $\omega_j$ are some real parameters such that for any invariant state $\sigma$, $\Delta_{\sigma}(A_j):=\sigma A_j\,\sigma^{-1}=\operatorname{e}^{-\omega_j}A_j$.
Moreover, there exists a conditional expectation ${E}_{\cN}:\cM\to\cF$ onto the fixed point algebra $\cF=\{X\in \cM  \pl |\pl  [A_{j},X]=0\pl \forall \pl j\pl \}$ such that \cite{frigerio1982long}.
$$
    \operatorname{e}^{t\cL}\underset{t\to\infty}{\to} {E}_\cF\,.
$$
% Next, we assume the following block decomposition for the fixed point algebra: assuming $\mathcal{F}\subset\cM\subset \cB(\cH)$, there exists a decomposition $\cH:=\bigoplus_i\,\cH_i\otimes \mathcal{K}_i$ of $\cH$ such that
% \begin{align*}
% \mathcal{F}=\bigoplus_{i}\,\mathcal{B}(\cH_i)\otimes\Id_{\cK_i}\,,\qquad \sigma_{\operatorname{Tr}}:=\sum_{i}\,\frac{d_{\cK_i}}{d_\cH}\Id_{\cH_i}\otimes\tau_i\,,
% \end{align*}
% where $\tau_i$ are full-rank states on each sector $\mathcal{K}_i$ such that $\sigma_{\tr}=E_*(\sigma_{\tr})$. Finally, we denote the spectral decomposition of each state $\tau_i$ as $\tau_i=\sum_{k}\,\lambda_k^{(i)}\,P_k^{(i)}$.
We are interested in the exponential convergence to this limit in terms of relative entropy. Recall that the \textit{entropy production} (sometimes also referred as \textit{Fisher information}) for a state $\rho\in\cD(\cM)$ is defined as
\begin{align*}
\operatorname{EP}_{\cL}(\rho):=-\left.\frac{d}{dt}\right|_{t=0}\,D(\cP_{t*}(\rho)\|E_{\cF*}(\rho))=-\tr(\cL_*(\rho)(\ln\rho-\ln E_{\cF*}(\rho)))\,,
\end{align*}
which is the opposite of the derivative of the relative entropy with respect to the equilibrium state.
Here and in the following $\cL_*$ (resp. $\cP_{t*}$ and $E_{\cF*}$) denote the adjoint maps of the generator $\cL$ (resp. semigroup map $\cP_{t}$ and conditional expectation $E_{\cF}$). We say a QMS $\cP_{t}:\cM\to \cM$
satisfies the \textit{modified logarithmic Sobolev inequality} (MLSI) with $\al>0$ if for any $\rho\in\cD(\cM)$,
\begin{align}\label{MLSI}\tag{MLSI}
\alpha \,D(\rho\|E_{\cF*}(\rho))\le \operatorname{EP}_{\cL}(\rho)\,.
\end{align}
The best constant $\alpha$ achieving this bound is called the modified logarithmic Sobolev constant of the semigroup, and is denoted by $\alpha_{\operatorname{MLSI}}(\cL)$. It turns out that
this inequality is equivalent to the following exponential decay of relative entropy
\begin{align*}
D(\cP_{t*}(\rho)\|E_{\cF *}(\rho))\le \operatorname{e}^{-\alpha t}D(\rho\|E_{\cF *}(\rho))\,.
\end{align*}
 We also consider the \textit{complete modified logarithmic Sobolev inequality} (CMLSI) which requires
\begin{align}\label{CLSI}\tag{CMLSI}
\alpha\,D(\rho\|(E_{\cF *}\otimes \id) (\rho))\le \operatorname{EP}_{(\cL\otimes \id)}(\rho)
\end{align}
to hold for all states $\rho$ on $\cM\otimes \cB(\cH)$ and any finite dimensional Hilbert space $\cH$
as a reference system (or even $\mathcal{B}(\cH)$ replaced by a finite von Neumann algebra). We denote
the best constant $\alpha$ achieving \eqref{MLSI} as $\alpha_{\operatorname{CMLSI}}(\cL)$. In \cite{junge2019stability}, it was shown that the proof of the positivity of $\alpha_{\operatorname{CMLSI}}$ for all GNS-symmetric quantum Markov semigroups can be reduced to that for (trace) symmetric quantum Markov semigroups, that is to those for which $\cL=\cL_*$. However, the problem of the positivity of the CMLSI constant for symmetric QMS has been left open despite considerable work delved on that topic in the recent years (see e.g. \cite{gao2020fisher,brannan2020complete,brannan2020complete2,wirth2020complete}). Here, we provide a positive answer to the question via a simple application of our key estimates from Subsection \ref{subsec:key}.

% First, we recall that any symmetric quantum Markov semigroup has generator of the form
% \begin{align*}
%     \cL(X)=-\sum_{j}[A_j^\dagger ,[A_j,\rho]]\,,
% \end{align*}
% where the operators $A_j$ commute with any fixed point: for all $X\in\cM$ and all $j$, $[A_j,X]=0$.

First, we recall that the Dirichlet form associated to $\cL$ takes the following simple form \cite[Section 5]{carlen2017gradient}: for any invariant state $\sigma=E_{\cF*}(\sigma)$,
\begin{align}\label{Dirichlet}
    \cE_{\sigma}(X):=-\langle X,\,\cL(X)\rangle_{\sigma}=\sum_{j}\,\int_0^1 \operatorname{e}^{(\frac{1}{2}-s)\,\omega_j}\langle \partial_j(X),\,\partial_j(X)\rangle_{{\sigma,s}}\,ds\,,
\end{align}
where $\partial_j(X):=[A_j,X]$. We denote
\begin{align}\label{Dirichletcompact}
    \|X\|_{\sigma,\omega_j}:=\int_0^1 \operatorname{e}^{(\frac{1}{2}-s)\,\omega_j}\langle \partial_j(X),\,\partial_j(X)\rangle_{{\sigma,s}}\,ds\quad \Rightarrow \quad \cE_\sigma(X)=\sum_j\,\|\partial_j(X)\|_{\sigma,\omega_j}^2\,.
\end{align}
Then the entropy production associated to $\cL$ can be written as (see \cite[Lemma 2.3]{junge2019stability}):
\begin{align}\label{eq:EP}
    \operatorname{EP}_{\cL}(\rho)&=\sum_{j}\,\|\Gamma_{\sigma,\frac{1}{2}}\circ \partial_j\circ \Gamma_{\sigma,\frac{1}{2}}^{-1}(\rho)\|_{\rho^{-1},\omega_j}^2\,,
    \end{align}
    where for any $X\in\cM$:
    \begin{align*}
  \|X\|_{{\rho}^{-1},\omega_j}^2=\int_{0}^\infty\tr\,\Big[X^\dagger \,(\operatorname{e}^{-\frac{\omega_j}{2}}\rho+u)^{-1}X(\operatorname{e}^{\frac{\omega_j}{2}}\rho+u)^{-1}\Big]\,du\,.
\end{align*}
We denote the kernels corresponding to the inner products $\|.\|_{\sigma,\omega_j}$ and $\|.\|_{\sigma^{-1},\omega_j}$ by $\Gamma_{\sigma,\omega_j}$ and $\Gamma_{\sigma^{-1},\omega_j}$ respectively.
\begin{lemma}
 The following relation holds for any full-rank state $\sigma$:
\begin{align}\label{GammaGammainverse}
    \Gamma_{\sigma,\omega_j}^{-1}=\Gamma_{\sigma^{-1},\omega_j}\,.
    \end{align}
    Moreover, whenever $\sigma=E_{\cF*}(\sigma)$:
    \begin{align}\label{Gammahalftoomega}
        \Gamma_{\sigma,\frac{1}{2}}\circ\partial_j\circ\Gamma_{\sigma,\frac{1}{2}}^{-1}=\Gamma_{\sigma,\omega_j}\circ\partial_j\circ\Gamma_{\sigma}^{-1}\,.
    \end{align}
\end{lemma}
\begin{proof}
The first identity follows from Lemma 5.8 in \cite{carlen2017gradient}. The proof of the second identity follows by direct computation using the commutation relation $\sigma A_j=\operatorname{e}^{-\omega_j}A_j\sigma$.
\end{proof}
We recall that the spectral gap $\lambda(\cL)$ of the Lindbladian $\cL$ is characterized as
\begin{align}\label{spectralgap}
    \lambda(\cL):=\inf_{X}\frac{\mathcal{E}_\sigma(X)}{\|X-E_\cF(X)\|^2_{\sigma}}\,,
\end{align}
for a given full-rank invariant state $\sigma$.
\begin{lemma}Suppose $\cP_t$ is $\operatorname{GNS}$-symmetric to a full-rank invariant state $\sigma=E_{\cF*}(\sigma)$.
Then the infimum in \eqref{spectralgap} is independent of the choice of the full-rank invariant state $\sigma$.
\end{lemma}
\begin{proof}
By assumption the generator $\mathcal{L}$ is symmetric with respect to the GNS inner product \eqref{eq:DBC}, which also implies self-adjointness with respect to the inner products $\langle.,.\rangle_\sigma$ (cf. \cite[Theorem 2.9]{carlen2017gradient}). Moreover, self-adjointness with respect to the GNS inner product is independent of the invariant state chosen. Therefore, $\mathcal{L}$ is self-adjoint with respect to $\langle.,.\rangle_\sigma$ for any full-rank invariant state $\sigma$. Now, the spectral gap \eqref{spectralgap} is the difference between the smallest eigenvalue (here, $0$) and the second smallest eigenvalue of $-\mathcal{L}$, hence a quantity independent of the inner product with respect to which $\mathcal{L}$ is self-adjoint, which allows us to conclude.
\end{proof}
We are now ready to prove Theorem \ref{thm:1}, which is the main theorem of this section.

\begin{theorem}\label{theo:CLSI}
Any $\operatorname{GNS}$-symmetric quantum Markov semigroup on a finite dimensional von Neumann algebra $\cM\subseteq \cB(\cH)$ which admits an extension to a semigroup on $\cB(\cH)$ satisfies the complete modified logarithmic Sobolev inequality. More precisely, given such a $\operatorname{QMS}$ $(\cP_t=\operatorname{e}^{t\cL}:\cM\to \cM)_{t\ge 0}$ with fixed point algebra $\cF$, the following bound holds true:
\begin{align}\label{CMLSIconstant}
 \frac{\lambda(\cL)}{C_{\tau,\operatorname{cb}}(\cM:\cF)}\le     \alpha_{\operatorname{CMLSI}}(\cL)\le 2\lambda(\cL)\,.
\end{align}
Similarly, the modified logarithmic Sobolev inequality constant is controlled by
\begin{align}\label{MLSIconstant}
   \frac{\lambda(\cL)}{C_\tau(\cM:\cF)}\,\le \alpha_{\operatorname{MLSI}}(\cL)\le 2\lambda(\cL)\,.
\end{align}
\end{theorem}
\begin{proof}
The proof of the upper bounds is standard and can be found in \cite{bardet2017estimating,kastoryano2013quantum}, so we focus on the lower bounds. We first provide a bound on the MLSI constant. For this we use the upper bound in Lemma \ref{lemma:keylemma} that for $X:=\Gamma_{E_{\cF*}(\rho)}^{-1}(\rho)$,
\begin{align*}
    D(\rho\|E_{\cF *}(\rho))\le \|\rho-E_{\cF *}(\rho)\|_{E_{\cF *}(\rho)^{-1}}^2=\|X-\Id\|_{E_{\cF *}(\rho)}^2\le \lambda(\cL)^{-1}\,\cE_{E_{\cF*}(\rho)}(X)\,,
    \end{align*}
where $\lambda(\cL)$ is the spectral gap of $\cL$. Next, we have by \eqref{Dirichletcompact} that
\begin{align*}
%   \cE(X)&= \int_0^\infty\,\sum_{j}\,\tr\,\Big[\big(\Gamma_{E(\rho)}\circ\partial_j(X)\big)^\dagger\,(E(\rho)+u)^{-1}\Gamma_{E(\rho)}\circ \partial_j(X)(E(\rho)+u)^{-1}\Big]\,du\\
\cE_{E_{\cF*}(\rho)}(X)&=\sum_j\|\partial_j(X)\|_{E_{\cF *}(\rho),\omega_j}^2\\
&\overset{(1)}{=}\sum_j\|\Gamma_{E_{\cF*}(\rho),\omega_j}\circ \partial_j\circ \Gamma_{E_{\cF*}(\rho)^{-1}}(\rho)\|_{E_{\cF*}(\rho)^{-1},\ \omega_j}^2\\
&\overset{(2)}{=}\sum_j\|\Gamma_{E_{\cF*}(\rho),\frac{1}{2}}\circ \partial_j\circ \Gamma_{E_{\cF*}(\rho)^{-1},\frac{1}{2}}(\rho)\|_{E_{\cF*}(\rho)^{-1},\ \omega_j}^2\\
  &\overset{(3)}{\le}\,C_{\tau}(\cM:\cF)\,\sum_j\,\|\Gamma_{E_{\cF*}(\rho),\frac{1}{2}}\circ\partial_j\circ \Gamma_{E_{\cF*}(\rho)^{-1},\frac{1}{2}}(\rho)\|_{\rho^{-1},\ \omega_j}\\
 & \overset{(4)}{=}C_{\tau}(\cM:\cF)\,\operatorname{EP}_{\cL}(\rho)\,.
\end{align*}
For the above equality (1), we used the inverse relation \eqref{GammaGammainverse}; in (2) we used the relation \eqref{Gammahalftoomega}; (3) is an application of Lemma \ref{lemma:compare} with the weights $\mu_1:=\operatorname{exp}({-\frac{\omega_j}{2}})$ and $\mu_2:=\operatorname{exp}({\frac{\omega_j}{2}})$; finally (4) follows from \eqref{eq:EP}. The proof of CMLSI \eqref{CMLSIconstant} follows the exact same steps, up to replacing the constant $C_{\tau}(\cM:\cF)$ by its completely bounded version $C_{\tau,\operatorname{cb}}(\cM:\cF)$.
\end{proof}

\begin{rem}
The above theorem applies for the derivation triples introduced in Carlen-Maas's work \cite{carlen2020non} as well as the symmetric quantum Markov semigroup on finite von Neumann algebra considered in \cite{davies1992non,cipriani1997dirichlet,brannan2020complete} whenever the index $C_{\operatorname{cb}}(\cM:\cN)$ is finite. Here the assumption of (trace) symmetry for the latter is to ensure the existence of a derivation $\delta$ such that the Lindbladian $\cL=-\delta^\dag\delta$ (see e.g. \cite[Theorem 2.1]{brannan2020complete}).
\end{rem}

\begin{rem}
When $\mathcal{M}:=\mathcal{B}(\cH)$ and the semigroup is primitive, comparison to the \textit{logarithmic Sobolev constant} $\alpha_{\operatorname{LSI}}$ combined with standard interpolation inequalities provide the following bounds for $\alpha_{\operatorname{MLSI}}$ \cite{olkiewicz1999hypercontractivity,kastoryano2013quantum,carbone}:
\begin{align}\label{eq:MLSIbound}
  \frac{\lambda(\cL)}{\ln(\mu_{\min}(\sigma)^{-1})+2} \le  \alpha_{\operatorname{LSI}}(\cL)\le \frac{\alpha_{\operatorname{MLSI}}(\cL)}{2}\le \lambda(\cL)\,.
\end{align}
 The lower bound can be compared with the one provided in \eqref{MLSIconstant} together with \eqref{fromtracialtonontracial} and \eqref{eq:indicesex}:
\begin{align}\label{eq:CMLSIbound}
 \frac{\mu_{\min}(\sigma)\lambda(\cL)}{d_\cH}\le     \alpha_{\operatorname{MLSI}}(\cL)\,,\qquad \frac{\mu_{\min}(\sigma)\lambda(\cL)}{d_\cH^2}\le     \alpha_{\operatorname{CMLSI}}(\cL)\,.
\end{align}
 Clearly, the lower bounds in \eqref{eq:MLSIbound} are asymptotically tighter. However, we emphasis that our bounds \eqref{eq:CMLSIbound} are the first generic non-trivial lower bounds for non-primitive QMS, and the CMLSI bound are independent of the size of the environment and hence stable under tensorization, which is even new for primitive semigroup. For classical Markov semigroups (equivalently, graph Laplacians of a weighted graph), \eqref{eq:CMLSIbound} gives an alternative CMLSI bounds to the one proved in \cite{li2020graph}.
 In Sections \ref{sec:symmetric} and \ref{sec:localsemigroups}, we will use the approximate tensorization bounds, which is the subject of Section \ref{sec:AT}, to derive bounds on the CMLSI constant that are sharper than \eqref{eq:CMLSIbound} above. As we will see, in some cases, the CMLSI lower bounds can scale similarly to the LSI bounds in the primitive setting. It remain opens whether the CMLSI constant admits asymptotic bounds better than $O(d_\cH^{-2})\lambda(\cL)$ in general.
\end{rem}

\section{Strong data processing inequalities}\label{sec:CSDPI}
In this section, we study the complete strong data processing inequality for a quantum channel, which
is a discrete time analog of CMLSI.
We recall the definition of the weighted $L_2$-norm corresponding to a full-rank state $\omega$:
\[\norm{X}{\omega^{-1}}^2=\int_{0}^\infty \tr\Big(X^\dagger \frac{1}{\omega+s}\,X\, \frac{1}{\omega+s}\Big) \,ds\pl,\,\quad  X\in \cM_*\pl.\]
If $X=\rho-\omega$ for some other state $\rho$,
\[ \chi_2(\rho,\omega):=\norm{\rho-\omega}{\omega^{-1}}^2\pl\]
is a special case of the quantum $\chi_2$-divergence studied in \cite{wolf2012quantum}. It is known that $\chi_2$ also satisfies the data processing inequality: for a quantum channel $\Phi$,
\begin{align}\label{DPICHI}
 \chi_2(\Phi(\rho),\Phi(\omega))\le \chi_2(\rho,\omega)\pl.
\end{align}Indeed, the data processing inequality of relative entropy follows from \eqref{DPICHI} and the argument used in Lemma \ref{lemma:keylemma}.
We shall now discuss how to control relative entropy contraction coefficients by their $\chi_2$ analogues.

Let $\Phi:\cM_*\to \cM_*$ be a quantum channel and $\Phi^*$ be the adjoint map of $\Phi$.
We denote by $\cN$ the multiplicative domain of $\Phi^*$. Then $\Phi^*$ restricted to $\cN$ is a $*$-isomorphism. %(see \cite[Theorem 5.7]{wolf2012quantum}).
Suppose $\Phi$ admits a full-rank invariant state $\si$. Denote by $E:\cM\to \cN$ the $\si$-preserving condition expectation and by $E_*$ its pre-adjoint on $\cM_*$.
For a full-rank state $\omega$, we have discussed the following $L_2$-contraction constant in  \cref{lemma:independence}
\[ \la(\omega):=\norm{ \Phi(\id-E_*):L_2(\omega^{-1})\to L_2(\Phi(\omega)^{-1})}{}^2\pl .\]
Equivalently, $\la(\omega)$ gives the contraction coefficient of $\chi_2$:
 \[\la(\omega)=\sup_{E_*(\rho)=E_*(\omega),\rho\neq \omega}\frac{\chi_2(\Phi(\rho),\Phi(\omega))} {\chi_2(\rho,\omega)}\pl,\]
where the supremum is over all state $\rho\neq \omega$ with $E_*(\rho)=E_*(\omega)$. Here we restrict our optimization to states $\rho$ and $\omega$ with the same ``mean'' (also called decoherence free part) given by the map $E_*$. This is because if $\cN\neq \mathbb{C}\Id$ is not trivial, then for any two invariant states $\si,\si'\in \cD(\cN)$,
\[ \chi_2(\Phi(\si'),\Phi(\si))= \chi_2(\si',\si)\pl,~~~ D(\Phi(\si')\|\Phi(\si))= D(\si'\|\si)\pl, \]
and hence $\la(\si)=1$ for any invariant state $\Phi(\si)=\si$.

%Under this restriction, we justifies this quantity is always strictly contractive

%Similarly, we define $$\al(omega):=\sup_{\rho\neq \omega}\frac{D(\Phi(\rho)\|\Phi(\omega))}{D(\rho\|\omega)}$$.
%as the entropy contraction coefficient at the base state $\omega$. We justifies the restriction $E_*(\rho)=E_*(\omega)$ in the following proposition.

%\begin{proposition}
%Let $\Phi:\cM_*\to\cM_*$ be a quantum channel that admits some full-rank invariant state $\si=\Phi(\si)$. Let $\omega$ be a full-rank state. Let $\omega$ be a full-rank state and $\rho$ be another state such that $E_*(\rho)=E_*(\omega)$. The following are equivalent.
%\begin{enumerate}
%\item[i)] $D(\Phi(\rho)\|\Phi(\omega))=D(\rho\|\omega)$
%\item[ii)] $\chi_2(\Phi(\rho), \Phi(\omega))=\chi_2(\rho,\omega)$
%\item[iii)] $\rho=\omega$
%\end{enumerate}
%Moreover, $\la(\omega)<1$.
%\end{proposition}
%\begin{proof}
%iii)$\rightarrow$ i) is trvial. $i)$ implies $\rho=R_\omega\circ \Phi(\rho)$ and $R(\omega)=\omega$ for the $\omega$'s Petz recovery map $$R_\omega(x)=\omega^{1/2}\Phi^*(\Phi(\omega)^\frac{1}{2} x\Phi(\omega)^\frac{1}{2})\omega^{1/2}.$$

%\end{proof}

The next theorem is a quantum analog of \cite[Theorem 3.4]{Raginsky16} which shows that the $\chi_2$ contraction coefficient implies local strong data processing inequality.
\begin{theorem}\label{thm:local}
Let $\Phi:\cM_*\to\cM_*$ be a quantum channel that  admits some full-rank invariant state $\si=\Phi(\si)$. Let $\omega$ be a full-rank state and denote $\la(\omega):=\norm{ \Phi(\id-E_*):L_2(\omega^{-1})\to L_2(\Phi(\omega)^{-1})}{}^2$. Then for any state $\rho$ with $E_*(\omega)=E_*(\rho)$,
\begin{align} D(\Phi(\rho)\|\Phi(\omega))\le c\, D(\rho\|\omega)\label{eq:local}\end{align}
where $c$ is a constant such that
\begin{align} \la(\omega)\le c \le c\,(C(\rho:\omega),\lambda(\omega))\,, \label{eq:localc}\end{align}
Here $ C(\rho:\omega):=\inf \{C | \rho\le C \,\omega\}$ and $c(C,\lambda)$ is an explicit function such that $c(C,\lambda)<1$ whenever $\lambda<1$. In particular, for any state $\rho$, $c\,(C(\rho:\omega),\lambda(\omega))\le c\,(\mu_{\min}(\omega)^{-1},\lambda(\omega))$ where $\mu_{\min}(\omega)$ is the minimum eigenvalue of $\omega$.
\end{theorem}
\begin{proof} We first show the lower bound. Write $\la\equiv \la(\omega)$.
Let $\rho$ be a state with $E_*(\rho)=E_*(\omega)$. Take the linear interpolation of states $\omega_t:=(1-t)\,\omega+t\,\rho,  t\in [0,1]$.
%\textcolor{red}{why do you need to asssume existence of $\eps$ here?}.
Now assume $\Phi$ satisfies \eqref{eq:local} for $c>0$. We have
\[ D(\Phi(\omega_t)\|\Phi(\omega))\le c\, D(\omega_t\|\omega)\,,\]
since $E_*(\omega_t)=E_*(\omega)$. Consider the function $f(t)= c\, D(\omega_t\|\omega)-D(\Phi(\omega_t)\|\Phi(\omega))$. Taking derivatives, we have $f(0)=f'(0)=0$, and \cite{lesniewski1999monotone}
\begin{align*}
 f''(0)=c\norm{\rho-\omega}{\omega^{-1}}^2-\norm{\Phi(\rho)-\Phi(\omega)}{\Phi(\omega)^{-1}}^2\pl.
\end{align*}
Note that $f''(0)\ge 0$ because $f(t)\ge 0$ for $t\in [0,\epsilon]$. Therefore,
\[\norm{\Phi(\rho-\omega)}{\Phi(\omega)^{-1}}^2\le c\norm{\rho-\omega}{\omega^{-1}}^2\pl. \]
 This proves the lower bound \[\la(\omega)\le c \ . \]
For the upper bound, denote
 $\rho_t=t\rho+(1-t)\,\omega$ and $g(t)=D(\rho_t\|\omega)-D(\Phi(\rho_t)\|\Phi(\omega))$. We have $g(0)=g'(0)=0$, and
\begin{align*}
 g''(t)=\norm{\rho-\omega}{\rho_t^{-1}}^2-\norm{\Phi(\rho-\omega)}{\Phi(\rho)_t^{-1}}^2\pl.
\end{align*}
It follows from \eqref{DPICHI} (see also \cite[Example 2]{lesniewski1999monotone}) that
$g''(t)\ge 0$.
Using Lemma  \ref{lemma:compare} and the definition of $\la(\omega)$, we also have
\begin{align*} g''(t)&=\norm{\rho-\omega}{\rho_t^{-1}}^2-\norm{\Phi(\rho-\omega)}{\Phi(\rho)_t^{-1}}^2\\ &\ge  \, \frac{1}{1+(C-1)t}\norm{\rho-\omega}{\omega^{-1}}^2-\frac{1}{1-t}\norm{\Phi(\rho-\omega)}{\Phi(\omega)^{-1}}^2\\ &\ge\,  \Big(\frac{1}{1+(C-1)t}-\frac{\la^2}{1-t}\Big)\norm{\rho-\omega}{\omega^{-1}}^2\end{align*}
where $C=\inf\{C \pl |\pl\rho\le C\, \omega\}$. Thus we have for $t_0:=\frac{1-\la^2}{1+\la^2(C-1)}$,
\begin{align*}
   g''(t)\ge
\begin{cases}
\Big(\frac{1}{1+(C-1)t}-\frac{\la^2}{1-t}\Big)\norm{\rho-\omega}{\omega^{-1}}^2\,,\pl \pl  t\le t_0\\
0\,,  ~~~~~~~~~~~~~~~~~~~~~~~~~ t> t_0
\end{cases}\,.
\end{align*}
Denote $\displaystyle a(s):=\int_0^s \frac{1}{1+(C-1)t}-\frac{\la^2}{1-t}\, dt= \frac{\ln (1+(C-1)s)}{C-1}+\la^2\ln(1-s)$. Since $g'(0)=0$, we have $g'(s)\ge a(s)\norm{\rho-\omega}{\omega^{-1}}^2$ if $s\le t_0$ and $g'(s)\ge a(t_0)\norm{\rho-\omega}{\omega^{-1}}^2$ if $s\ge t_0$. Denote \[b(t):=\int_{0}^ta'(s)ds=\frac{(1+(C-1)t)\ln (1+(C-1)t)-(C-1)t}{(C-1)^2}-\la^2((1-t)\ln (1-t)+t)\pl .\]
We have,
\begin{align*}
 D(\rho\|\omega)-D(\Phi(\rho)\|\Phi(\omega))=&g(1)-g(0)
 =\int_0^1g'(s)ds
 \\
 \ge& \big((1-t_0)a(t_0)+b(t_0)\big)\norm{\rho-\omega}{\omega^{-1}}^2
 \\
 \ge &  \big((1-t_0)a(t_0)+b(t_0)\big)\,D(\rho\|\omega)\,,
\end{align*}
where the last inequality follows from Lemma \ref{lemma:keylemma}. The SDPI constant is then upper bounded by
\[c= 1-(1-t_0)a(t_0)-b(t_0)<1\pl.\]
It is clear from the derivation that $c$ as a function depending on $C$ and $\la$ satisfies
\[ c(C,\la)\ge c(C',\la)\pl, \pl C'\ge C\ge 1 \pl.\]
Then the last assertion follows from $\rho\le \Id\le  \mu_{\min}(\omega)^{-1}\omega$.
\end{proof}
%\begin{rem}
%In the last step of the proof, we obtained the following improved data processing inequality (DPI)
%\[D(\rho\|\omega)-D(\Phi(\rho)\|\Phi(\omega))\ge  \big((1-t_0)a(t_0)+b(t_0)\big)\norm{\rho-\omega}{\omega^{-1}}^2\pl.\]
%This in particular implies that  $D(\rho\|\omega)=D(\Phi(\rho)\|\Phi(\omega))$ is saturated if and only if $\rho=\omega$. This gives a new form of refinement of DPI, which is different from recovery results (see e.g. \cite{J+,cv,GW} with correction term related to the recovery error).
%\end{rem}
%Based on the above Lemma \ref{lemma:independence} we define that the spectral gap of the channel as
%\[\la(\Phi):=\norm{\Phi^*(\id-E_{\cN}):L_2(\si)\to L_2(\si)}{}\pl.\]
%which is independent of the choice invariant state. Moreover, $\la(\phi)$ also upper bounds the contraction coefficients of $\chi_2$-divergence for any state $\rho\in D(E_\cN)$.

Next, we consider strong data processing inequality for a quantum channel $\Phi:\cM_*\to\cM_*$ with respect to its decoherence free states $\cD(E_\cN)$. We say $\Phi$ satisfies a $\al$-strong data processing inequality ($\al$-SDPI) for some $0<\al<1$ if for any state $\rho\in
 \cD(\cH)$,
\begin{align}\label{SDPI1}
D(\Phi(\rho)\|\Phi\circ E_{\cN*}(\rho))\le \al\,D(\rho\|E_{\cN*}(\rho))\,.
\end{align}
We say $\Phi$ satisfies the $\al$-complete strong data processing inequality ($\al$-CSDPI) for some $0<\al<1$ if for any $n\in\mathbb{N}$ and all bipartite states $\rho\in \cD(\mathbb{M}_n(\cM))$:
 \begin{align}\label{CSDPI2}
D((\Phi\otimes\id_n)(\rho)\|(\Phi\circ E_{\cN*}\otimes \id_n)(\rho))\le \al\,D(\rho\|(E_{\cN*}\otimes\id_n)(\rho))\,,
\end{align}
where $\id_n$ denotes the identity channel on the matrix algebra $\mathbb{M}_n$. We denote the best (smallest) constant achieving SDPI \eqref{SDPI1} (resp. CSDPI \eqref{SDPI1}) as $\al_{\operatorname{SDPI}}(\Phi)$ (resp. as $\al_{\operatorname{CSDPI}}(\Phi)$).
The advantage of the CSDPI constant is that it is stable under tensorization.

\begin{proposition}\label{prop:CSDPItensor}
Let $\Phi_1:\cM_{1*}\to\cM_{1*}$ and $\Phi_2:\cM_{2*}\to\cM_{2*}$ be two quantum channel. Denote $E_j:\cM_j\to \cN_j, j=1,2$ as the condition expectation onto the multiplicative domain of $\Phi_j^*$ respectively. Then
\[\al_{\operatorname{CSDPI}}(\Phi_1\ten \Phi_2)\le \max\{\al_{\operatorname{CSDPI}}(\Phi_1), \al_{\operatorname{CSDPI}}(\Phi_2)\}\pl. \]
Namely, for any $n\ge 1$ and states $\rho\in \cD(\cM_1\ten \cM_2\ten \mathbb{M}_n)$
\begin{align}\label{eq:tensorCSDPI}
 D(\Phi_1\ten \Phi_2\ten \id_{\mathbb{M}_n}(\rho)&\| (\Phi_1\circ E_{1*})\ten (\Phi_2\circ E_{2*})\ten \id_{\mathbb{M}_n}(\rho))\\& \le
\max\{\al_{\operatorname{CSDPI}}(\Phi_1), \al_{\operatorname{CSDPI}}(\Phi_2)\}D(\rho\| E_{1*}\ten E_{2*}\ten \id_{\mathbb{M}_n}(\rho))\,.\nonumber
\end{align}
\end{proposition}
\begin{proof}
The proof is a natural application of the data processing inequality. For ease of notations, we argue for $n=1$ as the case for general $n\ge 1$ follows the same argument. Note that for $j=1,2$,  $\Phi_j\circ E_{j*}=  E_{j*}\circ \Phi_j$. Write $\al_1:=\al_{\operatorname{CSDPI}}(\Phi_1)$ and $\al_2:=\al_{\operatorname{CSDPI}}(\Phi_2)$.  We have
\begin{align*}
D\Big(\Phi_1\ten \Phi_2(\rho)&\| (\Phi_1\circ E_{1*})\ten (\Phi_2\circ E_{2*})(\rho)\Big)
\\ = &D\Big(\Phi_1\ten \Phi_2(\rho)\| E_{1*}\otimes  E_{2*}\big(\Phi_1\ten \Phi_2(\rho)\big)\Big)
\\ = &D\Big(\Phi_1\ten \Phi_2(\rho)\|  E_{1*}\Phi_1\ten \Phi_2(\rho)\Big)+ D\Big(E_{1*}\Phi_1\ten \Phi_2(\rho)\| E_{1*}\otimes  E_{2*}\big(\Phi_1\ten \Phi_2(\rho)\big)\Big)
\\ \le & \al_1 D\Big(\id\ten \Phi_2(\rho)\|  E_{1*}\ten \Phi_2(\rho)\Big)
+\al_2 D\Big((E_{1*}\circ\Phi_1)\ten \id (\rho)\| (E_{1*}\circ\Phi_1)\ten E_{2*} (\rho)\Big)
% \\ \le & D\Big(\Phi_1\ten \Phi_2(\rho)\| (\Phi_1\circ E_{1*})\ten (\Phi_2\circ E_{2*})(\rho)\Big)
% \\ = &D\Big(\Phi_1\ten \Phi_2(\rho)\| E_{1*}\otimes E_{2*}\big(\Phi_1\ten \Phi_2(\rho)\big)\Big)
% \\ = &D\Big(\Phi_1\ten \Phi_2(\rho)\|  E_{1*}\big(\Phi_1\ten \Phi_2(\rho)\big)\Big)+ D\Big(E_{1*}\big(\Phi_1\ten \Phi_2(\rho)\big)\| E_{1*}\otimes  E_{2*}\big(\Phi_1\ten \Phi_2(\rho)\big)\Big)
% \\  \le &\al_1 D\Big(\id\ten \Phi_2(\rho)\|  E_{1*}\ten \Phi_2(\rho)\Big)
% +\al_2 D\Big((E_{1*}\circ\Phi_1)\ten \id (\rho)\| (E_{1*}\circ\Phi_1)\ten E_{2*} (\rho)\Big)
 \\ \le & \al_1 D\Big(\rho\|  E_{1*}\ten \id(\rho)\Big)
 +\al_2 D\Big(E_{1*}\ten \id (\rho)\| E_{1*}\ten E_{2*} (\rho)\Big)
\\ \le & \max\{\al_1, \al_2 \}D\Big(\rho\| E_{1*}\ten E_{2*} (\rho)\Big)
\end{align*}
where in the second equality and the last inequality, we used the chain rule \eqref{chain rule} and the second last inequality uses data processing inequality for the map $\id\ten \Phi_2$ and $\Phi_1\ten \id$ respectively.
\end{proof}

As an application of Theorem \eqref{thm:local}, we have $\al_{\operatorname{SDPI}}(\Phi)$ and $\al_{\operatorname{CSDPI}}(\Phi)$ are two-sided bounded by the spectral gap in finite dimensions.

\begin{corollary}\label{cor:SDPI}
Let $\Phi:\cM_*\to\cM_*$ be a quantum channel and $\cN$ be the multiplicative domain of $\Phi^*$. Assume that $\Phi^*$ satisfies the $\si$-$\operatorname{DBC}$ for some full-rank invariant state $\si=\Phi(\si)$. Denote the spectral gap $\la(\Phi):=\norm{\Phi^*(\id-E_{\cN}):L_2(\si)\to L_2(\si)}{}^2<1$. There exists an explicit constant $ c\,(C_{\tau,\operatorname{cb}}(\cM:\cN),\lambda)<1$ such that
\begin{align} \la(\Phi)\le \al_{\operatorname{CSDPI}}(\Phi) \le c\,(C_{\tau,\operatorname{cb}}(\cM:\cN),\lambda(\Phi))\,. \label{eq:CSDPI}\end{align}
The same estimate holds for $\al_{\operatorname{SDPI}}(\Phi)$ simply replacing $C_{\tau,\operatorname{cb}}(\cM:\cN)$ by $C_{\tau}(\cM:\cN)$.
\end{corollary}
\begin{proof} We have shown in Lemma \ref{lemma:independence} that $\la(\Phi)=\la(\si)$ and $\la(\si)\ge \la(\rho)$ for all decoherence free state $\omega\in \cD(E_{\cN*})$. Then \eqref{eq:CSDPI} follows from Theorem \ref{thm:local} and the fact that $\rho\le C_{\tau,\operatorname{cb}}(\cM:\cN) E_{\cN*}(\rho)$ for any $\rho\in \cD(\cH\otimes \mathbb{C}^n)$.
\end{proof}

\begin{rem}
For a primitive unital quantum channel $\Phi:\cB(\cH)\to \cB(\cH)$, it was proved in \cite{muller2016entropy} that
\[\al_{\operatorname{SPDI}}(\Phi)\le 1-\al_{\operatorname{LSI}}(\Phi^*\Phi-\id)\le 1-\frac{\la(\Phi)}{\ln d+2} \pl, \]
where $\al_{\operatorname{LSI}}(\Phi^*\Phi-\id)$ is the log-Sobolev constant of the map $\Phi^*\Phi-\id$ seen as the generator of a quantum Markov semigroup.
This is generically better than the bounds found in Corollary \ref{cor:SDPI}. Nevertheless, our results give explicit SDPI constants for general non-egordic GNS-symmetric quantum channels, independently of the size of the environment. Moreover, the CSDPI constant satisfies the tensorization property.
\end{rem}

\section{Approximate tensorization}\label{sec:AT}
In this section, we consider the approximate tensorization of the relative entropy in a general setting. Let $\cM$ be a finite dimensional von Neumann algebra equipped with a faithful trace $\tr$. Let $\cN_1,\cN_2\subset \cM$ be two subalgebras of $\cM$ and $\cN=\cN_1 \cap \cN_2$. Let $E_\cN:\cM\to \cN$ and $E_i:\cM\to \cN_i,i=1,2$, be conditional expectations such that $E_{\cN}\circ E_i=E_{\cN}$. If $\rho$ is a state that satisfies $E_{\cN*}(\rho)=\rho$, then
\[ \rho=E_{\cN*}(\rho)=E_{i*}\circ E_{\cN*}(\rho) =E_{i*}(\rho) \pl, i=1,2\pl.\]
Namely, every $E_{\cN}$ invariant state is both $E_1$ and $E_2$ invariant. Denote $\rho_\cN=E_{\cN*}(\rho)$ and $\rho_i=E_{i*}(\rho), i=1,2$. We are interested in the following \textit{approximate tensorization} property:
\begin{align}
    D(\rho\|\rho_\cN)\le c\,(D(\rho\|\rho_1)+D(\rho\|\rho_2))\pl, \pl \forall \rho\in \cD(E_\cM).  \label{eq:AT}
\end{align}
It was proved in \cite[Corollary 2.3]{gao2017unifying} that the constant $c$ equals to 1 if and only if $E_1$ and $E_2$ form a commuting square, i.e. $E_1\circ E_2=E_2\circ E_1=E_{\cN}$.
Using the chain rule $D(\rho\|\rho_\cN)=D(\rho\|\rho_i)+D(\rho_i\|\rho_\cN)$, the inequality \eqref{eq:AT} is equivalent to
the following entropic uncertainty relation
\begin{align}
    D(\rho\|\rho_\cN)\ge \alpha (D(\rho_1\|\rho_\cN)+D(\rho_2\|\rho_\cN))\pl, \pl \forall \rho\in \cD(E_\cM). \label{eq:UCR}
\end{align}
where $\alpha= \displaystyle \frac{c}{2c-1}>1/2$.
Take $\rho(t)=t\rho+(1-t)\rho_{\cN}$ and the function
\[ f(t)=D(\rho(t)\|\rho_\cN)- \alpha\big(\,D(\rho_1(t)\|\rho_\cN)+D(\rho_2(t)\|\rho_\cN)\,\big)\,.\]
Then we have $f(0)=f'(0)=0$ and
\begin{align*}
f''(0)=\norm{\rho-\rho_\cN}{\rho_{\cN}^{-1}}^2-\alpha\big(\norm{\rho_1-\rho_\cN}{\rho_{\cN}^{-1}}^2+\norm{\rho_2-\rho_\cN}{\rho_{\cN}^{-1}}^2\big)\pl.
\end{align*}
So a necessary condition for \eqref{eq:UCR} and equivalently \eqref{eq:AT} is that for any state $\rho$,
\[ \norm{\rho-\rho_\cN}{\rho_{\cN}^{-1}}^2\ge \alpha\big( \norm{\rho_1-\rho_\cN}{\rho_{\cN}^{-1}}^2+\norm{\rho_2-\rho_\cN}{\rho_{\cN}^{-1}}^2\big)\pl.\]
In particular, if we choose $\rho=\rho_1=E_{1*}(\rho)$, we have
\[  \frac{(1-\alpha)}{\alpha}\norm{\rho_1-\rho_\cN}{\rho_{\cN}^{-1}}^2\ge \norm{E_{2*}(\rho_1)-\rho_\cN}{\rho_{\cN}^{-1}}^2  \pl.\]
Because $1/2< \alpha \le 1$, for $\la=\frac{1-\al}{\al}$
this can be reformulated as the following $L_2$-clustering condition
\[ \norm{E_{2*}\circ E_{1*} -E_{\cN*}:L_2(\rho_{\cN}^{-1})\to L_2(\rho_{\cN}^{-1})}{}=\norm{E_{1}\circ E_{2} -E_{\cN}:L_2(\rho_{\cN})\to L_2(\rho_{\cN})}{}:=\la <1\pl.\]
Since $E_{2}\circ E_{1}$ is the identity on $\cN$ and satisfies the $\rho_{\cN}$-DBC condition, the above definition is independent of the choice of invariant state $\rho_\cN$ (see Lemma \ref{lemma:independence}, also \cite[Theorem 2]{bardet2020approximate}). Note that in finite dimensions, the constant $\la$ is always strictly less than $1$: otherwise there would exist a nonzero $X\notin \cN$ such that $E_{1}(X)=X,E_2(X)=X$ and hence $X\in \cN$, which leads to a contradiction. We now show that the $L_2$-clustering condition is also a sufficient condition for \eqref{eq:AT}.
\begin{theorem}
\label{thm:AT}
Let $\si\in \cD(E_\cN)$.
Denote $\norm{E_{1}\circ E_{2} -E_{\cN}:L_2(\si)\to L_2(\si)}{}=\la<1$ as the $L_2$-clustering constant. Then for any state $\rho$,
\begin{align}\label{eq:ATq} D(\rho\|\rho_\cN)\le c\big(D(\rho\|\rho_1)+D(\rho\|\rho_2)\big)\pl,
\end{align}
where the constant $c$ satisfies
\begin{align}\label{noncompleteATq}
\frac{1}{1-\lambda^2}\le c\le \frac{2\,C_{\tau}(\cM:\cN)}{(1-\la)^2}\,.
\end{align}
Similarly, for any $n\in\mathbb{N}$ and all states $\rho\in\cD(\cM\otimes \mathbb{M}_n)$, we have
\begin{align}\label{completeATq}
    D(\rho\|(E_{\cN*}\otimes\id)(\rho))\le \,c_{\operatorname{cb}}\big(D(\rho\|(E_{1*}\otimes \id)(\rho))+D(\rho\|(E_{2*}\otimes \id)(\rho))\big)
\end{align}
where $c_{\operatorname{cb}}$ satisfies \eqref{noncompleteATq} after replacing $C_{\tau}(\cM:\cN)$ by $C_{\tau,\operatorname{cb}}(\cM:\cN)$.
\end{theorem}
\begin{proof}
 The lower bound was proven at the beginning of the section, so we focus on the upper bound. Note that $E_1,E_2$ and $E_\cN$ are all projections on $L_2(\rho_{\cN})$. For a state $\rho$, we write $\rho_{12}=E_{1*} E_{2*}(\rho)$ and $\rho_{12}=E_{2*} E_{1*}(\rho)$. By the $L_2$-clustering condition
\begin{align}\nonumber
\norm{\rho-\rho_{21}}{\rho_\cN^{-1}}&\ge \norm{\rho-\rho_\cN}{\rho_{\cN}^{-1}}-\norm{\rho_\cN-\rho_{21}}{\rho_{\cN}^{-1}}\\
&\ge (1-\la)\norm{\rho-\rho_\cN}{\rho_{\cN}^{-1}}\pl\,.\label{eqgapE2E1}
\end{align}
Moreover, since $E_{1*},E_{2*}$ and $E_{\cN*}$ are projections on $L_2(\rho_\cN^{-1})$,
\begin{align*}
\norm{\rho-\rho_{\cN}}{\rho_{\cN}^{-1}}^2-\norm{\rho&-\rho_{1}}{\rho_{\cN}^{-1}}^2
-\norm{\rho-\rho_{2}}{\rho_{\cN}^{-1}}^2
\\ \le & \norm{\rho-\rho_{\cN}}{\rho_{\cN}^{-1}}^2-\norm{\rho-\rho_{2}}{\rho_{\cN}^{-1}}^2
-\norm{\rho_2-E_{2*}(\rho_{1})}{\rho_{\cN}^{-1}}^2
  \\ \le&\norm{\rho-\rho_{\cN}}{\rho_{\cN}^{-1}}^2-\norm{\rho-E_{2*}(\rho_1)}{\rho_{\cN}^{-1}}^2
\\ = & \norm{\rho-\rho_{\cN}}{\rho_{\cN}^{-1}}^2-\norm{\rho-\rho_{21}}{\rho_{\cN}^{-1}}^2
\\ \le & (1-(1-\la)^2)\norm{\rho-\rho_\cN}{\rho_{\cN}^{-1}}^2\,,
\end{align*}
where the last line follows from \eqref{eqgapE2E1}. Namely, we have
\[ \norm{\rho-\rho_{\cN}}{\rho_{\cN}^{-1}}^2\le \frac{1}{(1-\la)^2}\big(\norm{\rho-\rho_{1}}{\rho_{\cN}^{-1}}^2
+\norm{\rho-\rho_{2}}{\rho_{\cN}^{-1}}^2\big) \pl.\]
Now using Lemma \ref{lemma:keylemma},
\begin{align*}
D(\rho\|\rho_\cN)\le &\norm{\rho-\rho_\cN}{\rho_{\cN}^{-1}}^2
\\ \le & \, \frac{1}{(1-\la)^2}\big(\norm{\rho-\rho_{1}}{\rho_{\cN}^{-1}}^2
+\norm{\rho-\rho_{2}}{\rho_{\cN}^{-1}}^2\big)
\\ \le & \, \frac{C_\tau(\cM:\cN)}{(1-\la)^2} \big(\norm{\rho-\rho_{1}}{\rho_1(t)^{-1}}^2
+\norm{\rho-\rho_{2}}{\rho_2(t)^{-1}}^2\big)\,,
\end{align*}
where $\rho_1(t)=t\rho+(1-t)\rho_1$ and $\rho_2(t)=t\rho+(1-t)\rho_2$.
As in Lemma \ref{lemma:keylemma}, for $i=1,2$
$\displaystyle D(\rho\|\rho_i)=\int_0^1 \int_{0}^s \norm{\rho-\rho_i}{\rho_i(t)}^2dtds$. Then integrating the above inequality we have
\[ D(\rho\|\rho_\cN)\le \frac{2C_\tau(\cM:\cN)}{(1-\la)^2} \big(D(\rho\|\rho_1)+D(\rho\|\rho_2)\big)\pl.\]
That completes the proof of \eqref{noncompleteATq}. The proof of \eqref{completeATq} follows the exact same lines after replacing $C_{\tau}(\cM:\cN)$ by $C_{\tau,\operatorname{cb}}(\cM:\cN)$\,.
\end{proof}

\begin{rem}{\rm By using $\rho_{i}(t)=t\rho+(1-t)\rho_i\le \big(tC(\cM:\cN)+(1-t)C_\tau(\cN_i:\cN)\big)\rho_\cN$,
the constant $c$ in the above theorem can be improved to
\[ c=\frac{K(C_\tau(\cM:\cN),\max\{C_\tau(\cN_1:\cN),C_\tau(\cN_2:\cN)\})}{(1-\la)^2}\pl,\]
where $\displaystyle K(c_1,c_2):=\frac{c_1\ln c_1-c_1+c_2}{(c_1-c_2)^2}$\,. The same remark holds for $c_{\operatorname{cb}}$.
}
\end{rem}

Although the above theorem gives the equivalence of $L_2$-clustering condition and complete approximate tensorization, it does not recover the optimal constant $c=1$ in the case of commuting square ($\lambda=0$). The next theorem gives a refinement in this direction.

\begin{theorem}\label{theo:lambdaindexcontrolAT}
Let $\si\in \cD(E_\cN)$ and denote
$\norm{E_{1}\circ E_{2} -E_{\cN}:L_2(\si)\to L_2(\si)}{}=\la$. Suppose $\la<\frac{1}{\sqrt{2}}$.
Then the (complete) approximate tensorization \eqref{eq:ATq} and \eqref{completeATq} are satisfied with the constants
 \begin{align}
 &c\le 1+\Big(\frac{\la}{1-\la}+\frac{\la^2}{1-2\la^2}\Big)\,C_{\max}\pl, \label{eq:lasmall}\\
 &c_{\operatorname{cb}}\le 1+\Big(\frac{\la}{1-\la}+\frac{\la^2}{1-2\la^2}\Big)\,C_{\max,\operatorname{cb}}\pl,\nonumber
 \end{align}
 where \begin{align*}&C_{\max}:=\max\{ C_\tau(\cN_1:\cN)\pl, C_\tau(\cN_2:\cN),C_{\tau}(\cM:\cN_1),C_{\tau}(\cM:\cN_2)\}\ , \\ &C_{\max,\operatorname{cb}}:=\max\{ C_{\tau,\operatorname{cb}}(\cN_1:\cN)\pl, C_{\tau,\operatorname{cb}}(\cN_2:\cN),C_{\tau,\operatorname{cb}}(\cM:\cN_1),C_{\tau,\operatorname{cb}}(\cM:\cN_2)\}\,.
 \end{align*}
 In particular, the above bound converges to $1$ in the limit of commuting squares, i.e. when $\la\to 0$.
\end{theorem}
\begin{proof}
Denote $\rho_1:=E_{1*}(\rho)$, $\rho_2:=E_{2*}(\rho)$, $\rho_{12}:=E_{1*}\circ E_{2*}(\rho)$ and $\rho_{21}:=E_{2*}\circ E_{1*}(\rho)$. First, by chain rule (\cite[Lemma 3.4]{junge2019stability}), we have
\begin{align*}
D(\rho\|\rho_\cN)&=D(\rho\|\rho_1)+D(\rho_1\|\rho_\cN)\\
&=D(\rho\|\rho_1)+D(\rho_1\|\rho_{12})+\tr(\rho_{1}(\ln\rho_{12}-\ln\rho_{\cN}))\,,
% &\le D(\rho\|\rho_1)+D(\rho_1\|\rho_{12})+\tr(\rho_{1}(\ln\rho_{12}-\ln\rho_{\cN}))\,,
\end{align*}
and similarly
\begin{align*}
D(\rho\|\rho_\cN)&=D(\rho\|\rho_2)+D(\rho_2\|\rho_\cN)\\
&=D(\rho\|\rho_2)+D(\rho_2\|\rho_{21})+\tr(\rho_{2}(\ln\rho_{21}-\ln\rho_{\cN}))\,,
% &\le D(\rho\|\rho_2)+D(\rho_2\|\rho_{21})+\tr(\rho_{2}(\ln\rho_{21}-\ln\rho_{\cN}))
\end{align*}
It suffices to estimate the error term $ \tr(\rho_{1}(\ln\rho_{12}-\ln\rho_{\cN}))$ and $\tr(\rho_{2}(\ln\rho_{21}-\ln\rho_{\cN}))$.
Recall the integral identity that for positive $A,B> 0$
\[\ln A-\ln B=\int_{0}^\infty \frac{1}{A+s}(A-B)\frac{1}{B+s}\,ds\pl.\]
Thus by Cauchy-Schwarz inequality and Lemma \ref{lemma:keylemma}:
\begin{align*}
\tr(\rho_{1}(\ln\rho_{12}-\ln\rho_{\cN}))=&\tr((\rho_{1}-\rho_{12})(\ln\rho_{12}-\ln\rho_{\cN}))+\tr(\rho_{12}(\ln\rho_{12}-\ln\rho_{\cN}))
\\=&\int_{0}^\infty \tr((\rho_1-\rho_{12})\frac{1}{\rho_{12}+s}(\rho_{12}-\rho_\cN)\frac{1}{\rho_\cN+s})\, ds+D(\rho_{12}\|\rho_\cN)
\\ \le & \norm{\rho_1-\rho_{12}}{\rho_{12}^{-1}}\norm{\rho_{12}-\rho_{\cN}}{\rho_{\cN}^{-1}}+\norm{\rho_{12}-\rho_{\cN}}{\rho_{\cN}^{-1}}^2\,.
\end{align*}
Similarly
\[\tr(\rho_{2}(\ln\rho_{21}-\ln\rho_{\cN}))\le \, \norm{\rho_2-\rho_{21}}{\rho_{21}^{-1}}\norm{\rho_{21}-\rho_{\cN}}{\rho_{\cN}^{-1}}+\norm{\rho_{21}-\rho_{\cN}}{\rho_{\cN}^{-1}}^2\,.\]
Note that by the $L_2$-clustering condition
\begin{align*}&\norm{\rho_{12}-\rho_{\cN}}{\rho_{\cN}^{-1}}\le \la \norm{\rho_2-\rho_\cN}{\rho_{\cN}^{-1}}
\le \la(\norm{\rho_2-\rho_{21}}{\rho_{\cN}^{-1}}+\norm{\rho_{21}-\rho_\cN}{\rho_{\cN}^{-1}} )\,,
\\
&\norm{\rho_{21}-\rho_{\cN}}{\rho_{\cN}^{-1}}\le \la \norm{\rho_1-\rho_\cN}{\rho_{\cN}^{-1}}\le \la(\norm{\rho_1-\rho_{12}}{\rho_{\cN}^{-1}}+\norm{\rho_{12}-\rho_\cN}{\rho_{\cN}^{-1}})\,.
\end{align*}
Thus
\begin{align*} &\norm{\rho_{12}-\rho_{\cN}}{\rho_{\cN}^{-1}}^2\le 2\la^2(\norm{\rho_2-\rho_{21}}{\rho_{\cN}^{-1}}^2+\norm{\rho_{21}-\rho_\cN}{\rho_{\cN}^{-1}}^2)\pl,\\
&  \norm{\rho_{21}-\rho_{\cN}}{\rho_{\cN}^{-1}}^2\le 2\la^2(\norm{\rho_1-\rho_{12}}{\rho_{\cN}^{-1}}^2+\norm{\rho_{12}-\rho_\cN}{\rho_{\cN}^{-1}}^2)
\end{align*}
Therefore for $\la< \frac{1}{\sqrt{2}}$, by Lemma \ref{lemma:keylemma2}:
\begin{align*}
\norm{\rho_{12}-\rho_{\cN}}{\rho_{\cN}^{-1}}^2+\norm{\rho_{21}-\rho_{\cN}}{\rho_{\cN}^{-1}}^2\le & \frac{2\la^2}{1-2\la^2}\big( \norm{\rho_1-\rho_{12}}{\rho_{\cN}^{-1}}^2+\norm{\rho_2-\rho_{21}}{\rho_{\cN}^{-1}}^2\big)
\\  \le& \frac{2\la^2}{1-2\la^2}\big( 2C_1 D(\rho_1\|\rho_{12})+ 2C_2 D(\rho_2\|\rho_{21})\big)\,,
\end{align*}
where $C_1=C_\tau(\cN_1:\cN)$, and $C_2=C_\tau(\cN_2:\cN)$. On the other hand, denoting $$M:=\max\{\sqrt{2C_1 D(\rho_1\|\rho_{12})}, \sqrt{2C_2D(\rho_2\|\rho_{21})}\}\,,$$
we have
\begin{align*}
\norm{\rho_1-\rho_{12}}{\rho_{12}^{-1}}&\norm{\rho_{12}-\rho_{\cN}}{\rho_{\cN}^{-1}}+\norm{\rho_2-\rho_{21}}{\rho_{21}^{-1}}\norm{\rho_{21}-\rho_{\cN}}{\rho_{\cN}^{-1}}
\\ &\le M\,(\norm{\rho_{12}-\rho_{\cN}}{\rho_{\cN}^{-1}}+
\norm{\rho_{21}-\rho_{\cN}}{\rho_{\cN}^{-1}})
\\ &\le M\,\frac{\la}{1-\la}\,
(\norm{\rho_{2}-\rho_{21}}{\rho_{\cN}^{-1}}+
\norm{\rho_{1}-\rho_{12}}{\rho_{\cN}^{-1}})
\\ &\le M\,\frac{\la}{1-\la}\,
(\sqrt{2C_2D(\rho_2\|\rho_{21})}+
\sqrt{2C_1D(\rho_1\|\rho_{12})})
\\ &\le \frac{2\la}{1-\la} C_{\max}\,\big(D(\rho\|\rho_1)+D(\rho\|\rho_2)\big)\,,
\end{align*}
where $C_{\max}:=\max \{C_\tau(\cN_1:\cN),C_\tau(\cN_2:\cN),C_{\tau}(\cM:\cN_1),C_{\tau}(\cM:\cN_2)\}$.
Therefore, we have
\begin{align*}
2D(\rho\|\rho_\cN)-D(\rho\|\rho_1)&-D(\rho\|\rho_2)-
D(\rho_{1}\|\rho_{12})-D(\rho_2\|\rho_{21})
\\ &= \tr(\rho_{1}(\ln\rho_{12}-\ln\rho_{\cN}))+\tr(\rho_{2}(\ln\rho_{21}-\ln\rho_{\cN}))
\\& \le
\norm{\rho_1-\rho_{12}}{\rho_{12}^{-1}}\norm{\rho_{12}-\rho_{\cN}}{\rho_{\cN}^{-1}}+\norm{\rho_2-\rho_{21}}{\rho_{21}^{-1}}\norm{\rho_{21}-\rho_{\cN}}{\rho_{\cN}^{-1}}\\
& \ \ \ \ \
+\norm{\rho_{12}-\rho_{\cN}}{\rho_{\cN}^{-1}}^2+\norm{\rho_{21}-\rho_{\cN}}{\rho_{\cN}^{-1}}^2
\\& \le  \frac{2\la}{1-\la}C_{\max}\big( D(\rho\|\rho_{1})+D(\rho\|\rho_{2})\big)
+\frac{2\la^2}{1-2\la^2}\big( C_1 D(\rho_1\|\rho_{12})+ C_2 D(\rho_2\|\rho_{21})\big)\\
&\le \Big(\frac{2\la^2}{1-2\la^2}+ \frac{2\la}{1-\la}\Big)C_{\max} \big( D(\rho\|\rho_{1})+D(\rho\|\rho_{2})\big)\,.
\end{align*}
The result follows after rearranging the terms in the outer bounds and a last use of the data processing inequality. The proof for $c_{\operatorname{cb}}$ follows the same strategy after replacing $C_{\max}$ by $C_{\max,\operatorname{cb}}$.

\end{proof}

In the classical literature \cite{cesi2001quasi,dai2002entropy}, approximate tensorization constants were found under the strong condition of smallness of the norm $\norm{E_1\circ E_2-E_\cN:L_1\to L_\infty}{}$  instead of the $L_2$-condition $\norm{E_1\circ E_2-E_\cN:L_2\to L_2}{}$ that we use. In that setting, the approximate tensorization constant obtained in \Cref{theo:lambdaindexcontrolAT} is not tight because the Pimsner-Popa indices coincide with the dimension bounds for the $L_1\to L_\infty$ norm. Quantum extensions using $L_1\to L_\infty$ cluster condition were recently found in \cite{bardet2020approximate}, however they yield additive error terms in generic noncommutative situations, e.g. when the algebra $\cN$ is not trivial. This generalization however was found fruitful in deriving the positivity of the MLSI constant for some classes of Gibbs samplers in \cite{capel2020modified}, where the multiplicative constant could be related to the notion of clustering of correlations in the equilibrium Gibbs state. There, the analysis could be reduced to the case of states $\rho$ for which the additive error vanishes. However, the problem of the vanishing of the additive constant for general states remained open.

After the preprint submission of a preliminary version of the present paper, \cite{laracuente2019quasi} introduced a method based on our Lemma \ref{lemma:keylemma}
to find asymptotically tight approximate tensorization constants. %His method recovers exact tensorization for commuting conditional expectations, and more generally applies to uncertainty relations.
% The main technical contribution of \cite{laracuente2019quasi} resides in an application of the map $(E_1\circ E_2)^k$ in order to control the smallness of $L_1\to L_\infty$ condition via $\|(E_1\circ E_2)^k-E_\cN:L_2(\sigma)\to L_2(\sigma)\|=\lambda^k$, at the cost of having to multiply the approximate tensorization constant by $k$. Our next two results can be interpreted as a merge of these different contributions to solve the problem of tightness of the approximate tensorization constant. We first start with a meta-theorem which almost generalizes the original results of \cite{cesi2001quasi,dai2002entropy} without the additive error terms found in \cite{bardet2020approximate}, at the cost of having to replace the $L_1\to L_\infty$ condition by a stronger quantum generalization of it. Although the condition is closer to that of \cite{laracuente2019quasi}, our proof is arguably more straightforward and resembles the ones of \cite{cesi2001quasi,dai2002entropy,bardet2020approximate,capel2020modified}. Moreover, our bound can be used to re-derive the asymptotically tight bound on the MLSI constant for graph Laplacians found in \cite{laracuente2019quasi}.
% We recall the definition of the function $k$ in Lemma \ref{lemma:keylemma}:
% \begin{align*}
%     k(c):=\frac{c\ln c-c+1}{(c-1)^2}\,.
% \end{align*}
% In particular, $k(2)=2\ln 2 -1$.
\begin{theorem}\label{theorem:meta}
Let $\cN\subset \cM$ be a finite dimensional von Neumann subalgebra and $E_\cN:\cM\to \cN$ be a conditional expectation.
Let $\Phi:\cM_*\to \cM_*$ be a quantum channel such that $\Phi^*$ is $\operatorname{GNS}$-symmetric to a full-rank $E_\cN$-invariant states and satisfies
$\Phi^*\circ E_\cN=E_\cN\circ \Phi^*=E_\cN$. Suppose for some $0<\eps<\sqrt{2\ln 2-1}$,
\begin{align*}
(1-\eps)E_{\cN}\le_{\operatorname{cp}} \Phi^*\le_{\operatorname{cp}} (1+\eps) E_{\cN}\,,
\end{align*}
where the inequalities hold in completely positive order. Then, for all $n\in \mathbb{N}$ and states $\rho\in \cD(\cM\ten \mathbb{M}_n)$:
\begin{align}\label{eq:metatheorem}
    D(\rho\| E_{\cN*}(\rho))\le \frac{1}{1-\eps^2(2\ln 2-1)^{-1}}\,D(\rho\|\Phi^2(\rho))\ .
\end{align}
\end{theorem}
\begin{proof}
Let $\rho\in\cD(\cM\otimes \mathbb{M}_n)$ and $\rho_\cN:=E_{\cN*}(\rho)$. Then,
\begin{align*}
    D(\rho\|\rho_\cN)-D(\rho\| \Phi^2(\rho))=\tr\big[\rho\,(-\ln(\rho)+\ln A) \big]=-D\Big(\rho\Big\|\frac{A}{\tr(A)}\Big)+\ln\tr(A)\le \ln\tr(A)\,,
\end{align*}
for $A:=\exp(\ln \Phi^2(\rho)-\ln\rho_\cN+\ln\rho)$. Here the last inequality follows from the positivity of the relative entropy. Using Lieb's triple matrix inequality (see \cite[Theorem 7]{lieb1973convex}),
\begin{align*}
     \ln\tr(A)\le\ln\int_{0}^{\infty}\,\tr\Big( \Phi^2(\rho)\frac{1}{\rho_\cN+s}\,\rho\,\frac{1}{\rho_\cN+s} \Big)\,ds\ .
\end{align*}
Then by the GNS-symmetry of $\Phi^*$,
\begin{align*}
  D(\rho\|\rho_\cN)&\le D(\rho\| \Phi^2(\rho))+\ln\int_{0}^{\infty}\,\tr\Big( \Phi(\rho)\frac{1}{\rho_\cN+s}\,\Phi(\rho)\,\frac{1}{\rho_\cN+s} \Big)\,ds\\
  &\overset{(1)}{\le} D(\rho\| \Phi^2(\rho))+\int_{0}^{\infty}\,\tr\Big( \big(\Phi(\rho)-E_{\cN*}(\rho)\big)\frac{1}{\rho_\cN+s}\, \big(\Phi(\rho)-E_{\cN*}(\rho)\big)\,\frac{1}{\rho_\cN+s} \Big) \,ds\,,
\end{align*}
where $(1)$ arises from the basic inequality $\ln(x)\le x-1$ and the trace preserving property of $\Phi$ and $E_{\cN*}$. Now, since $\Phi^*\ge_{\operatorname{cp}} (1-\eps)E_{\cN}$, there exists a quantum channel $\Psi:\cM_*\to\cM_*$ such that $\Phi=(1-\eps)E_{\cN*}+\eps\Psi$.
% Moreover, since $\Phi^*\circ E_\cN=E_\cN\circ \Phi^*=E_\cN$, the same identity holds for $\Psi^*$.
Therefore,
\begin{align*}
    D(\rho\|\rho_\cN)&\le D(\rho\| \Phi^2(\rho))+\eps^2\,\int_{0}^{\infty}\,\tr\Big( (\Psi(\rho)-E_{\cN*}(\rho))\frac{1}{\rho_\cN+s}\, (\Psi(\rho)-E_{\cN*}(\rho))\,\frac{1}{\rho_\cN+s} \Big) \,ds\,\\
    &= D(\rho\| \Phi^2(\rho))+\eps^2\,\| (\Psi-E_{\cN*})(\rho)\|_{\rho_\cN^{-1}}^2\\
    &\overset{(2)}{\le}  D(\rho\| \Phi^2(\rho))+\eps^2\,k(2)^{-1}\,D(\Psi(\rho)\|\rho_\cN)\\
    &{\le}  D(\rho\| \Phi^2(\rho))+\eps^2\,k(2)^{-1}\,D(\rho\|\rho_\cN)\,,
\end{align*}
where $(2)$ comes from Lemma \ref{lemma:keylemma} and the fact that $\Phi^*\le (1+\eps)E_{\cN}$ so that $$\Psi(\rho)\le \eps^{-1}(1+\eps-(1-\eps))\rho_\cN=2\rho_\cN.$$ The result follows after rearranging the terms in the last inequality.
\end{proof}

The above theorem can be used to derive approximate tensorization bounds. For instance, two natural choices of the map $\Phi^*$ are either $\frac{1}{m}\sum_{i=1}^mE_i$ or $\frac{1}{2}\big(\prod_{i=1}^m E_i+\prod_{i=m}^1 E_i\big)$, for which we find the following:

\begin{corollary}\label{Coro:sharperAT}
Let $\cN_1,\ldots,\cN_m\subset \cM$ be finite dimensional von Neumann subalgebras of $\cB(\cH)$, and let $\cN=\cap_{i=1}^m \cN_i$. Let $E_\cN:\cM\to \cN$ and $E_i:\cM\to \cN_i$ be some corresponding conditional expectations. Suppose for some full-rank $E_\cN$-invariant state $\sigma=E_{\cN*}(\sigma)$, $\si=E_{i*}(\si)$ for each $i$. Then for $\Phi^*=\frac{1}{m}\sum_{i=1}^mE_i$ or $\Phi^*=\frac{1}{2}\big(\prod_{i=1}^m E_i+\prod_{i=m}^1 E_i\big)$, we have that for all $n\in\mathbb{N}$ and all states $\rho\in\cD(\cM\otimes \mathbb{M}_n)$,
\begin{align}\label{eq:corollaryAT}
    D(\rho\|E_{\cN*}(\rho))\le \frac{2k}{1-\eps^2(2\ln(2)-1)^{-1}}\,\sum_{i=1}^m\,D(\rho\|E_{i*}(\rho))\,,
    \end{align}
whenever $0<\eps< \sqrt{2\ln(2)-1}$ and $k$ satisfies
\begin{align}
    (1-\eps)E_\cN\le_{\operatorname{cp}} (\Phi^*)^k\le_{\operatorname{cp}} (1+\eps)E_\cN \, \label{eq:orderinequality}.
\end{align}
If additionally $E_i$ and $E_\cN$ are trace preserving conditional expectations, we have that for all $n\in\mathbb{N}$ and all state $\rho\in\cD(\cM\otimes \mathbb{M}_n)$,
\begin{align}\label{eq:gapAT}
   D(\rho\|E_{\cN*}(\rho))\le 4\left\lceil\frac{\ln C_{\operatorname{cb}}(\cM:\cN)+1}{\ln(\lambda^{-1})}\right\rceil\,\sum_{i=1}^m\,D(\rho\|E_{i*}(\rho))\,,
\end{align}
where $\lambda:=\|\Phi-E_\cN:L_2(\tr)\to L_2(\tr)\|$ and $\lceil s \rceil$ denotes the smallest integer greater than or equal to $s$.
\end{corollary}
\begin{proof}
 \Cref{eq:corollaryAT} is a direct consequence of \Cref{theorem:meta} and successive applications of the data processing inequality and chain rule \cite[Lemma 3.2]{laracuente2019quasi},
\begin{align*}
    D(\rho\|\Phi^{2k}(\rho))&\le
 2 k\sum_{i=1}^m\,D(\rho\|E_{i*}(\rho))\,,
\end{align*}
%  that for any conditional expectation $E_j$ and any channel $\Psi$:
% \begin{align*}
%     D(\rho\|E_{j*}\circ \Psi(\rho))=D(\rho\|E_{j*}(\rho))+D(E_{j*}(\rho)\|E_{j*}\circ\Psi(\rho))\le D(\rho\|E_{j*}(\rho))+D(\rho\|\Psi(\rho))\,.
% \end{align*}
for both choices of $\Phi^*$. In the case when all $E_i$ and $E$ are trace preserving, we prove in \Cref{sec:normstoorder} that $\eps$ can be chosen as (see Lemma \ref{lemmageneral})
\begin{align*}
    \eps= \lambda^k C_{\operatorname{cb}}(\cM:\cN)
\end{align*}
for $k$ large enough so that the condition \eqref{eq:orderinequality} is satisfied.
Therefore, we can choose $\epsilon \le \sqrt{\ln 2-\frac{1}{2}}$ and $1-\eps^2(2\ln(2)-1)^{-1}\ge 1/2$ by taking
\[
k=\left\lceil (-\ln \la)^{-1}\Big(\ln C-\ln\sqrt{\ln 2-\frac{1}{2}} \Big) \right\rceil\le \left\lceil\frac{\ln C+1}{-\ln \la}\right\rceil
,\qquad C:=C_{\operatorname{cb}}(\cM:\cN)\,.\]
Then \eqref{eq:gapAT} follows from \eqref{eq:corollaryAT}.
\end{proof}

\begin{rem}
Although the bound \eqref{eq:corollaryAT} does not recover the exact tensorization for commuting conditional expectations, it has the merit over our other bounds to be independent of the index $C_{\tau,\operatorname{cb}}(\cM:\cN)$. In Sections \ref{sec:symmetric} and \ref{sec:localsemigroups}, we use the bound \eqref{eq:gapAT} to derive sharper CMLSI constant than \Cref{theo:CLSI} for several classes of examples.
\end{rem}

\begin{rem}
For the second inequality \eqref{eq:gapAT}, the assumption on trace preserving conditional expectation is not really necessarily. Indeed, the content in Appendix B can be extended to state preserving conditional expectation, hence also for the approximate tensorization inequality \eqref{eq:gapAT}. Nevertheless, the current form of \eqref{eq:gapAT} is sufficient for all examples discussed in this paper.
\end{rem}

\noindent {\bf Entropic Uncertainty relations. }% \label{sec:Uncertainty relations}
In the rest of this section, we apply the approximate tensorization bounds to the field of entropic uncertainty relations. We refer to \cite{Coles17} for a recent review on that topic. We first consider uncertainty relations for two complementary measurements: recall that a family of positive operator $\{X_x\}\subset \cB(\cH)$ is called a POVM (positive operator valued measurement) if $\sum_x X_x=\Id_{\cH}$.
Given two POVMs $\mathbf{X}:= \{X_x\}_{x}$ and $\mathbf{Y}:=\{Y_y\}_{y}$ on a quantum system $A=\cB(\cH_A)$, we denote by $\Phi_{\mathbf{X}}$ and $\Phi_{\mathbf{Y}}$ the quantum-classical channels for the measurement in $\mathbf{X}$ and $\mathbf{Y}$ respectively:
\begin{align*}
\Phi_{\mathbf{X}}(\rho):=\sum_{x\in \mathcal{X}}\,\tr(\rho_A X_x)\,|x\rangle\langle x|_X\ ,\ \ \ \  \Phi_{\mathbf{Y}}(\rho):=\sum_{y\in\mathcal{Y}}\,\tr(\rho_A Y_y)\,|y\rangle\langle y|_Y\,.\end{align*}It was found by Berta \emph{et al} \cite{berta2010uncertainty} that (for the special case of projective measurements), the presence of side information $M$ can help to better predict the outcomes of $\mathbf{X}$ and $\mathbf{Y}$, compared to the Maasen-Uffink uncertainty relation \cite{maasen1988uncertainty} in the memoryless setting.
Later, Frank and Lieb \cite{frank2013extended} further obtained the following state-dependent entropic uncertainty relation for two POVMs: for any bipartite state $\rho_{AM}\in\cD(\cH_A\otimes \cH_M)$,
\begin{align}\label{uncert}
S(X|M)_{(\Phi_{\mathbf{X}}\otimes \id_M)(\rho)}+S(Y|M)_{(\Phi_{\mathbf{Y}}\otimes \id_M)(\rho)}\ge- \ln c+S(A|M)_\rho\,,
\end{align}
where $c=\max_{x,y}\tr(X_x\,Y_y)$ is the maximal overlap of the two measurements, whereas  $$S(A|B)_{\rho}=-\tr(\rho_{AB}\ln\rho_{AB})+\tr(\rho_B\ln\rho_B)$$ denotes the conditional entropy of a bipartite state $\rho_{AB}$ conditioned on system $B$. The above inequality has been recently extended to the setting where the POVMs are replaced by two arbitrary quantum channels in \cite{gao2018uncertainty}.

In this section, we restrict ourselves to the setting of projective measurement, so that the measurement channels are Pinching maps onto different orthonormal bases
\begin{align}\label{eq:pinch} E_\mathcal{X}(\rho):=\sum_{x\in \mathcal{X}}\,|e_x^{(\mathcal{X}
)}\rangle\langle e_x^{(\mathcal{X})}|\,\langle e_x^{(\mathcal{X})}|\rho| e_x^{(\mathcal{X})}\rangle\,, \quad E_\mathcal{Y}(\rho):=\sum_{x\in \mathcal{Y}}\,|e_y^{(\mathcal{Y}
)}\rangle\langle e_y^{(\mathcal{Y})}|\,\langle e_Y^{(\mathcal{Y})}|\rho| e_Y^{(\mathcal{Y})}\rangle\,.\end{align}
In this setting, when $\rho_{AM}$ is close to $\frac{\Id}{d_A}\ten \rho_M$, the additive constant term "$-\ln c$" is responsible for the untightness of \eqref{uncert}, because $S(X|M),S(Y|M)$ and $S(A|M)$ are all small. As was realized in \cite{bardet2020approximate}, approximate tensorization implies a tightening of the inequality. However, \cite{bardet2020approximate} only discussed the memoryless setting, due to the lack of complete approximate tensorization results at that time. Here, we generalize their results to the presence of memory by a direct application of
\Cref{thm:AT} (see also \cite{laracuente2019quasi}). In this section, given a finite set $\mathcal{X}$ and $p\ge 1$, we will denote by $l_p(\mathcal{X})$ the space of complex valued functions $f:\mathcal{X}\to \mathbb{C}$ provided with the norm $$\|f\|_{l_p}:=\left(\sum_{x\in\mathcal{X}}|f(x)|^p\right)^{\frac{1}{p}}\,.$$
\begin{corollary}\label{cor:twobasis}
In the notations introduced above, denote \begin{align*}&\lambda_1:=\max_{y\in \mathcal{Y},x\in\mathcal{X} }\Big| \  |\langle e_x^{(\mathcal{X})}| e_y^{(\mathcal{Y})}\rangle |^2-\frac{1}{d_A}\Big|\\
&\lambda_2:=\,\norm{O_{\mathcal{X},\mathcal{Y}}:l_2(\mathcal{Y})\to l_2(\mathcal{X})}{}
\,,\end{align*}
where the matrix $O_{\mathcal{X},\mathcal{Y}}$ is defined as
\begin{align}\label{eq:matrixOXY}
O_{\mathcal{X},\mathcal{Y}}=\Big(|\langle e_x^{(\mathcal{X})}| e_y^{(\mathcal{Y})}\rangle |^2-\frac{1}{d_A} \Big)_{x,y}\pl.
\end{align}
Then for all states $\rho_{AM}\in\cD(\cH_{AM})$,
\begin{align*}
    (2c_{\operatorname{cb}}-1)\,S(A|M)_{\rho_{AM}}\le c_{\operatorname{cb}}\,\Big(S(X|M)_{(E_{\mathcal{X}}\otimes \id_M)(\rho_{AM})}+S(Y|M)_{(E_{\mathcal{Y}}\otimes \id_M)(\rho_{AM})} \Big)-\ln d_A\,,
    \end{align*}
where the constant $c_{\operatorname{cb}}$ can be given by the following cases
\begin{enumerate}
\item[$\operatorname{i)}$] $c_{\operatorname{cb}}=\frac{2d_A}{(1-\la_2)^2}$
\item[$\operatorname{ii)}$] if $\la_2<\frac{1}{\sqrt{2}}$, $c_{\operatorname{cb}}=1+\Big(\frac{\la_2}{1-\la_2}+\frac{\la_2^2}{1-2\la^2}\Big)d_A$.
\item[$\operatorname{iii)}$] if $\la_1<\sqrt{2\ln 2-1}$, $c_{\operatorname{cb}}=2\Big(1-\la_1^2(2\ln 2-1)^{-1}\Big)^{-1}$.
\end{enumerate}
In particular, the inequality is tight for $\rho_{AM}:=\frac{\Id}{d_A}\ten \rho_M$.
\end{corollary}
\begin{proof}
Let $\cM:=\cB(\cH_{A})$, and let $\cN_1$ ( resp. $\cN_2$) be the subalgebra of diagonal matrices in the orthonormal basis $\{|e^{(\mathcal{X})}_x\rangle\}_{x\in\mathcal{X}}$ (resp. $\{|e^{(\mathcal{Y})}_y\rangle\}_{y\in\mathcal{Y}}$). Then $\cN=\cN_1\cap \cN_2$ is the trivial algebra $\CC\Id_{A}$. For each alphabet $\mathcal{Z}\in\{\mathcal{X},\mathcal{Y}\}$, the conditional expectation $E_\mathcal{Z}$ is the Pinching map onto the diagonal $\operatorname{span}\{ |e_z^{(\mathcal{Z})}\rangle\langle e_z^{(\mathcal{Z})}|\}$ and $E_\cN(\rho)=\tr_A(\rho)d_A^{-1}{\Id}{}$. When  tensorizing with the identity map on $M$, we have for every bipartite  state $\rho_{AM} \in \mathcal{D}(\cH_{AM})$,
\begin{align*}
S(X|M)_{(E_{\mathcal{X}}\otimes \id_M)(\rho_{AM})} & = - D\Big( (E_{\mathcal{X}}\otimes \id_M)(\rho_{AM}) \Big\|  \frac{\Id_A}{d_A} \otimes \rho_M \Big) +\ln d_A \\
& = D(\rho_{AM} \| (E_{\mathcal{X}}\otimes \id_M)(\rho_{AM})) - D(\rho_{AM} \| (E_{\mathcal{N}}\otimes \id_M)(\rho_{AM})) +\ln d_A,
\end{align*}
where the last equality is derived from the chain rule identity \eqref{chain rule}. Similarly, we have
\begin{equation*}
D(\rho_{AM} \| (E_{\mathcal{N}}\otimes \id_M)(\rho_{AM})) = - S(A| M)_{\rho_{AM}} + \ln d_A .
\end{equation*}
Note that the map $E_\mathcal{X}E_\mathcal{Y}-E_{\cN}:\cB(\cH_A)\to \cB(\cH_A)$ from its support $\cN_1$ to the range $\cN_2$ is given by the matrix $O_{\mathcal{X},\mathcal{Y}}$. Therefore, we have that
\begin{align*}
&\|E_{\cX}E_{\cY}-E_\cN:L_1(\cB(\cH_A),d_A^{-1}\Id)\to \cB(\cH_A)\|_{\operatorname{cb}}
=\|O_{\mathcal{X},\mathcal{Y}}:l_1(\mathcal{Y})\to l_\infty(\mathcal{X})\|=\la_1
\\&\|E_{\cX}E_{\cY}-E_\cN:L_2(\cB(\cH_A),d_A^{-1}\Id)\to L_2(\cB(\cH_A),d_A^{-1}\Id)\|
=\|O_{\mathcal{X},\mathcal{Y}}:l_2(\mathcal{Y})\to l_2(\mathcal{X})\|=\la_2
\end{align*}
Then the assertions i) and ii) follow from \Cref{thm:AT} with the fact that $C_{\max,\operatorname{cb}}=d_A$ (cf. \eqref{eq:indicesex}).
The assertion iii) follows from Corollary \ref{Coro:sharperAT} and the equivalence (see the proof of Lemma \ref{lemmageneral} in \Cref{sec:normstoorder}) that \begin{align*}\|E_{\cX}E_{\cY}-E_\cN:L_1(\cB(\cH_A),d_A^{-1}\Id)&\to \cB(\cH_A)\|_{\operatorname{cb}}
\le \la_1<1  \\ \Longleftrightarrow\pl & (1-\la_1)E_{\cN} \le_{\operatorname{cp}} E_{\cX}E_{\cY}\le_{\operatorname{cp}} (1+\la_1)E_{\cN}  .\end{align*}
\end{proof}
As applications of Corollary \ref{Coro:sharperAT},
our analysis can be extended to uncertainty relations with quantum memory for multiple measurements. In the case of three measurements \cite{coles2011information,berta2019quantum}, we have
\begin{corollary}\label{cor:3basis}
Let $\mathcal{X}=\{\ket{e_{x}^{\mathcal{X}}}\}, \mathcal{Y}=\{\ket{e_{y}^{\mathcal{Y}}}\}$ and $\mathcal{Z}=\{\ket{e_{z}^{\mathcal{Z}}}\}$ be three orthonormal bases on $\cH_A$. Let $O_{\mathcal{Y},\mathcal{X}}$ and $O_{\mathcal{Z},\mathcal{Y}}$ be the matrices for the corresponding bases defined as in \eqref{eq:matrixOXY}. Denote
 \begin{align*}&\lambda_1:=\norm{O_{\mathcal{X},\mathcal{Y}}\,O_{\mathcal{Y},\mathcal{Z}}\,O_{\mathcal{Z},\mathcal{Y}}\,O_{\mathcal{Y},\mathcal{X}}:l_1(\mathcal{X})\to l_\infty(\mathcal{X})}{}\\
&\lambda_2:=\norm{O_{\mathcal{Z},\mathcal{Y}}\,O_{\mathcal{Y},\mathcal{X}}:l_2(\mathcal{X})\to l_2(\mathcal{Z})}{}
\,.\end{align*}
Then for all state $\rho_{AM}\in\cD(\cH_{AM})$,
\begin{align*}
    (3c_{\operatorname{cb}}-1)\,S(A|M)\le & \, c_{\operatorname{cb}}\,\Big(S(X|M)_{(E_{\mathcal{X}}\otimes \id_M)(\rho_{AM})}+S(Y|M)_{(E_{\mathcal{Y}}\otimes \id_M)(\rho_{AM})}\\ &+S(Z|M)_{(E_{\mathcal{Z}}\otimes \id_M)(\rho_{AM})} \Big)-\ln d_A\,,
    \end{align*}
where the constant $c_{\operatorname{cb}}$ is given by
\begin{align*}
c_{\operatorname{cb}}:=4(k+1)\Big(1-\eps^2(2\ln 2-1)^{-1}\Big)^{-1}
\end{align*}
for every $k\in \mathbb{N}$ such that $\epsilon:=\la_1\la_2^{2k}<2\ln 2-1$.
\end{corollary}
\begin{proof}
 Let $E_\mathcal{X}, E_\mathcal{Y}$ and $E_\mathcal{Z}$ be the Pinching maps onto $\mathcal{X}$, $\mathcal{Y}$ and $\mathcal{Z}$ respectively. Moveover, take $E_\cN(\rho)=\tr(\rho)\frac{\Id}{d_A}$.
Note that $O_{\mathcal{X},\mathcal{Y}}\,O_{\mathcal{Y},\mathcal{Z}}=(O_{\mathcal{Z},\mathcal{Y}}\,O_{\mathcal{Y},\mathcal{X}})^*$. Therefore, we have
(see e.g. \cite[Lemma 3.1]{gao2020fisher})
%\begin{align*}
%&\norm{E_\mathcal{X}E_\mathcal{Y}E_\mathcal{Z}-E_\cN:L_1(\cB(\cH_A),\frac{\Id}{d_A})\to L_2(\cB(\cH_A),\frac{\Id}{d_A}) }{cb}^2\\=&
%\norm{E_\mathcal{Z}E_\mathcal{Y}E_\mathcal{X}-E_\cN:L_2(\cB(\cH_A),\frac{\Id}{d_A})\to \cB(\cH_A)}{cb}^2\\=&\norm{E_\mathcal{X}E_\mathcal{Y}E_\mathcal{Z}E_\mathcal{Y}E_\mathcal{X}-E_\cN:L_1(\cB(\cH_A),\frac{\Id}{d_A})\to \cB(\cH_A)}{cb}
%\\=&\norm{O_{\mathcal{X},\mathcal{Y}}O_{\mathcal{Y},\mathcal{Z}}O_{\mathcal{Z},\mathcal{Y}}O_{\mathcal{Y},\mathcal{X}}:l_1(\mathcal{X})\to l_\infty(\mathcal{X})}{}=\la_1\
%\end{align*}
\begin{align*}
\la_1=&\norm{O_{\mathcal{X},\mathcal{Y}}\,O_{\mathcal{Y},\mathcal{Z}}\,O_{\mathcal{Z},\mathcal{Y}}\,O_{\mathcal{Y},\mathcal{X}}:l_1(\mathcal{X})\to l_\infty(\mathcal{X})}{}\\=&\norm{O_{\mathcal{X},\mathcal{Y}}\,O_{\mathcal{Y},\mathcal{Z}}:l_2(\mathcal{X})\to l_\infty(\mathcal{X})}{}^2=\norm{O_{\mathcal{Z},\mathcal{Y}}\,O_{\mathcal{Y},\mathcal{X}}:l_1(\mathcal{X})\to l_2(\mathcal{Z})}{}^2\,.
\end{align*}
Using the $L_2$ condition
\begin{align*}
\norm{E_\mathcal{X}E_\mathcal{Y}E_\mathcal{Z}-E_\cN:L_2(d_A^{-1}{\Id})\to L_2(d_A^{-1}{\Id}) }{\operatorname{cb}}=\norm{O_{\mathcal{Z},\mathcal{Y}}\,O_{\mathcal{Y},\mathcal{X}}:l_2(\mathcal{X})\to l_2(\mathcal{Z})}{}:=\la_2
\end{align*}
we have for any integer $k\in \mathbb{N}$,
 \begin{align*}
 &\norm{(E_\mathcal{X}E_\mathcal{Y}E_\mathcal{Z}E_\mathcal{Y}E_\mathcal{X})^{(k+1)}-E_\cN:L_1(d_A^{-1}{\Id})\to \cB(\cH_A)}{\operatorname{cb}}
 \\ &\qquad=\,\norm{(O_{\mathcal{Z},\mathcal{Y}}\,O_{\mathcal{Y},\mathcal{X}}\,O_{\mathcal{X},\mathcal{Y}}\,O_{\mathcal{Y},\mathcal{Z}})^{(k+1)}:l_1(\mathcal{X})\to l_\infty(\mathcal{X}) }{}
 \\&\qquad=\,\norm{O_{\mathcal{X},\mathcal{Y}}\,O_{\mathcal{Y},\mathcal{Z}}:l_2(\mathcal{X})\to l_\infty(\mathcal{X}) }{}
\cdot \norm{(O_{\mathcal{Z},\mathcal{Y}}\,O_{\mathcal{Y},\mathcal{X}}\,O_{\mathcal{X},\mathcal{Y}}\,O_{\mathcal{Y},\mathcal{Z}})^k:l_2(\mathcal{X})\to l_2(\mathcal{X}) }{}\\
 &
\qquad\qquad\qquad\qquad\qquad\qquad\qquad\qquad\quad\cdot \norm{O_{\mathcal{Z},\mathcal{Y}}\,O_{\mathcal{Y},\mathcal{X}}:l_1(\mathcal{X})\to l_2(\mathcal{X}) }{}
 \\ &\qquad= \,\la_1^{1/2}\cdot \la_2^{2k}\cdot \la_1^{1/2}=\la_1 \la_2^{2k}\ .
\end{align*}When
$\eps= \la_{1}\la_2^{2k}<1$, this implies
\begin{align*}(1-\eps)E_\cN\le_{\operatorname{cp}} (E_\mathcal{X}E_\mathcal{Y}E_\mathcal{Z}E_\mathcal{Y}E_\mathcal{X})^{k+1}\le_{\operatorname{cp}} (1+\eps)E_\cN .\end{align*}
Therefore, when $\eps= \la_{1}\la_2^{2k}<\sqrt{2\ln 2-1} $, the assertion follows from Theorem \ref{theorem:meta} and chain rule, similarly to the proof of Corollary \ref{Coro:sharperAT}.
\end{proof}
\begin{rem}
{\rm The constant $\la_1$ in Corollary \ref{cor:3basis} can be explicitly written as
\[\la_1=\max_{x,x'}\sum_{y,y',z}\Big| \pl |\langle e_{x'}^{(\cX)} \ | e_{y'}^{(\cY)}\rangle\langle e_{y'}^{(\cY)} | e_{z}^{(\cZ)}\rangle \langle e_{z}^{(\cZ)} | e_{y}^{\cY}\rangle\langle e_{y}^{(\cY)} | e_{x}^{(\cX)}\rangle|^2-\frac{1}{d}\Big| \pl.\]
Actually, one can alternatively use the product of $E_\mathcal{X}$, $E_\mathcal{Y}$, $E_\mathcal{Z}$ in different orders in the definition of $\la_1$ and $\la_2$. Indeed, define $\Phi_\mathcal{W}=E_{\mathcal{W}_1}\,E_{\mathcal{W}_2}\,E_{\mathcal{W}_3}$ where $\mathcal{W}=(\mathcal{W}_1\,\mathcal{W}_2\,  \mathcal{W}_3)$ is some permutation of $(\mathcal{X}\mathcal{Y}\mathcal{Z})$.
Then the assertion in Corollary \ref{cor:3basis} remains valid for
\begin{align*}\lambda_1:=&\min_{\mathcal{W}}\norm{\Phi_\mathcal{W}^*\,\Phi_\mathcal{W}-E_{\cN}:L_1(d_A^{-1}{\Id})\to \cB(\cH_A)}{}
\\ =& \min_{\mathcal{W}}\norm{O_{\mathcal{W}_1,\mathcal{W}_2}\,O_{\mathcal{W}_2,\mathcal{W}_3}\,O_{\mathcal{W}_3,\mathcal{W}_2}\,O_{\mathcal{W}_2,\mathcal{W}_1}:
l_1(\mathcal{W}_1)\to l_\infty(\mathcal{W}_1)}{}\\
\lambda_2:=&\min_{\mathcal{W}}\norm{\Phi_\mathcal{W}-E_{\cN}:L_2(d_A^{-1}{\Id})\to L_2(d_A^{-1}{\Id})}{} \\
=&\min_{\mathcal{W}}\norm{O_{\mathcal{W}_3,\mathcal{W}_2}\,O_{\mathcal{W}_2,\mathcal{W}_1}:l_2(\mathcal{W}_1)\to l_2(\mathcal{W}_3)}{}
\,.\end{align*}
}
\end{rem}

\section{Symmetric semigroups}\label{sec:symmetric}
\subsection{Symmetric Lindbladians}In this section, we use the approximate tensorization results of Section \ref{sec:AT} to give tighter bounds on the CMLSI constant for  symmetric quantum Markov semigroup.
Let $\cH$ be a Hilbert space and $(\cP_t)_{t\ge 0}$ be a quantum Markov semigroup on the algebra $\cB(\cH)$. We say $\cP_t$ is symmetric if for each $t$,
\[ \tr(Y^\dagger\cP_t(X))=\tr(\cP_t(Y^\dagger)X)\ , \ \forall\ X,\,Y\in \cB(\cH) .\]
Namely, $\cP_t=\cP_t^*$ is GNS-symmetric with respect to the completely mixed state $\frac{\Id}{d_\cH}$. In this case, we do not distinguish $\cP_t$ and with it predual $\cP_{t*}$ because
the Schr\"{o}dinger picture is equivalent to the Heisenberg picture. Due to this symmetry, the Lindbladian \eqref{eqlindbladsigma} takes a simple form
\[ \cL(\rho)=-\sum_{k=1}^l[a_k,[a_k,\rho]]\]
where $a_1,\ldots,a_l\in \cB(\cH)$ form a family of self-adjoint operators. Let us first consider a single term in the generator
\[\cL_a(\rho)=-[a,[a,\rho]]\]
for a self-adjoint $a\in \cB(\cH)$. Let $a=\sum_{i=1}^n\kappa_i P_i$ be the spectral decomposition of $a$,
where $P_i$ is the spectral projection with respect to the eigenvalue $\kappa_i$. One calculates that
\[ [a,\rho]=\sum_{i,j}(\kappa_i-\kappa_j)P_i\rho P_j\ , \  \cL_a(\rho)=-[a,[a,\rho]]=-\sum_{i,j}(\kappa_i-\kappa_j)^2P_i\rho P_j\ .\]
Then $\cL_a$ generates the semigroup
\[ e^{\cL_at}(\rho)= \sum_{i,j}e^{-(\kappa_i-\kappa_j)^2t}P_i\rho P_j\ .\]
which is a Schur multiplier semigroup (also called generalized dephasing semigroup). The invariant subalgebra $\cN_a=\oplus_{i=1}^n \cB(P_i\cH)$ and the indices are
\[ C(\cB(\cH),\cN_a)=C_{\operatorname{cb}}(\cB(\cH),\cN_a)=n\le d_\cH. \]
Viewing $\cL_a$ as a self-adjoint operator on $L_2(\cB(\cH),\tr)$,
 $P_i\cB(\cH)P_j$ corresponds to the eigenspace associated to the eigenvalue $-(\kappa_i-\kappa_j)^2$.
Thus the norm and spectral gap of $\cL_a$ are
\[\norm{\cL_a:L_2(\cB(\cH),\tr)\to L_2(\cB(\cH),\tr)}{}=\max_{i,j}|\kappa_i-\kappa_j|^2\ ,\quad  \la(\cL_a)=\min_{i\neq j}|\kappa_i-\kappa_j|^2\ .\]
It was proved in \cite[Theorem 4.23]{brannan2020complete} that Schur multiplier semigroups admit the following estimates on their CMLSI constant
\begin{align} \frac{\la(\cL_a)}{2\ln(2n)}\le \al_{\operatorname{CMLSI}}(\cL_a)\le 2\la(\cL_a) \ \label{schur} .\end{align}
(Note that our normalization differs with \cite{brannan2020complete} by a factor of $2$). Moreover, for a commuting family $\{a_1,\ldots,a_l\}$,
the Lindbladian $\cL(\rho)=-\sum_{k=1}^l[a_k,[a_k,\rho]]$ also generates a Schur multiplier semigroup and the above estimate \eqref{schur} remains valid.

To extend the above estimate to general Lindbladians of the form $\cL=\sum_{k}\cL_{a_k}$ with not necessarily commuting operators $a_k$, we make use of the sharper approximate tensorization bounds derived in Corollary \ref{Coro:sharperAT}.

% Now, we use the above results together with the approximate tensorization bounds of \Cref{sec:AT} in order to give alternative estimates on the CMLSI constant of $\cL=\sum_{i}\cL_{a_i}$.

\begin{theorem}\label{thm:symmetric}
 Let $\cP_t=e^{\cL t}: \mathbb{M}_d\to \mathbb{M}_d$ be a symmetric quantum Markov semigroup, and let its generator be given by
 \[ \cL(X)=\sum_{k=1}^l\cL_{a_k}(X)=-\sum_{k=1}^l[a_k,[a_k,X]]\,.\]
 Denote $E_\cN$, resp. $E_k$, as the conditional expectation onto the kernel of $\cL$, resp. that of $\cL_{a_k}$. Denote $\la(\cL_{a_k})$ as the spectral gap of $\cL_{a_k}$ and $\la_E$ as the spectral gap of $\cL_E=\sum_{k}(E_k-\id)$. Then,
\[ \al_{\operatorname{CMLSI}}(\cL)\ge \left\lceil\frac{2\ln d+1}{\ln(\lambda^{-1})}\right\rceil^{-1}\, \frac{\min_{k}\la(\cL_{a_k})}{8
\ln(2d)}\,,\]
where $\la:=1-\frac{\la_E}{l}$.
\end{theorem}
\begin{proof}
 Consider
$\frac{1}{l}\sum_{k}E_k-E_\cN$ as a positive contraction supported on $L_2(\cN)^\perp \subset L_2(\cM)$. Since $$\frac{1}{l}\sum_{k}E_k-E_\cN+\frac{1}{l}\cL_E=\frac{1}{l}\sum_{k}E_k-E_\cN+\frac{1}{l}\sum_{k}\id-E_k=\id-E_\cN\ ,$$
we have
\[ \Big\|{\frac{1}{l}\sum_{k=1}^l E_k-E_\cN:L_2(\M)\to L_2(\M)}\Big\|{}=1-\frac{\la_E}{l}=:\la\ .\]
Using $\Psi=\frac{1}{l}\sum_{k}E_k$ for the map in Corollary \ref{Coro:sharperAT}, we have the approximate tensorization
\begin{align}
   D(\rho\|E_{\cN}(\rho))\le 4\left\lceil\frac{2\ln d+1}{\ln(\lambda^{-1})} \right\rceil\,\sum_{k=1}^l\,D(\rho\|E_{k}(\rho))\,,
\end{align}
where we used the fact  $C_{\operatorname{cb}}(\mathbb{M}_d:\cN)\le C_{\operatorname{cb}}(\mathbb{M}_d:\mathbb{C})=d^2$.
For each $k$, we use the estimate \eqref{schur}, so that
\[ \frac{\la(\cL_{a_k})}{2
\ln(2d)}D(\rho\|E_{k}(\rho))\le -\tr(\cL_{a_k}(\rho)\ln \rho)\ .\]
Therefore,
\begin{align*}
D(\rho\|E_{\cN}(\rho))\le& 4\left\lceil\frac{2\ln d+1}{\ln(\lambda^{-1})} \right\rceil\,\sum_{k=1}^l\,D(\rho\|E_{k}(\rho)) \\
\le& -4\left\lceil\frac{2\ln d+1}{\ln(\lambda^{-1})} \right\rceil\,\sum_{k=1}^l \frac{2
\ln(2d)}{\la(\cL_{a_k})}\tr(\cL_{a_k}(\rho)\ln \rho)\\
\le &- 4\left\lceil\frac{2\ln d+1}{\ln(\lambda^{-1})} \right\rceil\, \frac{2
\ln(2d)}{\min_{k}\la(\cL_{a_k})}
\sum_{k=1}^l \tr(\cL_{a_k}(\rho)\ln \rho)\\
=&- 4\left\lceil\frac{2\ln d+1}{\ln(\lambda^{-1})} \right\rceil\, \frac{2
\ln(2d)}{\min_{k}\la(\cL_{a_k})}
\tr(\cL(\rho)\ln \rho)\ .
\end{align*}
Since both Corollary \ref{Coro:sharperAT} and \eqref{schur} remain true in their complete version, the above estimates remain the same for any $n\in\mathbb{N}$ and all $\rho\in \cD(\mathbb{C}^d\ten \mathbb{C}^n)$, which implies \[ \al_{\operatorname{CMLSI}}(\cL)\ge \left\lceil\frac{2\ln d+1}{\ln(\lambda^{-1})} \right\rceil^{-1}\, \frac{\min_{k}\la(\cL_{a_k})}{8
\ln(2d)}\ .
\]
\end{proof}
Next, we present a variation of above theorem using the detectability lemma \ref{thm:detectability}. %We say a family of operators $\{a_1,\ldots,a_l\}$ can be divided into $m$ commuting sub-families if there exists mutually disjoint subsets $I_1, I_2,\ldots, I_m\subset \{1,\ldots,k\}$ such that $\cup_{i=1}^k I_i=\{1,\ldots,k\}$ and for each $i$, $[a_{k},a_{k'}]=0$ whenever $k,k'\in I_i$.
\begin{corollary}\label{cor:symmetric2}
In the notations of Theorem \ref{thm:symmetric},
suppose that for each $k$, $a_k$ commutes with all but at most $m$ other elements of $\{a_1,\ldots,a_l\}$. Then
\[\al_{\operatorname{CMLSI}}(\cL)\ge \left\lceil\frac{2\ln d+1}{\ln(\lambda^{-1})} \right\rceil^{-1}\, \frac{\min_{k}\la(\cL_{a_k})}{8
\ln(2d)}\ ,\]
where $\la=\frac{1}{m^{-2}\la_E+1}$.
\end{corollary}
\begin{proof}
Recall that $E_k$ is the condition expectation onto the kernel of $\cL_{a_k}$. We have $[E_k,E_{k'}]=0$ if $[a_k,a_{k'}]=0$, and $E_\cN$ is the conditional expectation onto the intersection of the ranges of each $E_k$. Then by the detectability lemma \ref{thm:detectability}, given an arbitrary ordering of the conditional expectations
\[\norm{E_{1}E_2\ldots E_l-E_\cN:L_2(\tr)\to L_2(\tr)}{}\le \frac{1}{m^{-2}\la_E+1}\ .\]
By symmetry, this implies
\[\norm{\Phi-E_\cN:L_2(\tr)\to L_2(\tr)}{}\le \frac{1}{m^{-2}\la_E+1}\ \]
for $\Phi:=\frac{1}{2}(E_1E_2\ldots E_l+ E_lE_{l-1}\ldots E_1)$. The rest of argument is identical to that of the proof of Theorem \ref{thm:symmetric}.
\end{proof}

\begin{rem}{\rm The main difference between Theorem \ref{thm:symmetric} and Corollary \ref{cor:symmetric2} is that they use different maps from Corollary \ref{Coro:sharperAT}. When $\cL_{a_k}$ are all mutually non-commuting, it is likely that the map $\Psi=\frac{1}{l}\sum_{k}E_k$ provides tighter bounds. In the case when $\cL_{a_k}$ commutes with most $\cL_{a_{k'}}$, Corollary \ref{cor:symmetric2} is better, as argued for examples in Section \ref{sec:localsemigroups}.
}
\end{rem}

\begin{rem}{\rm For the spectral gap, we have that for each $k$
\[\min_{k}\la(\cL_{a_k}) \le \la(\cL_{a_k})\le \la(\cL)\ .\]
Also, the spectral gap of $\cL_E=\sum_{i}E_i-\id$ satisfies
\[ \la(\cL_E)\le \norm{\cL:L_2(\tr)\to L_2(\tr)}{}^{-1} \la(\cL)\ ,\]
because $\id-E_k\ge -\norm{\cL_{a_k}:L_2\to L_2}{}^{-1}\cL_{a_k}$ and
$\norm{\cL_{a_k}:L_2\to L_2}{}\le \norm{\cL:L_2\to L_2}{}$, since each generator $\cL_{a_k}$ is a non-positive operator in the $L_2$ sense.
Moreover in general we have $m\le l\le d^2$, so that the constant obtained in Theorem \ref{thm:symmetric} can be further estimated by
\[\al_{\operatorname{CMLSI}}(\cL)\ge \Big\lceil\frac{2\ln d+1}{-\ln(1-\frac{\la(\cL)}{d^2})} \Big\rceil^{-1}\, \frac{\min_{k}\la(\cL_{a_k})}{8
\ln(2d)}=O\Big( \frac{\la(\cL)\min_{k}\la(\cL_{a_k})}{(d\ln d)^2}\Big)\]
Thus in general, the constants obtained in Theorem \ref{thm:symmetric} and Corollary \ref{cor:symmetric2} are not tighter than that of Theorem \ref{theo:CLSI}.}
\end{rem}

\subsection{Groups and transference}\label{sec:transference}
We now briefly review the technique of transference of semigroups introduced in \cite{gao2020fisher}, which will be used in the next example.
Let $G$ be a compact group and $\mu$ be the Haar measure on $G$. Let $\mathbb{M}_d$ be the $d\times d$ matrix algebra and $u:G\to \mathbb{M}_d$ be a (projective) unitary representation. This induces the
transference map
\begin{align}\label{transferrencemap}
\pi: \mathbb{M}_d\to  L_\infty(G,\mathbb{M}_d),\quad  \pi(X)(g)=u(g)^\dagger Xu(g)\pl,
\end{align}
which is a trace preserving $*$-homomorphism. We say that a symmetric quantum Markov semigroup
$(\cP_t: \mathbb{M}_d\to \mathbb{M}_d)_{t\ge 0}$ on $\mathbb{M}_d$ is a semigroup transferred from a classical Markov semigroup $S_t:L_\infty(G) \to L_\infty(G)$ if the following diagram commutes
\begin{align*}
 \begin{array}{ccc}  L_{\infty}(G, \mathbb{M}_d)\pl\pl &\overset{ S_t\ten \id_{\mathbb{M}_d}}{\longrightarrow} & L_{\infty}(G, \mathbb{M}_d) \\
                    \uparrow \pi    & & \uparrow \pi  \\
                    \mathbb{M}_d\pl\pl &\overset{\cP_t}{\longrightarrow} & \mathbb{M}_d
                     \end{array} \pl .
                     \end{align*}
Then $\cP_t\cong (S_t\ten \id_{\mathbb{M}_d})|_{\pi(\mathbb{M}_d)}$ via the embedding $\pi$. More explicitly, let $S_t:L_\infty(G,\mu)\to L_\infty(G,\mu)$ be a right invariant classical Markov semigroup given by
\[ S_t f(g)=\int k_t(gh^{-1})f(h)d\mu(h) \]
with some kernel function $k_t$. Write $L_h f(g)=f(hg)$ as the action of left multiplication. Then
\[S_tf(g)=\int k_t(h^{-1}) L_hf(g) d\mu(h)  \] and the corresponding transferred semigroup is defined as
\[\cP_t(X)=\int k_t(h^{-1})u(h)^\dagger Xu(h) d\mu(h)\,.\]
In particular, if $S_t$ is ergodic, i.e. $\displaystyle\lim_{t\to\infty }S_tf= \mathbb{E}_\mu[f] \,1$, the fixed point subalgebra of $\cP_t$ coincides with the commutant of the representation
$$\cN=u(G)':=\{X\in\mathbb{M}_d \ |\ Xu(g)=u(g)X\ ,\ \forall\  g\in G\ \}\ .$$
The following is taken from \cite[Proposition 4.7]{gao2020fisher} which follows from the observation that  (see \cite[Lemma 4.6]{gao2020fisher}):
\begin{align*}
    \pi\circ\cP_t=(S_t\otimes \id_{\mathbb{M}_d})\circ\pi\,.
\end{align*}
% $\cP_t$ is the restriction of $S_t\ten \id_{\mathbb{M}_d}$ onto $\pi(\mathbb{M}_d)\cong \mathbb{M}_d$.
\begin{proposition}\label{prop:transference}
Suppose $\cP_t=e^{\cL t}:\mathbb{M}_d\to \mathbb{M}_d$ is a semigroup transferred from a classical Markov semigroup $S_t=e^{L_Gt}:L_\infty(G)\to L_\infty(G)$. Then
\[ \la(L_G)\le \la(\cL)\pl, \pl \quad \al_{\operatorname{CMLSI}}(L_G)\le \al_{\operatorname{CMLSI}}(\cL)\\ .\]
\end{proposition}

%\textcolor{red}{perhaps we can quickly explain in words how transference work and refer to the literature (for both discrete and lie groups), unless we want to explicitely explain the problem faced by it in the examples that are not the random transpositions and $U:V\to \mathcal{H}$ a representation of $V$ %onto the Hilbert space \mathcal{H}$, one can define %the generator \cite{bardet2019group}
%\begin{align*}
   % \mathcal{L}_{G*}(\rho):=\sum_{g\in V}\,c_g\,\big(\rho- U_g^\dagger\, \rho \,U_g \big)\,.
%\end{align*}}
%\textcolor{red}{mention the transference bounds of Junge Haojian, and the one of Li Junge for groups (although if we add it here, we probably want to define Lie groups here also)}:
%\begin{align}\label{eq:transference}
   % \alpha \textcolor{red}{???}
%\end{align}

\subsection{Semigroups transferred from the sub-Laplacian on $\operatorname{SU}(2)$}\label{SU2sec}
Let $\operatorname{SU}(2)$ be the special unitary group on $\mathbb{C}^2$. Denote
\[ X=\left[\begin{array}{cc} 0& 1\\ -1&0 \end{array}\right]\pl, \quad Y=\left[\begin{array}{cc} 0& i\\ i&0 \end{array}\right]\pl, \quad Z=\left[\begin{array}{cc} i& 0\\ 0&-i \end{array}\right]\pl .\]
as the anti-selfadjoint Pauli matrices. Then
\[\operatorname{SU}(2)=\{ aX+bY+cZ+d\Id\ |\  |a|^2+|b|^2+|c|^2+|d|^2=1\}\pl\]
which is isomorphic to the $3$-sphere $\mathbb{S}^3$. Its Lie algebra is the anti-selfadjoint matrix  space $$i(\mathbb{M}_2)_{\operatorname{sa}}:=\text{span}\{X,Y,Z\}$$ equipped with the Lie bracket relations
\begin{align}\label{eq:lie} [X,Y]=2Z,\quad  [Y,Z]=2X,\quad [Z,X]=2Y\,.\end{align}
The canonical bi-invariant Riemannian metric on $\operatorname{SU}(2)$ admits $\text{span}\{X,Y,Z\}$ as an orthonormal basis. The representation theory of $\operatorname{SU}(2)$ gives the well-known spin structure of quantum mechanics, where any irreducible representation of $\operatorname{SU}(2)$ is indexed by an integer $m\in \mathbb{N}^+$.
Let $\eta_m:\mathfrak{su}(2)\to i(\mathbb{M}_{m})_{\operatorname{sa}}$ be the Lie algebra homomorphism induced by the $m$-th irreducible representation, and let $\{\ket{j}| j=1,\cdots,m \}$
be the eigenbasis of $\eta_m(Z)$.  Denote $X_m:=\eta_m(X)$, and similarly for $Y_m$ and $Z_m$ as short notations. Under the normalization of \eqref{eq:lie},
\begin{align*}
&X_m\ket{j}=\sqrt{(j-1)(m-j+1)}\,\ket{j-1}-\sqrt{(j+1)(m-j-1)}\,\ket{j+1}\\
 &Y_m\ket{j}=i\sqrt{(j-1)(m-j+1)}\,\ket{j-1}+i\sqrt{(j+1)(m-j-1)}\,\ket{j+1}\\
 &Z_m\ket{j}=(m-2j+1)\,\ket{j}\,.
\end{align*}
%Since $\operatorname{SU}(2)$ is simply connected (isomorphic to 3-sphere $S^3$), every such Lie algebra representation induces a group representation $\pi_n:\operatorname{SU}(2)\to M_{n+1}$.
For each irreducible representation, we consider quantum Markov semigroups transferred from two classical Markov semigroups on $\operatorname{SU}(2)$. The heat semigroup $(e^{\Delta t})_{t\ge 0}$ on $\operatorname{SU}(2)$ is given by the Casimir operator
\[\Delta=X^2+Y^2+Z^2.\]
It follows from the complete Barkry-Emery Theorem \cite[Theorem 4.3]{li2020graph} that the heat semigroup has  $\al_{\operatorname{CMLSI}}(\Delta)\ge 2$ because its Ricci curvature is $1$.
Therefore, the transferred Lindbladian on $\mathbb{M}_{m}$
\[\cL_{m}^{\Delta}(\rho)=-[X_m,[X_m,\rho]]-[Y_m,[Y_m,\rho]]-[Z_m,[Z_m,\rho]]\pl\]
admits a uniform CMLSI constant $\al_{\operatorname{CMLSI}}(\cL_m^\Delta)\ge 2$ by Proposition \ref{prop:transference}.

The canonical sub-Laplacian on $\operatorname{SU}(2)$ is
\[ \Delta_H=X^2+Y^2\]
It is known (see e.g. \cite{Baudoin09}) that $\Delta_H$ is hypoelliptic and generates a classical Markov semigroup $(e^{\Delta_H t})_{t\ge 0}$.
Its transferred Lindbladian on $\mathbb{M}_m$ is given by
\[\cL_{m}^{H}(\rho)=-[X_m,[X_m,\rho]]-[Y_m,[Y_m,\rho]]\pl.\]
Although the CMLSI constant for the sub-Laplacian $\Delta_H$ is currently still unknown, we can use approximate tensorization to obtain a lower bound on $\al_{\operatorname{CMLSI}}(\cL_{m}^{H})$ for each $m$.

\medskip

\textbf{Case $m=2$:} It was observed in \cite[Corollary 4.10]{li2020graph} that
\[ e^{\cL_{2}^H t}(aX+bY+cZ+dI)=e^{-t}aX+e^{-t}bY+e^{-2t}cZ+d\Id\pl.\]
is exactly the Fermionic Ornstein-Uhlenbeck semigroup \cite{gross1975hypercontractivity,Carlen1993,carlen2014analog,carlen2017gradient}. Hence $\al_{\operatorname{CMLSI}}(\cL_{2}^{H})=2$.

\medskip

\textbf{Case $m=3$:} The operators $X_3$ and $Y_3$ take the following form in the eigenbasis of $Z_3$:
\begin{align*}
 X_3=\left[\begin{array}{ccc} 0& \sqrt{2}& 0\\ -\sqrt{2}& 0& \sqrt{2}\\ 0& -\sqrt{2}& 0 \end{array}\right]\pl,\qquad  Y_3=i\left[\begin{array}{ccc} 0& \sqrt{2}& 0\\ \sqrt{2}& 0& \sqrt{2}\\ 0& \sqrt{2}& 0 \end{array}\right]\pl.
\end{align*}
The spectra of both $X_3$ and $Y_3$ coincide with $\{2i,0,-2i\}$ and their eigenvectors are
\begin{align*}
 &\ket{x_1}=\frac{1}{2}(1,\sqrt{2},1), ~~\ket{x_2}=\frac{1}{2}(1,-\sqrt{2},1),~~\ket{x_3}=\frac{1}{\sqrt{2}}(1,0,-1)\\
 &\ket{y_1}=\frac{1}{2}(-1,i\sqrt{2},1),~~ \ket{y_2}=\frac{1}{2}(-1,-i\sqrt{2},1),~~\ket{y_3}=\frac{1}{\sqrt{2}}(1,0,1)\,.
\end{align*}
Keeping the notations of \Cref{eq:pinch}, we denote by $E_\mathcal{X}$ and $E_\mathcal{Y}$ the conditional expectations onto the basis of $X_3$ and $Y_3$, and by $E_\cN:=\tr(.)3^{-1}\Id$. Therefore, we have
\[O_{\mathcal{X},\mathcal{Y}}=
\left[\begin{array}{ccc} -\frac{1}{12}& -\frac{1}{12}& \frac{1}{6}\\ -\frac{1}{12}& -\frac{1}{12}& \frac{1}{6}\\ \frac{1}{6}& \frac{1}{6}& -\frac{1}{3} \end{array}\right]\  ,\quad  \la_1= \frac{1}{3}\,, \quad\la_2=\frac{1}{2}\,
\]
where $\la_1$ and $\la_2$ are defined as in Corollary \ref{cor:twobasis}. Denote $c_1,c_2$ and $c_3$ as the complete approximate tensorization constants obtained respectively in i), ii) and iii) of Corollary \ref{cor:twobasis}. Since $\la_1=\frac{1}{3}<\sqrt{2\ln 2-1}$ and $\la_2=\frac{1}{2}<\frac{1}{\sqrt{2}}$, we have
\begin{align*}
c_1=24,\  \quad c_2=11/2,\ \quad  c_3=2(1-\frac{1}{9}(2\ln 2-1)^{-1})^{-1}\approx 2.808<3 \,.
\end{align*}
Therefore, the constant $c_3$ provides the best complete approximate tensorization constant for this example. Denote $\rho_\cN:= (E_\cN\ten \id)(\rho)$, $\rho_X=( E_\mathcal{X}\ten \id)(\rho)$ and $\rho_Y=( E_\mathcal{Y}\ten \id)(\rho)$.
We obtained from Corollary \ref{cor:twobasis} that
\begin{align*}
D(\rho\|\rho_\cN)\le c_3\left(D(\rho\|\rho_X) + D(\rho\|\rho_Y))\le 3(D(\rho\|\rho_X) + D(\rho\|\rho_Y)\right)\,.
\end{align*}
Denoting
\[ \mathcal{L}_{X_3}(\rho)=-[X_3,[X_3,\rho]]\pl,  \qquad \mathcal{L}_{Y_2}(\rho)=-[Y_3,[Y_3,\rho]],\]
both by \eqref{schur}, we get
\[\al_{\operatorname{CMLSI}}(\mathcal{L}_{X_3})=\al_{\operatorname{CMLSI}}(\mathcal{L}_{Y_3})\ge \frac{1}{\ln 6}\pl.\]
Combining these bounds, we get the following bound on the CMLSI constant of $\cL_3^H$
\begin{align*}D(\rho\|\rho_\cN)\le & 3 (D(\rho\|\rho_X) + D(\rho\|\rho_Y))
 \le
\Big(3\ln 6\Big) (\operatorname{EP}_{\cL_{X_3}\otimes \id}(\rho) + \operatorname{EP}_{\cL_{Y_3}\otimes \id}(\rho))\pl.
\end{align*}
Namely, $\al_{\operatorname{CMLSI}}(\cL_3^H)\ge \big(3\ln 6\big)^{-1}\approx 0.18$.

\medskip

\textbf{Case $m=4$:} Here, the operators $X_4$ and $Y_4$ take the following form in the eigenbasis of $Z_4$:

\begin{align*}
 X_4=\left[\begin{array}{cccc} 0& \sqrt{3}& 0&0\\ -\sqrt{3}& 0& 2&0\\ 0& -2& 0&\sqrt{3}\\ 0& 0& -\sqrt{3}&0 \end{array}\right]\ ,  \qquad Y_4=i\left[\begin{array}{cccc} 0& \sqrt{3}& 0&0\\ \sqrt{3}& 0& 2&0\\ 0& 2& 0&\sqrt{3}\\ 0& 0& \sqrt{3}&0 \end{array}\right]\pl .
\end{align*}
Once again, the spectra of $X_4$ and $Y_4$ coincide with $\{3i,i,-i,-3i\}$, and their eigenvectors are
\begin{align*}
&\ket{x_1}=\frac{1}{2\sqrt{2}}(1,\sqrt{3},\sqrt{3},1),\  &\ket{x_2}=\frac{1}{2\sqrt{2}}(-\sqrt{3},-1,1,\sqrt{3}),
\\ &\ket{x_3}=\frac{1}{2\sqrt{2}}(\sqrt{3},-1,-1,\sqrt{3}),\ &\ket{x_4}=\frac{1}{2\sqrt{2}}(-1,\sqrt{3},-\sqrt{3},1),
\\ &\ket{y_1}=\frac{1}{2\sqrt{2}}(i,-\sqrt{3},-i\sqrt{3},1),\  &\ket{y_2}=\frac{1}{2\sqrt{2}}(-i\sqrt{3},1,-i,\sqrt{3}),
\\ &\ket{y_3}=\frac{1}{2\sqrt{2}}(i\sqrt{3},1,i,\sqrt{3}),\ &\ket{y_4}=\frac{1}{2\sqrt{2}}(-i,-\sqrt{3},i\sqrt{3},1)\,.
\end{align*}
Now, the matrix $O_{\mathcal{X},\mathcal{Y}}$ is equal to
\[O_{\mathcal{X},\mathcal{Y}}=\frac{1}{8}
\left[\begin{array}{cccc} -1& 1& 1& -1 \\ 1& -1& -1& 1 \\
1& -1& -1& 1 \\
-1& 1& 1& -1 \end{array}\right]\  , \quad \la_1= \frac{1}{8}\,,\quad  \la_2=\frac{1}{2}\,,
\]
where $\la_1$ and $\la_2$ are defined as in Corollary \ref{cor:twobasis}. Then, the approximate tensorization constants in Corollary \ref{cor:twobasis} are equal to
\begin{align*}
c_1=32,\  c_2=7,\  c_3=2(1-\frac{1}{64}(2\ln 2-1)^{-1})^{-1}\approx 2.084<2.1 \,.
\end{align*}
Once again, the constant $c_3$ provides us with the tightest bound. Moreover, by \eqref{schur},
\[\al_{\operatorname{CMLSI}}(\mathcal{L}_{X_4})=\al_{\operatorname{CMLSI}}(\mathcal{L}_{Y_4})\ge \frac{1}{\ln 8}\pl.\]
Then we get the following bound on the CMLSI constant of $\cL_4^H$:
\begin{align*}D(\rho\|\rho_\cN)\le & 2.1 (D(\rho\|\rho_X) + D(\rho\|\rho_Y))
 \le
(2.1\ln 8) (\operatorname{EP}_{\cL_{X_4}}(\rho) + \operatorname{EP}_{\cL_{Y_4}}(\rho))\pl.
\end{align*}
Namely, $\al_{\operatorname{CMLSI}}(\cL_4^H)\ge (2.1\ln 8)^{-1}\approx 0.22$

\medskip

\textbf{General case $m>4$:} One can numerically see that $\la_1\le \frac{1}{5}$ and
\begin{align*}
 c_3=2(1-\frac{1}{25}(2\ln 2-1)^{-1})^{-1}\approx 2.231<\frac{9}{4}\,.
\end{align*}
Furthermore, since both $X_m$ and $Y_m$ have integer spectrum $\{m-1,m-3,m-5,\cdots,-(m-1)\}$, the generators $\cL_{X_m}$ and $\cL_{Y_m}$ can be transferred from the Laplace operator $\Delta_\mathbb{T}$ on the unit torus $\mathbb{T}$, whose CMLSI constant is estimated in \cite[Theorem 4.12]{brannan2020complete}:
\[\al_{\operatorname{CMLSI}}(\Delta_\mathbb{T})\ge (2\ln 3)^{-1}\]
(note that our normalization differs by a factor of $2$). This by transference implies
\[\al_{\operatorname{CMLSI}}(\mathcal{L}_{X_m})=\al_{\operatorname{CMLSI}}(\mathcal{L}_{Y_m})\ge \frac{1}{2\ln(3)}\pl,\]
which in particular behaves better than the Schur multiplier bound
\[ \frac{1}{2\ln 3}=\frac{1}{\ln 9}\ge \frac{1}{\ln 2m}\]
for $m>4$.
% \begin{align*}D(\rho\|\rho_\cN)\le & \frac{9}{4} (D(\rho\|\rho_X) + D(\rho\|\rho_Y))
%  \le
% (\frac{9}{2}\ln(2m)) (I_X(\rho) + I_Y(\rho))\pl.
% \end{align*}
 Therefore we obtain the following dimension free numerical bound $$\al_{\operatorname{CMLSI}}(\cL_m^H)\ge \Big(\frac{9}{2}\ln(3)\Big)^{-1}\approx 0.2\,,$$
 which can be compared to the bound of \Cref{theo:CLSI} ($\cL_m^H$ has a uniform spectral gap $1$ for all $m$ by transference, see \cite[Proposition 3.1]{Baudoin09}).
It remains open  whether the sub-Laplacian $\Delta_H$ itself satisfies CMLSI.

\section{Local semigroups}\label{sec:localsemigroups}

In this section, we consider symmetric Markov semigroups whose Lindblad operators act on edges of a given graph. Such generators have been extensively studied from the point of view of functional inequalities in the classical setting. (see e.g. \cite{stroock1992equivalence} and the references therein). Here, we use (i) the CMLSI constants found in \Cref{theo:CLSI} to control the local interaction in combination with (ii) the sharpening of the approximate tensorization constant found in \Cref{sec:symmetric} to derive asymptotically tight lower bounds on the CMLSI constant for various models of relevance.

Let $G:=(V,E)$ be a finite, connected and undirected graph with vertex set $V$, of cardinality $|V|=n$, and edge set $E:=\{(v,w)\in V\times V:\,v\sim w\}$. We recall that the degree $\operatorname{deg}(v)$ of a vertex $v\in V$ is the number of edges that are incident to $v$. Moreover, $G$ is said to be $\gamma$-regular if all vertices $v\in V$ have same degree $\gamma$. Important examples include finite groups through the scope of their Cayley graphs. %Conversely, for a graph $G$, one can consider the group of transitions of adjacent vertices that are connected by an edge.
Given a graph $G:=(V,E)$, the graph Laplacian $\Delta_G$ acting on the function spaces $\{f| f:V\to \mathbb{C} \}$ is defined as
\begin{align}\label{Laplaceclass}
    \Delta_G f(v)=\sum_{w:(w,v)\in E}\,(f(w)-f(v))\, ,
\end{align}
$\Delta_G$ is a negative semi-definite matrix on $l_2(V)$, %as a Hilbert space equipped with $\langle f,f\rangle=\sum_{v}|f(v)|^2$.
which generates the heat semigroup $T_t=e^{\Delta_G t}$ on $l_\infty(V)$. Note that here we choose $\Delta_G$ to be negative to match our convention for quantum Lindbladians.
The spectral gap is defined as the gap between the first and second eigenvalue of $\Delta_G$.

 A sequence of $d$-regular graphs $\{G_i=(V_i,E_i)\}_{i\in\mathbb{N}}$ of increasing size $\lim_{i} |V_i|=+\infty$ is called a \textit{family of expander graphs} if there exists $\lambda_0>0$ such that the spectral gaps $\lambda(\Delta_{G_i})\ge \lambda_0$ uniformly for all $i$ \cite{hoory2006expander}. Modified logarithmic Sobolev inequalities for such generators have been widely considered in the classical literature,  (see e.g. the survey \cite{bobkov2003modified}).

There are two natural ways of defining a quantum Markov semigroup that retains the locality structure of the graph.
The first approach considered in \cite{li2020graph} introduces a \textit{quantized graph Laplacian} on the $n\times n$ dimensional matrix space ($n=|V|$) as follows
\begin{align}\label{eq:quantizedgraph}
    \Delta_G^{\operatorname{q}}(\rho):=\sum_{e\in E}2X_e\rho X_e-\,\{X_e^2,\rho\}\,,
\end{align}
where we identify vertices in $V$ with a fixed orthonormal basis $\{|v\rangle\}_{v\in V}$ of $\mathbb{M}_{n}$, and $X_e:=|v\rangle\langle w|-|w\rangle\langle v|$ is the transition operator for $e=(v,w)$. In this approach, the CMLSI constant of $\Delta_G^q$ was proved to be controlled by the CMLSI constant of the classical graph Laplacian $\Delta_G$ in \cite[Theorem 7.7]{li2020graph}:
\begin{align}\label{graphtransfer}
\frac{ \alpha_{\operatorname{CMLSI}}(\Delta_G)}{1+5\pi^2 \alpha_{\operatorname{CMLSI}}(\Delta_G)}  \le   \alpha_{\operatorname{CMLSI}}(\Delta_G^q)\le \alpha_{\operatorname{CMLSI}}(\Delta_G)\,.
\end{align}

A second approach to defining a quantum Markov semigroup with the locality structure of $G=(V,E)$ consists in introducing a local evolution on the $n$-fold tensor product $\cH_V:=\bigotimes_{v\in V}\cH_v$ of a given finite dimensional local Hilbert space $\cH$, namely, a $n$-qudit system for $d=\dim(\cH)$. The Lindblad operators are supported on the edges $e\in E$ as follows,
\begin{align}\label{eq:subsystemLindbladian}
    \mathcal{L}_{G}:=\sum_{e\in E}\,\mathcal{L}_e\,,\qquad\text{ where }\qquad \mathcal{L}_{e}(\rho):=\sum_{j\in J^{(e)}} L^{(e)}_j \rho L^{(e)}_j -\frac{1}{2}\{ L^{(e)}_j L^{(e)}_j,\,\rho\}\,,
\end{align}
where for any edge $e=(v,w)$ and $\{L^{(e)}_j\}_{j\in J^{(e)}}$ are the family of local Lindblad operators that act trivially on subsystems other than
$\cH_{v}\otimes \cH_{w}$. We call \eqref{eq:subsystemLindbladian} a \textit{subsystem Lindbladian} in order to distinguish it from the quantized graph Laplacian introduced in \Cref{eq:quantizedgraph}. In what follows, we  denote by $\cM_V:=\cB(\cH_V)$ (resp. $\cM_e:=\cB(\cH_{v}\otimes \cH_w)$) the algebra of operators on which $\cL_G$ (resp. $\cL_e$ for $e=(v,w)\in E$) acts. We also denote by $E_e$ ( resp. $E_G$) the conditional expectation projecting onto the kernel of $\mathcal{L}_e$ for $e\in E$ (resp, of $\mathcal{L}_G$). Finally, we introduce the indices $$C_G:=C_{\operatorname{cb}}(\cM_V:E_G(\cM_V))\pl ,\pl \quad C_e:=C_{\operatorname{cb}}(\cM_e:E_e(\cM_e)),$$ and $c_G$ as the minimum non-zero eigenvalue of the Choi-Jamiolkowski state $J_{E_G}$ of $E_{G}$. It is also useful to introduce the generator
\begin{align*}
    \widetilde{\cL}_G:=\sum_{e\in E}E_e-\id\,.
\end{align*}

As previously discussed, the lower bounds on the complete modified logarithmic Sobolev constant derived in \Cref{theo:CLSI} are asymptotically not tight when the total dimension is large. For instance, in the case of a primitive semigroup on $n$-qudit systems, the completely bounded Pimsner-Popa index is equal to $d^{n}$, where $d$ is the dimension of the local Hilbert space $\cH$. This gives lower bounds on the CMLSI constant of a subsystem Lindbladian that are exponentially small in the number of vertices. In the next theorem, we essentially leverage the locality structure of $\cL_G$ to provide exponentially tighter bounds by combining \Cref{theo:CLSI}, Corollary \ref{Coro:sharperAT}, Lemma \ref{lemmageneral} and the detectability lemma \cite{aharonov2009detectability,aharonov2010quantum,anshu2016simple,kastoryano2016quantum}
 (see \Cref{sec:detectabilitylemma}).

\begin{theorem}\label{thm:graphs}
Let $G=(V,E)$ be a finite, connected, undirected graph of maximum degree $\gamma$, and let $\mathcal{L}_G$ be a symmetric subsystem Lindbladian of the form \eqref{eq:subsystemLindbladian}. Then for all $m\in\mathbb{N}$ and any state $\rho\in\cD(\cH_V\otimes  \mathbb{C}^{\otimes m})$,
   \begin{align}\label{eq:ATgraphs}
     D(\rho\|E_{G}(\rho)) \le 4\Big\lceil\frac{ (1+\ln(C))}{\ln\big(\frac{\lambda(\widetilde{\cL}_G)}{4(\gamma-1)^2}+1\big)}\Big\rceil\,\sum_{e\in E}\,D(\rho\|E_{e}(\rho))\,,
     \end{align}
 where $C:=\min\{C_G,c_G^{-1}\}$ and  $\lambda(\widetilde{\cL}_G)$ is the spectral gap of $\widetilde{\cL}_G$. Moreover, the $\operatorname{CMLSI}$ constant for the generator $\cL_G$ satisfies
\begin{align}
\label{eq:CMLSIgraphs}
&\frac{1}{4}\Big\lceil\frac{\ln(C)+1}{\ln\big(\frac{\lambda(\widetilde{\cL}_G)}{4(\gamma-1)^2}+1\big)}\Big\rceil^{-1}\min_{e\in E}\frac{\lambda(\cL_e)}{C_e}\le \frac{1}{4}\Big\lceil\frac{\ln(C)+1}{\ln\big(\frac{\lambda(\widetilde{\cL}_G)}{4(\gamma-1)^2}+1\big)}\Big\rceil^{-1}\min_{e\in E}\alpha_{\operatorname{CMLSI}}(\cL_e) \le    \alpha_{\operatorname{CMLSI}}(\mathcal{L}_G)\,.
% \\
% & \frac{1}{4}\Big\lceil\frac{\ln(C)+1}{\ln\big(\frac{\lambda(\widetilde{\cL}_G)}{\gamma^2}+1\big)}\Big\rceil^{-1}\le    \alpha_{\operatorname{CMLSI}}(\widetilde{\mathcal{L}}_G)\,.
\end{align}
% for $\lambda:=\|\Phi^*-E_\cN:L_2(\sigma)\to L_2(\sigma)\|$ and an arbitrary invariant state $\sigma$, with
% \begin{align*}
%     \Phi^*\in\Big\{|E|^{-1}\sum_{e\in E}E_e,\, 2^{-1}\big(\prod_{e\in E} E_e+\big[\prod_{e\in E} E_e\big]^*\big)\Big\}\,,
% \end{align*}
% where the product is with respect to an arbitrary ordering of the conditional expectations. Here, we denoted $C_G:=C_{\tau,\operatorname{cb}}(\cB(\cH^{\otimes n}):E_G(\cB(\cH^{\otimes n})))$,
As a consequence, whenever $\lambda(\widetilde{\cL}_G)$ is uniformly lower bounded by a constant independent of $|V|$, $\alpha_{\operatorname{CMLSI}}(\widetilde{\cL}_G)=\Omega\Big(\frac{1}{\ln(C)}\Big)$, hence recovering the asymptotics of classical expanders.
\end{theorem}
\begin{proof}
% Let us first prove \eqref{eq:ATgraphs}. We let $\Phi^*:=\prod_{e\in E}E_e$, where the ordering of the conditional expectations can be chosen arbitrarily. By definition, each local Lindbladian $\cL_e$ does not commute with at most $\gamma$ other local Lindbladians. Therefore, by a direct application of the detectability lemma (cf. \Cref{thm:detectability}), the generator $\widetilde{\cL}_G:=\sum_{e\in E}\id-E_e$ satisfies \begin{align*}
%   \lambda:= \|\Phi^*-E_G:L_2(\sigma)\to L_2(\sigma)\|^2\le \frac{1}{\frac{\lambda(\widetilde{\cL}_G)}{\gamma^2}+1}\,.
% \end{align*}
% Therefore, by a direct application of \Cref{Coro:sharperAT} with $\eps:=\sqrt{\ln(2)-\frac{1}{2}}$, we have \textcolor{red}{worth adding this directly to the corollary last section for a general map $\Phi$}
%     \begin{align}
%     D(\rho\|E_{\cN*}(\rho))&\le \frac{4(\ln(C_G)+1)}{\ln(\lambda^{-1})}\,\sum_{e\in E}\,D(\rho\|E_{e*}(\rho))\,\\
%     &\le \frac{4(\ln(C_G)+1)}{\ln\big(\frac{\lambda(\widetilde{\cL}_G)}{\gamma^2}+1\big)}\,\sum_{e\in E}\,D(\rho\|E_{e*}(\rho))
%     \end{align}
% where we also used the bound $-2\ln(\ln(2)-\frac{1}{2})\le 4$ in the first line.
We first establish \eqref{eq:ATgraphs}: by Corollary \ref{Coro:sharperAT}, the approximate tensorization constant of the family $\{E_e\}_{e\in E}$ of conditional expectations can be upper bounded by the constant $4k$, where $k$ is the integer such that
\begin{align*}
  (1-\eps)\,E_G\le   \Big(\prod_{e\in E}\,E_e\Big)^k\le (1+\eps)E_G\,,
\end{align*}
for $\eps:= \sqrt{\ln(2)-\frac{1}{2}}$. Here the ordering in the product $\prod_{e\in E}\,E_e$ is arbitrary. Next, we have from Lemmas \ref{lemmageneral} and \ref{lemma:brandao} in \Cref{sec:normstoorder} that $k$ can be chosen as
 \begin{align*}
k= \left\lceil\frac{\ln\frac{1}{\eps}+\ln(C)}{\ln\frac{1}{\lambda}}\right\rceil, \quad\text{ where }\quad \lambda:=\Big\|\prod_{e\in E}E_e-E_G:L_2\to L_2\Big\|,\quad \,C:=\min\{C_G,c_G^{-1}\}\,.
\end{align*}
Finally, the $L_2$-constant $\lambda$ can be controlled by the gap of the generator $\widetilde{\cL}_G$ using Lemma \ref{thm:detectability} of \Cref{sec:detectabilitylemma}:
\begin{align*}
\lambda=\Big\|\prod_{e\in E}E_e-E_G:L_2\to L_2\Big\| \le\frac{1}{\lambda(\tilde{\mathcal{L}}_G)/4(\gamma-1)^2+1}\,,
\end{align*}
where $\gamma$ is the maximum degree of the graph.
Note that $\ln\eps^{-1}\le 1$ for $\eps= (\ln(2)-\frac{1}{2})^{\frac{1}{2}}$.
 To lower bound $\alpha_{\operatorname{CMLSI}}({\cL}_e)$ in terms of its spectral gap and corresponding complete Pimsner-Popa index, we use \Cref{theo:CLSI}.

% \Cref{eq:ATgraphs} is a direct consequence of the approximate tensorization result of Corollary \ref{Coro:sharperAT} for the product map $\Phi^*$ and \Cref{thm:detectability}. \Cref{eq:CMLSIgraphs} follows from a direct application of \Cref{theo:CLSI} for each edge generator $\cL_e$.
\end{proof}

\begin{rem}In the limit of large expander graphs (i.e.~$|V|\to\infty$, $\lambda(\widetilde{\cL}_G)\ge \lambda_0>0$),
using Corollary \ref{Coro:sharperAT} for the average map $\Phi=\frac{1}{|E|}\sum_{e}E_e$ instead of resorting to the detectability lemma for the product map $\Phi=\prod_{e\in E}\,E_e$ would lead to an asymptotically weaker bound. This is due to the dependence on $|E|$ of the quantity
\begin{align}\label{eq:fromlambdatogap}
\lambda:= \Big\|\frac{1}{|E|}\sum_{e\in E} E_e-E_G:L_2(\sigma)\to L_2(\sigma)\Big\|^2\equiv 1-\frac{1}{|E|}\lambda(\widetilde{\cL}_G)\ .
\end{align}
\end{rem}

% \textcolor{red}{check again: detectability lemma should give something stronger from the gap condition on $\sum_i E_i-\id$, since it is independent of $n$ whereas to go from the gap of the generator to bounding the sum map, we get something like $1-1/n$, which becomes trivial as $n\to\infty$. If so, we should emphasize this}

% \begin{rem}
% Similar bounds were recently proven in \cite{laracuente2019quasi}, where instead of using the detectability lemma to control the approximate tensorization constant, transference techniques were used. Both methods provide equivalent asymptotics. \textcolor{red}{try to understand Nick's path, seems like he can control everything only in terms of $\Phi\le c E$, which is better than what we do in all our results}
% \end{rem}

In the following subsections, we illustrate the bounds derived in \Cref{thm:graphs} on some well-known models.
\subsection{Random permutations}\label{sec:randomperturb}
We consider quantum Markov semigroups introduced in \cite[Section IV.D]{bardet2019group} which represent the action of a random transposition gate applied to two registers $i,j$ on the $n$ qudit system $\cH_V=\cH^{\otimes n}$ with $\dim(\cH)=d$. Let $G=(V,E)$ be a finite graph with $|V|=n$. Denote the swap gate $S_{i,j}$ acting on registers of vertex $i$ and $j$ as
\begin{align}\label{eq:swap}
    S_{i,j}(|\psi\rangle \otimes |\varphi\rangle)=|\varphi\rangle\otimes |\psi\rangle\,,
\end{align}
for any two $|\psi\rangle,|\varphi\rangle\in\cH$. The generator of the \textit{quantum nearest neighbour random transposition} model \cite[Section IV.D]{bardet2019group} is defined as
\begin{align*}
    \cL^{\operatorname{NNRT}}_G(\rho):=\frac{1}{2}\sum_{(i,j)\in E}(S_{i,j}\rho S_{i,j}-\rho)\,.
\end{align*}
The above generator can simply be understood as that of the natural action of the permutation group $\mathcal{S}_n$ on $\cH_V$, which allows  infinitesimal transitions between random adjacent registers connected by edges. In other words, $\cL^{\operatorname{NNRT}}_G$ is the subsystem Lindbladian of the graph $G=(V,E)$ with the local Lindbladian at edge $(i,j)\in E$ given by
$$\mathcal{L}_e(\rho)=S_{i,j}\rho S_{i,j}-\rho=S_{i,j}\rho S_{i,j}-\frac{1}{2}\{S_{i,j}^2,\rho\}\ .$$
This corresponds to a classical random transposition on the permutation group $\mathcal{S}_n$ on $[n]:=\{1,\cdots, n\}$. Denote $\si_{ij}\in \mathcal{S}_{n}$ as the $2$-permutation switching $i$ and $j$. We consider the following classical generator on $\mathcal{S}_n$
\[\Delta_{\mathcal{S}_n}^G f(\si)=\frac{1}{2}\sum_{(i,j)\in E} (f(\si_{ij}\si)-f(\si))\ ,\quad \si \in \mathcal{S}_n\ .\]
\begin{proposition}\label{bounds:CMLSIRT} In the above notations, $\cL^{\operatorname{NNRT}}_G$ generates a semigroup transferred from $\Delta_{\mathcal{S}_n}^G$ on $\cB(\cH_V)$.
\end{proposition}
\begin{proof}It suffices to verify the transference diagram for the generator
\begin{align*}
 \begin{array}{ccc}  L_{\infty}(\mathcal{S}_n, \cB(\cH) )\pl\pl &\overset{ \Delta_{\mathcal{S}_n}^G\ten \id_{\cB(\cH_V)}}{\longrightarrow} & L_{\infty}(\mathcal{S}_n, \cB(\cH) ) \\
                    \uparrow \pi    & & \uparrow \pi  \\
                    \cB(\cH)\pl\pl &\overset{\mathcal{L}^{\operatorname{NNRT}}_G}{\longrightarrow} & \cB(\cH)
                     \end{array} \pl .
                     \end{align*}
where
$\pi(\rho)(\si)=S_{\si}^\dag \rho S_{\si}$
is the transference map defined in \Cref{transferrencemap}. Here the permutation gate $S_\si$ is defined as the composition $S_\si=S_{\si_{i_1,j_1}}\cdots S_{\si_{i_k,j_k}}$ for $\si= \si_{i_1,j_1}\cdots \si_{i_k,j_k}$. Indeed,
\begin{align*}
\big(\Delta^G_{\mathcal{S}_n}\ten \id_{\cB(\cH_V)}\big)\circ \pi(x)(\si)=&
\frac{1}{2}\sum_{(i,j)\in E} \pi(x)(\si_{i,j}\si)-\pi(x)(\si)
\\ =& \frac{1}{2}\sum_{(i,j)\in E} S_{\si_{i,j}\si}^\dag x (S_{\si_{i,j}\si})-S_{\si}^\dag \rho S_{\si}
\\ =& \frac{1}{2}\sum_{(i,j)\in E} S_{\si}^\dag S_{i,j} x S_{i,j} S_{\si}-S_{\si}^\dag \rho S_{\si}
\\ =& S_{\si}^\dag\Big(\frac{1}{2}\sum_{(i,j)\in E} S_{i,j} x S_{i,j} - \rho \Big)S_{\si}=\pi(\mathcal{L}^{\operatorname{NNRT}}_G(x)) (\si) \ .
\end{align*}
\end{proof}

It follows from the transference principle of \Cref{sec:transference} that $\al_{\operatorname{CMLSI}}(\Delta_{\mathcal{S}_n}^G)\le \al_{\operatorname{CMLSI}}(\mathcal{L}^{\operatorname{NNRT}}_G)$ and similarly for spectral gap $\la(\Delta_{\mathcal{S}_n}^G)\le \la(\mathcal{L}^{\operatorname{NNRT}}_G)$. In particular, these  lower bounds on $\al_{\operatorname{CMLSI}}(\mathcal{L}^{\operatorname{NNRT}}_G)$ and $\la(\mathcal{L}^{\operatorname{NNRT}}_G)$ are independent of the local dimension $d=\dim(\cH)$. Whenever the CMLSI constant or the spectral gap of $\Delta_{\mathcal{S}_n}^G$ are known, the transference gives strong estimates.

\begin{example}{(Random transposition) }{\rm  Consider the (full)
\textit{quantum random transposition model} and
\begin{align*}
    \mathcal{L}^{\operatorname{RT}}(\rho):=\frac{1}{2}\sum_{i\ne j}\,(S_{i,j}\rho S_{i,j}-\rho)\,.
\end{align*}
This corresponds to $G=(V,E)$ being the complete graph. The corresponding classical random transposition model on $\mathcal{S}_n$
\[\Delta_{\mathcal{S}_n} f(\si)=\frac{1}{2}\sum_{i\neq j} (f(\si_{i,j}\si)-f(\si))\ .\]
was well-studied. $\Delta_{\mathcal{S}_n}$ was proven to have spectral gap $\lambda(\Delta_{\mathcal{S}_n})=n$ \cite{diaconis1981generating}
and MLSI constant \cite{gao2003exponential} (see also \cite{bobkov2003modified} where slightly worse bounds were derived using different techniques) \begin{align*}
\frac{n}{2}\le     \alpha_{\operatorname{MLSI}}(\Delta_{\mathcal{S}_n})\le 2n\,.
\end{align*}
using the martingale method of \cite{lee1998logarithmic}. Recently, the CMLSI constant of the generator $\Delta_{\mathcal{S}_n}$ was shown in \cite{li2020complete} to satisfy the same bounds
\begin{align*}
 \frac{n}{2} \le  \alpha_{\operatorname{CMLSI}}(\Delta_{\mathcal{S}_n})\le 2 n \,.
\end{align*}
By Proposition \ref{bounds:CMLSIRT}, we thus have
\begin{align*}
 \frac{n}{2}\le    \alpha_{\operatorname{CMLSI}}(\cL^{\operatorname{RT}})\le 2 n\,.
\end{align*}
The upper bounds follows from the spectral gap and the fact the representation $\mathcal{S}_n$ on $\cH_{V}=\cH^{\ten n}$ contains all irreducible component \cite{Audenaert06}.
For this example, the degree of a vertex is $\frac{n(n-1)}{2}$, which scales quadratically with $n$. As we will see below, \eqref{eq:CMLSIgraphs} would provide asymptotically worse bounds than the transference method of Proposition \ref{bounds:CMLSIRT}.
}
\end{example}

For a general graph $G=(V,E)$, $\Delta_{\mathcal{S}_n}^G$ is the graph Laplacian of the Cayley graph of $\mathcal{S}_n$ with generating set $\{\si_{ij}\  |  \ (i,j)\in E \}$.
It was proved in \cite{li2020graph} that
for a graph Laplacian $\Delta_G$,
\[\al_{\operatorname{CMLSI}}(\Delta_G)\ge \frac{2}{45\,dl^2}\]
where $\gamma$ is the maximum degree of $G$ and $l$ is the number of egdes of a minimum spanning tree, which is $l=|G|+1$. Here, we have $V=\mathcal{S}_n$ with $|\mathcal{S}_n|=n!$ growing exponentially.
This exponential growth also appears if we use Theorem \ref{theo:CLSI}
\[ \frac{\la(\Delta_{\mathcal{S}_n}^G)}{n!}\le \al_{\operatorname{CMLSI}}(\Delta_{\mathcal{S}_n}^G) \le 2\la(\Delta_{\mathcal{S}_n}^G)\,.\]
since $C_{\operatorname{cb}}(l_\infty(\mathcal{S}_n);\mathbb{C})=|\mathcal{S}_n|=n!$. We show in the following that Theorem \ref{thm:graphs} gives a lower bound on the CMLSI constant for $\mathcal{L}^{\operatorname{NNRT}}_G$ that has exponentially better dependence of $|G|=n$ (and is also independent of $d=\dim(\cH)$).

\begin{corollary}
Let $G=(V,E)$ be a connected finite graph and let $\mathcal{L}_{\operatorname{NNRT}}^G$ be the generator of the quantum nearest neighbour random transposition model defined as above. Then
\begin{align*}
    \alpha_{\operatorname{CMLSI}}(\cL^{\operatorname{NNRT}}_G)\ge \frac{1}{4}\Big\lceil\frac{(\ln(n!)+1)}{\ln\big(\frac{\lambda(\cL^{\operatorname{NNRT}}_G)}{4(\gamma-1)^2}+1\big)}\Big\rceil^{-1}\,,
\end{align*}
where $\gamma$ is the maximal degree of $G$ and $\la(\cL^{\operatorname{NNRT}}_G)$ is the spectral gap.
\end{corollary}

\begin{proof}
%\textcolor{red}{left to be done: figure out which constant is better}
We first note that for each edge, $$\mathcal{L}_e(\rho)=\frac{1}{2}(S_{i,j}\rho S_{i,j}-\rho)=E_{i,j}(\rho)-\rho,$$ where
$E_{i,j}(\rho)=\frac{1}{2}(S_{i,j}\rho S_{i,j}+\rho)$ is a conditional expectation onto the symmetric space on $\cH_{i}\ten \cH_j$. Then
given the second bound derived in \Cref{thm:graphs}, it suffices to calculate the index $C_G:=C_{\operatorname{cb}}(\cB(\cH_V):\cN)$, where $\cN$ is the fixed point subalgebra of $\cL^{\operatorname{NNRT}}_G$. Since $G$ is connected, then $\{\si_{i,j}\ |\ (i,j)\in E \}$ is a generating set
for $\mathcal{S}_n$. Thus $\cN$ is the commutant of the representation
$$\pi:\mathcal{S}_n\to\cB(\cH_V)\ ,\ \quad  \pi(\si_{i,j})=S_{i,j}.$$
As discussed in Example \ref{exam:urep} of \Cref{sec:normstoorder},
the index is $$C_{\operatorname{cb}}(\cB(\cH_V):\cN)=\sum_{\pi_i\in \operatorname{Irr}(\mathcal{S}_n), \pi_i\subset \pi}m_i^2,$$
where $m_i$ is the dimension of irreducible representation $\pi_i$ and the summation is over all irreducible representations in the decomposition of $\pi=\oplus_{i}\pi_i\ten \id_{{n_i}}$. By the expression provided in \eqref{formula}, we know all irreducible representations (up to unitary equivalent) are contained in $\pi$. Then by Schur-Weyl Theorem,
\[C_G=C_{\operatorname{cb}}(\cB(\cH_V):\cN)=\sum_{\pi_i\in \operatorname{Irr}(\mathcal{S}_n)}m_i^2=|\mathcal{S}_n|=n! \ .\]
%$the index turns out to be $|\mathcal{S}_m|=m!$ (see e.g.~\cite{Audenaert06}).
%first need to estimate the index $C_{\mathcal{S}_m}:=C_{\tau,\operatorname{cb}}(\cB(\cH^{\otimes m}):E_G(\cB(\cH^{\otimes m})))$. First, we know that $E_G(\cB(\cH^{\otimes m})):=\{S_{i,j}:(i,j)\in E\}'=\{\mathcal{S}_{m,d}\}'$, i.e.~the commutant of the tensor representation $\mathcal{S}_{m,d}$ of $\mathcal{S}_m$ on $\cH^{\otimes m}$, where $\cH\simeq \mathbb{C}^d$. By the expression provided in \eqref{formula}, the index turns out to be $|\mathcal{S}_m|=m!$ (see e.g.~\cite{audenaert2006digest}).
% \textcolor{red}{need to compute gap, then use AT plus CMLSI=1 for conditional expectations}
\end{proof}

\subsection{Approximate unitary designs}\label{sec:designs}
In this subsection, we consider local quantum Markov semigroups converging to the Haar unitary $k$-designs over the unitary group $\operatorname{U}(d^n)$:
\begin{align}\label{eq:asymptoticsdesigns}
    e^{t\cL^{(k)}}(\rho)\underset{t\to \infty}{\to} \mathcal{D}^{(k)}_{\operatorname{Haar}}(\rho):=\mathbb{E}_{\operatorname{Haar}(d^n)}\big[ U^{\otimes k}\rho (U^\dagger)^{\otimes k} \big]\,.
\end{align}
A generator $\cL^{(k)}$ with the above asymptotic behavior can be understood as a continuous-time version of a local random quantum circuit \cite{brandao2016local}. Such circuits are of great importance in various areas of quantum information, quantum computing and physics (see e.g. \cite{brandao2016local,onorati2017mixing} and the references therein). Here, we mainly consider two classes of local quantum Markov semigroups with the property \eqref{eq:asymptoticsdesigns}. Before defining their generators in more detail, we
recall the notion of universal (gate) set and distributions.
\begin{Def}\label{def:universalmeasure}
$\operatorname{(i)}$ Let $\mathcal{U}=\{U_1,\cdots,U_m\}$ be a finite subset in the group $\operatorname{U}(N)$ (resp. $\operatorname{SU}(N)$) of $N\times N$ unitary (resp. special unitary) matrices. The set $\mathcal{U}$ is said to be universal if the subgroup generated by $\mathcal{U}$ is dense in $\operatorname{U}(N)$ (resp. $\operatorname{SU}(N)$).\\
$\operatorname{(ii)}$ Let $\mu$ be a Borel probability measure on $\operatorname{U}(N)$ (resp. $\operatorname{SU}(N)$). Then $\mu$ is said to be universal if for all $V\in \operatorname{U}(N)$ (resp. $\operatorname{SU}(N)$) and any $\delta>0$, there exists a positive integer $\ell$ such that
\begin{align*}
    \mathbb{P}_{\mu^{\otimes \ell}} (U_1\ldots U_\ell \in B_\delta(V))>0\,,
\end{align*}
where $B_\delta(V)$ is the open ball around $V$ of radius $\delta>0$ and the distance can for instance be chosen as the geodesic distance.
\end{Def}
\noindent The (ii) above is a probabilistic generalization of (i) (see also \cite{harrow2009random,onorati2017mixing}). Indeed, it follows from the definition that for a universal set $\mathcal{U}=\{U_1,\cdots,U_m\}$, the discrete probability measure
\[ \mu_{\mathcal{U}}=\frac{1}{m}\sum_{i=1}^m \delta_{U_i}\]
is a universal measure supported on $\mathcal{U}$.

Next, given a probability measure $\mu$ on $\operatorname{U}(N)$ and an integer $k$, we define the unital channel $\Phi^{(k)}_\mu$ on $\cB(\mathbb{C}^{N^k})$ as
\begin{align*}
    \Phi^{(k)}_\mu:\rho\mapsto \mathbb{E}_\mu\big[U^{\otimes k} \rho\,(U^\dagger)^{\otimes k} \big]\,.
\end{align*}
When the measure $\mu$ is universal in the sense of Definition \ref{def:universalmeasure}, it was shown in \cite[Lemma 3.7]{harrow2009random} that for each $k\in \mathbb{N}$, there exists $\eps(k)<1$ such that
$  \|\Phi_\mu^{(k)}-\cD_{\operatorname{Haar}}^{( k)}:L_2\to L_2\|\le\eps(k)<1$, which implies that $(\Phi_\mu^{(k)})^{\ell}\to \cD_{\operatorname{Haar}}^{(k)}$ as $\ell\to\infty$. This first convergence result opened the way to the study and refinements of the speed of convergence of such Markovian dynamics. Below,
%It follows from the deep result by Bourgain and Gamburd \cite{BG08,BG12} that for every symmetric universal set $\{U_1,\cdots,U_m\}$ and , $\|\Phi_\mu^{(k)}-\cD_{\operatorname{Haar}}^{( k)}:L_2\to L_2\|<\la<1$ for a $\la$ independent
%When the measure $\mu$ is universal in the sense of Definition \ref{def:universalmeasure}, it was shown in \cite[Lemma 3.7]{harrow2009random} that $  \|\Phi_\mu^{(k)}-\cD_{\operatorname{Haar}}^{( k)}:L_2\to L_2\|<\la(k)<1$, which implies that $(\Phi_\mu^{(k)})^{\ell}\to \cD_{\operatorname{Haar}}^{(k)}$ as $\ell\to\infty$.
we collect various notions of closeness to $\cD_{\operatorname{Haar}}^{(k)}$ considered in the literature \cite{low2010pseudo,brandao2016local} and introduce a new entropic variant. We recall that the diamond norm $\norm{\cdot}{\diamond}$ of a map $\Phi$ is defined as
    \[\norm{\Phi}{\diamond}\,:=\,\norm{\Phi\ten \id_\cH:\mathcal{T}_1(\cH\ten \cH)\to \mathcal{T}_1(\cH\ten \cH)}{}\pl.\]
    where $\mathcal{T}_1(\cH)$ is the space of trace class operators on a Hilbert space $\cH$.
\begin{Def}\label{def:designs}
Let $\mu$ be a probability measure on $\operatorname{U}(N)$. Then $\mu$ is
\begin{itemize}
    \item[(i)] an $(N,\eps,k)$-quantum tensor product expander, $\operatorname{TPE}$ in short, if
    \begin{align}
        \|\Phi_\mu^{(k)}-\cD_{\operatorname{Haar}}^{( k)}:L_2\to L_2\,\|\le \eps\,.
    \end{align}
    We denote by $c_2(\mu,k)$ the smallest constant $\eps$ achieving the above bound;
    \item[(ii)] an $\eps$-approximate unitary $k$-design in $\operatorname{cp}$-order (completely positive order) if
    \begin{align*}
      (1-\eps) \cD_{\operatorname{Haar}}^{(k)}  \le_{\operatorname{cp}}\Phi_\mu^{(k)}\le_{\operatorname{cp}}(1+\eps) \cD_{\operatorname{Haar}}^{(k)}\,.
    \end{align*}
    We denote by $c_{\operatorname{cp}}(\mu,k)$ the smallest constant $\eps$ achieving the above bound;
    \item[(iii)] an $\eps$-approximate unitary $k$-design in diamond norm if
    \begin{align*}
         \|\Phi_\mu^{(k)}-\cD_{\operatorname{Haar}}^{( k)}\|_{\diamond}\le \eps\,.
    \end{align*}
    We denote by $c_{\diamond}(\mu,k)$ the smallest constant $\eps$ achieving the above bound.
\item[(iv)] a complete $\eps$-approximate unitary $k$-design in relative entropy if for all $m\in \mathbb{N}$ and any state $\rho\in\cD(\mathbb{C}^{N^{km}})$,
\begin{align*}
    D((\Phi_\mu^{(k)}\otimes \id_{m})(\rho)\|(\cD_{\operatorname{Haar}}^{(k)}\otimes \id_{m})(\rho) )\le \eps \,D(\rho\|(\cD_{\operatorname{Haar}}^{(k)}\otimes \id_{m})(\rho))\,.
\end{align*}
We denote by $c_{\operatorname{Ent}}(\mu,k)$ the smallest constant $\eps$ achieving the above bound.
\end{itemize}
\end{Def}
\begin{lemma}\label{lemma:eps}
In the notations of Definition \ref{def:designs}, for $\eps<1$,
\begin{align}\label{eq:old}
\frac{c_{\operatorname{cp}}(\mu,k)}{N^{2k}}    \le c_{\diamond}(\mu,k)\le 2c_{\operatorname{cp}}(\mu,k)\,,\qquad\frac{c_2(\mu,k)}{2N^{k/2}}\le c_{\operatorname{cp}}(\mu,k)\le N^{2k}c_2(\mu,k)\,
\end{align}
\begin{align}\label{eq:new}
c_{2}(\mu,k)^2\le c_{\operatorname{Ent}}(\mu,k)\le c_{\operatorname{cp}}(\mu,k)\,, \qquad c_{\diamond}(\mu,k)^2\le 2c_{\operatorname{Ent}}(\mu,k)\,\ln\binom{k+2N-1}{k}\,.
\end{align}
\end{lemma}

\begin{proof}
The two relations in \eqref{eq:old} can be found in Lemmas 3 and 4 of \cite{brandao2016local}. For the first bounds in
\eqref{eq:new}, note that $c_{\text{Ent}}(\mu,k)$ is the CSDPI constant for the channel $\Phi_\mu^{(k)}$. That
$ c_{\text{Ent}}\le c_{\operatorname{cp}}$ follows from convexity and monotonicity of the relative entropy, whereas $c_{2}^2\le c_{\text{Ent}} $ is the lower bound in Corollary \ref{cor:SDPI}. The second bound in \eqref{eq:new} follows from Pinsker inequality and the expression for the cb-index.
Indeed, by the discussion in Example \ref{exam:urep} of \Cref{sec:normstoorder},
\[C_{\operatorname{cb}}\big(\cB(\mathbb{C}^{N^k}):\cD_{\operatorname{Haar}}^{(k)}(\cB(\mathbb{C}^{N^k}))\big)=\sum_{i}m_i^2\,,\]
where $m_i$ is the dimension of an irreducible representation of $\operatorname{GL}(N)$ contained in the $k$-tensor power representation $\pi_k:U\to U^{\ten k}$ on $(\mathbb{C}^{N})^{\ten k}$. It follows from the combination of \cite[Equation 3.5]{Halverson2019set} and \cite[Example 5]{Audenaert06} that
\begin{align*}
   C_{\operatorname{cb}}\big(\cB(\mathbb{C}^{N^k}):\cD_{\operatorname{Haar}}^{(k)}(\cB(\mathbb{C}^{N^k}))\big)=\,\binom{k+2N-1}{k}\, .
\end{align*}
\end{proof}
From now on, we always assume that the measures $\mu$ considered are symmetric and supported on $\operatorname{U}(N)$, which implies the self-adjointness of the maps $\Phi_\mu^{(k)}$ with respect to the Hilbert-Schmidt inner product. We introduce several classes of subsystem Lindbladians which turn out to generate approximate unitary designs in relative entropy. Here, we fix a graph $G=(V,E)$ and its associated $n$ qudit system $\cH_V=\otimes_{j\in V}\cH_j$. We consider the total Hilbert space $\cH_V^{\ten k}=\otimes_{j\in V}\cH_j^{\ten k}$ which has dimension $d^{nk}$. The first class can be understood as a continuous-time extension of the maps considered in \cite{brandao2016local}: given a universal measure on $\operatorname{SU}(d^2)$:
\begin{align}\label{def:Lmu}
    \cL_\mu^{(k)}:=\sum_{e\in E}(\Phi_\mu^{(k)})_{e}-\id\,,
\end{align}
where for each edge $e=(i,j)\in E$, the map $(\Phi_\mu^{(k)})_{e}\equiv \Phi_\mu^{(k)}$ is chosen to act locally on the $2k$-qudit system $\cB(\cH_e^{\ten k})\cong\cB(\cH_i^{\ten k}\ten \cH_j^{\ten k})$.
%Namely, the local Lindbladian is $\cL_{e}^{(k)}=\Phi_\mu^{(k)}-\id$.
%for each each $e\in E$, the map $(\Phi_\mu^{(k)})_{e}\equiv \Phi_\mu^{(k)}$ is chosen to act locally on the $2k$-qudit system $\cM^{(k)}_e:=\cB(\cH_e^{\otimes k})$.
By universality of $\mu$, we know from \cite[Lemma 3.7]{harrow2009random} that for each edge $e$,
% $(\Phi_\mu^{(k)})^{\ell}$ converges to the local Haar unitary design  as $\ell\to \infty$, which we refer to as $E_e^{(k)}$, i.e.~:
\begin{align*}
(\Phi_\mu^{(k)})_e^{\ell}\underset{\ell\to\infty}{\to} E_e^{(k)}(\rho):=\cD_{\operatorname{Haar},e}^{(k)}(\rho)\equiv \mathbb{E}_{\operatorname{Haar}(d^2)}\big[ (U_e\otimes \Id_{e^c})^{\otimes k}\,\rho\, (U_e^\dagger \otimes \Id_{e^c})^{\otimes k}\big] \,.
\end{align*}
In the specific case when the local measure $\mu$ is taken to be the Haar measure, we denote the corresponding subsystem Lindbladian by
\begin{align}\label{def:Lhaar}
    \cL_{\operatorname{Haar}}^{(k)}:=\sum_{e\in E}E_e^{(k)}-\id\,.
    \end{align}
Since the special unitary group on $\cH_V$ is generated by unitaries acting on each edge system $\cH_e$, the fixed point space is equal to $$\cap_{e\in E}\{E_e^{(k)}(\cB(\cH_V^{\ten k}))\}=\cD_{\operatorname{Haar}}^{(k)}(\cB(\cH_V^{\ten k}))\,,$$ which implies the convergence \eqref{eq:asymptoticsdesigns} for the semigroups generated by $\cL_\mu^{(k)}$ and $\cL_{\operatorname{Haar}}^{(k)}$ by basic ergodic theory \cite{frigerio1982long}.
In the case of a nearest neighbour graph, the spectral gap of the latter was considered in \cite[Theorem 5]{brandao2016local} by mapping it to the frustration free, local Hamiltonian of a nearest neighbour spin chain, whose ground space is spanned by matrix-product states \cite{fannes1992finitely,perez2006matrix}. Such Hamiltonians were previously considered in \cite{nachtergaele1996spectral} where their gap was shown to be controlled by its value on local subregions. Combining this fact with a local control of the constant by means of a path coupling method \cite{bubley1997path} already used in \cite{Oliveira2009} to study the Wasserstein convergence of the Kac model on the unitary group, \cite{brandao2016local} found:
\begin{align}\label{haar}
    \lambda(\cL_{\operatorname{Haar}}^{(k)})\ge \Big(42500\left\lceil \frac{\ln(4k)}{\ln(d)}\right\rceil^2\,d^2\,k^5\,k^{\frac{3.1}{\ln(d)}} \Big)^{-1}\,.
\end{align}

%\textcolor{red}{missing: gap of $\Phi_\mu^{(k)}$, relation to Bourgain's paper mentioned in Brandao?}
For our third example, we take $\mu$ to be a discrete measure. For a universal gate set $\mathcal{U}=\{U_1,\cdots,U_m\} \subset \operatorname{SU}(N)$, we denote the unital channel for the uniform discrete measure on $\mathcal{SU}$ as
 \begin{align}
    \Phi^{(k)}_\mathcal{U}:\rho\mapsto \frac{1}{m}\sum_{i=1}^m U_i^{\otimes k} \rho\,(U_i^\dagger)^{\otimes k}  \,.\label{eq:discretedesign}
\end{align}
It was a beautiful result by Bourgain and Gamburd \cite{BG08,BG12} that for every symmetric universal set $\{U_1,\cdots,U_m\}\subset \operatorname{SU}(N)$ with each
$U_i$ composed of algebraic entries, the Hecke operator
\[ H_\mathcal{U}: L_2(\operatorname{SU}(N))\to L_2(\operatorname{SU}(N))\ ,\qquad H_\mathcal{U}f(g)=\frac{1}{m}\sum_{i=1}^mf(U_i\,g) \ ,\]
has a spectral gap. Namely, \begin{align}\label{eq:BGgap}\varepsilon(\mathcal{U}):=\norm{H_\mathcal{U}-\mathbb{E}_{\operatorname{Haar}(N)}:L_2(\operatorname{SU}(N))\to L_2(\operatorname{SU}(N)) }{}<1\end{align} where $\mathbb{E}_{\operatorname{Haar}(N)}$ denotes the mean corresponding to the Haar measure. It then follows from the transference method of Proposition \ref{prop:transference} that $$\|\Phi_{\mathcal{U}}^{(k)}-\cD_{\operatorname{Haar}}^{( k)}:L_2\to L_2\,\|\le\varepsilon <1$$ independently of $k$. This was already used in \cite{brandao2016local} to obtain estimates for approximate random unitary designs constructed from \eqref{eq:discretedesign}. %Since $\operatorname{U}(N)$ is a double cover of $\operatorname{SU}(N)$, it is easy to see that a similar result applies to $\operatorname{U}(N)$.
For a universal set $\mathcal{U}\subset \operatorname{SU}(N)$,  we define the associated subsystem Lindbladian  as
 \begin{align}\label{def:LU}
    \cL_{\mathcal{U}}^{(k)}:=\sum_{e\in E}(\Phi_{\mathcal{U}}^{(k)})_e-\id\,.
    \end{align}
Thanks to Bourgain and Gamburd's spectral gap theorem, if $\mathcal{U}$ is symmetric and each
$U_i\in U$ contains only algebraic entries,
the local Lindbladians $(\Phi_{\mathcal{U}}^{(k)})_e-\id$ admit a uniform spectral gap independent of $k$, which we denote by $\la(\mathcal{U}):=1-\eps(\mathcal{U})$

Another family of a subsystem Lindbladians satisfying the asymptotic behavior \eqref{eq:asymptoticsdesigns} corresponds to local Brownian motions on the unitary group \cite{onorati2017mixing}:
\begin{align}\label{heatsemigroup}
    \cL^{(k)}_{\operatorname{Heat}}:=\,\sum_{e\in E} \,\cL_{e}^{(k)} \,,
\end{align}
where for each edge $e\equiv(i,j)\in E$, $\cL_{e}^{(k)}$ is the quantum Markov semigroup transferred from the heat semigroup on the unitary group $\operatorname{U}(d^2)$ over $2$ qudits via the $k$-th tensor power representation. More precisely, let $\{A^{(i)}\}_{i=1}^{d^4-1}$ be an orthonormal basis of the Lie algebra $\mathfrak{su}(d^2)$\footnote{The orthonormality of the basis $\{A^{(i)}\}_{i=1}^{d^4-1}$ is not necessary and was in fact not required in \cite{onorati2017mixing}. Here, we assume it for sake of conciseness of our exposition.} and let $\pi_{e,k},k:U\mapsto U^{\otimes k}$ be the representation of $\operatorname{U}(d^2)$ onto $\cH_{e}^{\otimes k}\cong (\cH_i\ten \cH_j)^{\otimes k}$,
\begin{align}\label{localbrownianmotiongap}
    \cL_{e}^{(k)}(\rho):=\cL_{\operatorname{Heat},e}^{(k)}=a\sum_{i=1}^{d^4-1}2\pi_{e,k}(A^{(i)})\,\rho\,\pi_{e,k}(A^{(i)})-\{\pi_{e,k}(A^{(i)})^2,\rho\}\,,
\end{align}
for some fixed parameter $a>0$. Here with a slight abuse of notations, we used $\pi_{e,k}$ to also denote the Lie algebra representation of $\mathfrak{su}(d^2)$ onto $\cB(\cH_{e}^{\otimes k})$. That is, for any $A\in\mathfrak{su}(d^2)$,
\begin{align*}
    \pi_{e,k}(A):=\sum_{j=1}^k A_j\otimes \Id_{j^c}\,,
\end{align*}
where $A_j\equiv A$ acts on $k$-th tensor copy of $\cH_e$. The gap $\lambda(\cL^{(k)}_{\operatorname{Heat}})$ was estimated in the proof of \cite[Theorem 9]{onorati2017mixing} in order to derive a lower bound on the time needed for the semigroup generated by $\cL_{\operatorname{Heat}}^{(k)}$ to become an $(N,\eps,k)$-TPE:
\begin{align}\label{eq:gapapproxunitarydesign}
   \lambda(\cL^{(k)}_{\operatorname{Heat}})\ge a\,\Big(42500\left\lceil \frac{\ln(4k)}{\ln(d)}\right\rceil^2\,d^2\,k^5\,k^{\frac{3,1}{\ln(d)}} \Big)^{-1}\,\,.
\end{align}
The proof of \eqref{eq:gapapproxunitarydesign} required a control of the gap in terms of (i) the global control of the gap of $\cL_{\operatorname{Haar}}^{(k)}$ found in \eqref{haar}, and (ii)
% \begin{align}\label{eq:L2clusterdesign}
% \Big\|\frac{1}{|E|}\sum_{e\in E}E^{(k)}_e-\cD_{\operatorname{Haar}}^{(k)}:L_2\to L_2\,\Big\|\le 1-\Big(42500n\left\lceil \frac{\ln(4k)}{\ln(d)}\right\rceil^2\,d^2\,k^5\,k^{\frac{3,1}{\ln(d)}} \Big)^{-1}
% \end{align}
% for $E_e^{(k)}=\lim_{t\to\infty}e^{t\cL^{(k)}_e}$, where the bound \eqref{eq:L2clusterdesign} directly follows from \eqref{haar} and \eqref{eq:fromlambdatogap};
 a control of the gap of the local generator in \eqref{localbrownianmotiongap} using the representation theory of $\operatorname{SU}(d^2)$:
\begin{align}\label{localgap}
\lambda(\cL^{(k)}_{\operatorname{Heat},e})={a}\,.
\end{align}
In fact, \eqref{localgap} can be directly found by the transference principle of Proposition \ref{prop:transference}. An even stronger convergence for the local generator can be bound by transferring its Ricci curvature: up to normalization, we can assume that the Riemannian metric is given by the negative Killing form, and in this case the Ricci curvature of  $\operatorname{SU}(d^2)$ is $1/4$ (see \cite[Section 7.1]{Milnor76}). It then follows from the complete Bakry-\'{E}mery theorem \cite[Theorem 4.3]{li2020graph} and Proposition \ref{prop:transference} that
\begin{align}\label{eq:CMLSIfronricci}
2a\,\equiv 2\la(\cL_{\operatorname{Heat},e}^{(k)})\ge \al_{\text{CMLSI}}(\cL_{\operatorname{Heat},e}^{(k)})\ge a/2\,.
\end{align}
In particular, both the local CMLSI constant and the local spectral gap are independent of $k$.

% In fact, the main result of \cite{brandao2016local}, namely \eqref{eq:L2clusterdesign}, directly results in a bound on the number of times the averaging local quantum channel $|E|^{-1}\sum_{e\in E}E_e^{(k)}$ needs to be applied in order for it to approximate $\cD_{\operatorname{Haar}}^{(k)}$ in the $L_2$ sense. This bound can hence be directly applied together with \Cref{theorem:meta} in order to control the CMLSI constant of the generator
% \begin{align*}
%     \widetilde{\cL}^{(k)}:=\sum_{e\in E} E_e^{(k)}-\id\,,
% \end{align*}
% which in turn can hence be interpreted as a continuous time, entropic version of \cite[Theorem 5]{brandao2016local}.

We shall now discuss the convergence in terms of relative entropy for the four subsystem Lindbladians introduced above:
\begin{align}\label{generator:design}
 \mathcal{L}^{(k)}_\mu\ ,\  \mathcal{L}^{(k)}_{\operatorname{Haar}}\ , \ \mathcal{L}^{(k)}_{\mathcal{U}}\ , \ \mathcal{L}^{(k)}_{\operatorname{Heat}}\ .
 \end{align}
Note that any of these Lindbladians can be realized as a transferred semigroup on $\operatorname{U}(\cH_V)=\operatorname{U}(d^n)$. Nevertheless, due to the lack of information on either the spectral gap or the CMLSI constant of the corresponding classical Markov semigroup, the transference method does not give concrete estimates here. Instead, we use \Cref{thm:graphs} to derive asymptotically sharper bounds on their CMLSI constants.
\begin{theorem}\label{thm:designs}
Assume that the graph $G=(V,E)$ has a linear subgraph, i.e. a subgraph $G'=(V,E')\subseteq G$ whose vertices can be listed in the order
$v_1,v_2,\ldots,v_n$ such that the edges are $(v_i,v_{i+1})$ where
$i=1,\ldots,n-1$.
% Then for any $m\in \mathbb{N}$ and all $\rho\in\cD(\mathbb{C}^{dkn}\otimes \mathbb{C}^m)$,
% \begin{align}\label{eq:ATdesigns}
%     D(\rho\|\cD_{\operatorname{Haar}}^{(k)}(\rho))\le \frac{4\big(d^n(2d^n-1)\ln(k+1)+1 \big)}{\ln\Big[  \Big(1-\Big(42500\left\lceil \frac{\ln(4k)}{\ln(d)}\right\rceil^2\,d^2\,k^5\,k^{\frac{3,1}{\ln(d)}} \Big)^{-1}\Big)^{-1}\Big]}\,\sum_{e\in E}\,D(\rho\|E_{e}^{(k)}(\rho))\,.
% \end{align}
\begin{enumerate}
\item[$\operatorname{(i)}$] for the subsystem Lindbladian $\mathcal{L}_{\operatorname{Haar}}^{(k)}$ associated to the Haar measure,
   \begin{align}\label{bounddesignsbrownbetter}
 \frac{1}{4}\Big\lceil \frac{\,(2kn\ln(d)+1)}{\ln\Big(\Big(680000\left\lceil \frac{\ln(4k)}{\ln(d)}\right\rceil^2\,d^2\,k^5\,k^{\frac{3.1}{\ln(d)}} \Big)^{-1}+1\Big)}\Big\rceil^{-1}\le \alpha_{\operatorname{CMLSI}}(\mathcal{L}^{(k)}_{\operatorname{Haar}})\,.
  \end{align}
   \item[$\operatorname{(ii)}$] for the subsystem Lindbladian $\mathcal{L}_{\mu}^{(k)}$ associated to a symmetric universal measure $\mu$,
  \begin{align}\label{bounddesignshaarbetter}
\frac{1-\eps(k)}{4\binom{k+2d^2-1}{k}}\,\Big\lceil \frac{\,(2kn\ln(d)+1)}{\ln\Big(\Big(680000\left\lceil \frac{\ln(4k)}{\ln(d)}\right\rceil^2\,d^2\,k^5\,k^{\frac{3.1}{\ln(d)}} \Big)^{-1}+1\Big)}\Big\rceil^{-1}\le  \alpha_{\operatorname{CMLSI}}(\mathcal{L}^{(k)}_{{\mu}})\,,
\end{align}
where $\eps(k):=\|\Phi_\mu^{(k)}-\cD_{\operatorname{Haar}}^{( k)}:L_2\to L_2\|<1$.

\item[$\operatorname{(iii)}$] for the subsystem Lindbladian $\mathcal{L}_{\mathcal{U}}^{(k)}$ associated to a symmetric universal gate set $\mathcal{U}=\{U_1,\cdots, U_m\}$ with each
$U_i$ composed of algebraic entries,
 \begin{align}\label{bounddesignshaarbetter}
\frac{1-\eps}{4\binom{k+2d^2-1}{k}}\,\Big\lceil \frac{\,(2kn\ln(d)+1)}{\ln\Big(\Big(680000\left\lceil \frac{\ln(4k)}{\ln(d)}\right\rceil^2\,d^2\,k^5\,k^{\frac{3.1}{\ln(d)}} \Big)^{-1}+1\Big)}\Big\rceil^{-1}\le  \alpha_{\operatorname{CMLSI}}(\mathcal{L}^{(k)}_{\mathcal{U}})\,,
\end{align}
where $\varepsilon<1$ is defined in \eqref{eq:BGgap}.
\item[$\operatorname{(iv)}$] for the subsystem Lindbladian $\mathcal{L}_{\operatorname{Heat}}^{(k)}$ with local Brownian motion,
 \begin{align}\label{bounddesignshaarbetter}
\frac{a}{8}\Big\lceil \frac{\,(2kn\ln(d)+1)}{\ln\Big(\Big(680000\left\lceil \frac{\ln(4k)}{\ln(d)}\right\rceil^2\,d^2\,k^5\,k^{\frac{3.1}{\ln(d)}} \Big)^{-1}+1\Big)}\Big\rceil^{-1}\le  \alpha_{\operatorname{CMLSI}}(\mathcal{L}^{(k)}_{\operatorname{Heat}})\,.
\end{align}
\end{enumerate}
\end{theorem}
%Then, the generators $\cL^{(k)}_{\operatorname{Haar}}$, $\cL^{(k)}_{\operatorname{Heat}}$ and $\cL^{(k)}_{\mu}$ for some universal measure $\mu$ satisfy
% \begin{align}\label{boundsdesigns}
% \frac{\ln\Big[  \Big(1-\Big(42500\left\lceil \frac{\ln(4k)}{\ln(d)}\right\rceil^2\,d^2\,k^5\,k^{\frac{3,1}{\ln(d)}} \Big)^{-1}\Big)^{-1}\Big]}{4\big(d^n(2d^n-1)\ln(k+1)+1 \big)}\,   \le   \alpha_{\operatorname{CMLSI}}(\widetilde{\cL}^{(k)})\le \textcolor{red}{??}
% \end{align}
% begin{align}

\begin{proof}
% In order to prove the lower bound, we resort to \Cref{theorem:meta} for the symmetric map $\Phi:=\frac{1}{|E|}\sum_{e\in E}E_e^{(k)}$.
% In fact,
Without loss of generality, we can assume that $E=E'$ for the linear subgraph because adding edges to $G'$ results in a larger gap and a larger CMLSI constant. (i) follows from a direct application of \eqref{eq:CMLSIgraphs} for $C=c_G^{-1}$ together with $\alpha_{\operatorname{CMLSI}}(E_e^{(k)}-\id)\ge 1$. For (ii) and (iii), we recall that the index  $C_e$ is $C_{\operatorname{cb}}(\cB(\cH_e^{\ten k}):\text{ker}(\Phi_{\mu}^{(k)}-\id))$. By the discussion in Example \ref{exam:urep},
\[C_e=\sum_{i}m_i^2\]
where $m_i$ is the dimension of the irreducible representation of $\operatorname{GL}(d^2)$ contained in the $k$-tensor power representation $\pi_k:U\to U^{\ten k}$ on $(\mathbb{C}^{d^2})^{\ten k}$. It follows from the combination of \cite[Equation 3.5]{Halverson2019set} and \cite[Example 5]{Audenaert06} that
\begin{align*}
    C_e=\,\binom{k+2d^2-1}{k}\, .
\end{align*}
Then, it suffices to note that for a nearest neighbour graph, the maximal degree is $\gamma=2$ and it was proved in \cite[Lemma 30]{brandao2016local} that the constant $c_G \ge d^{-2kn}$. (ii) and (iii) similarly follow from the bound obtained in \eqref{eq:CMLSIgraphs}. For (ii), the local CMLSI constant of $\Phi_{\mu}^{(k)}-\id$  can be estimated by Theorem \ref{theo:CLSI} as
\[ \al_{\operatorname{CMLSI}}(\Phi_{\mu}^{(k)}-\id)\ge \frac{\la(\Phi_{\mu}^{(k)}-\id)}{C_e}\ .\]
where the spectral gap $\la(\Phi_{\mu}^{(k)}-\id)=1-\|\Phi_\mu^{(k)}-\cD_{\operatorname{Haar}}^{( k)}:L_2\to L_2\|\le 1-\eps(k)<0$ might still depend on $k$. (iii) is similar to (ii), where the only difference is that the local spectral gap is lower bounded by the classical Hecke operator in \eqref{eq:BGgap} by the transference method of Proposition \ref{prop:transference}.
(iv) is similar to (i) by noting that the CMLSI constant $\al_{\operatorname{CMLSI}}(\cL_{\operatorname{Heat},e}^{(k)})$ for the local Brownian motion is lower bounded by $a/2$ (cf \eqref{eq:CMLSIfronricci}).
\end{proof}
\begin{rem}{\rm In applying Theorem \ref{thm:graphs}, the alternative choice for $C$ is the index $$C_G:= C_{\operatorname{cb}}(\cB((\mathbb{C}^{d})^{\otimes nk}): \cD_{\operatorname{Haar}}^{(k)}(\cB((\mathbb{C}^{d})^{\otimes nk})) )\ .$$ %Nevertheless, this constant is exponentially worse than $c_G$.
More precisely,
$C_G$ is the sum of squares of the multiplicities of all the irreducible representations of $\mathcal{S}_k$ contained in its natural representation on $(\mathbb{C}^{d^n})^{\otimes k}$.
As argued before, by the combination of \cite[Equation 3.5]{Halverson2019set} and \cite[Example 5]{Audenaert06},
\begin{align*}
    C_G=\,\binom{k+2d^n-1}{k}\, .
\end{align*}
%which is asymptotic larger then $d^{2kn}$

%$(k+1)^{d^n(2d^n-1)}$

%For this, we simply need to express the commutant algebra $\{U^{\otimes k}|\,U\in\operatorname{SU}(d^n)\}'$. By Schur-Weyl duality, this corresponds to the block decomposition of the tensor product representation $\mathcal{S}_{k,d^n}$ of the symmetric group $\mathcal{S}_k$ which we already encountered in Example \ref{ex:NNSWAP}. By the formula provided in \eqref{formula}, we need to compute the sum of squares of the multiplicities of all the irreducible representations of $\mathcal{S}_k$: \textcolor{blue}{justify this better, add refs}
% \begin{align}\label{boundCG}
%     C_G=\binom{k+2d^n-1}{k}\le (k+1)^{d^n(2d^n-1)}\,.
% \end{align}

}
\end{rem}

\begin{rem}{One can also use the local gap $\lambda(\cL^{(k)}_{\operatorname{Heat},e})=a$ proved in \cite{onorati2017mixing} with our Theorem \ref{theo:CLSI} in order to derive a different lower bound than \eqref{bounddesignshaarbetter}, but this would result in a worsening of the dependence of the constant on $k$.}
\end{rem}

The above theorem can be used to derive new bounds on the time it takes for any of the system Lindbladians introduced above to become $\eps$-close to a $k$-unitary design. First, we introduce time continuous versions of the notions introduced in Definition \ref{def:designs}.

\begin{lemma}
Let $\cL^{(k)}$ be a system Lindbladian satisfying \eqref{eq:asymptoticsdesigns}, and let $0<\eps<1$. Then,
\begin{enumerate}
\item[$\operatorname{(i)}$]
For $t> \al_{\operatorname{CMLSI}}(\mathcal{L}^{(k)})^{-1}\ln \frac{1}{\eps}$, we have that for all $m\in \mathbb{N}$ and any state $\rho\in\cD(\mathbb{C}^{N^{km}})$,
\begin{align*}
    D(e^{t\mathcal{L}^{(k)}}\otimes \id_{m})(\rho)\|(\cD_{\operatorname{Haar}}^{(k)}\otimes \id_{m})(\rho) )\le \eps \,D(\rho\|(\cD_{\operatorname{Haar}}^{(k)}\otimes \id_{m})(\rho))\,.
    \end{align*}
\item[$\operatorname{(ii)}$] for $t> \lambda_{\operatorname{CMLSI}}(\mathcal{L}^{(k)})^{-1}\ln \frac{1}{\eps}$,
\begin{align*}
    \norm{e^{t\mathcal{L}^{(k)}}-\cD_{\operatorname{Haar}}^{(k)}:L_2\to L_2}{} \le \eps
  \end{align*}
\item[$\operatorname{(iii)}$] for $t> \al_{\operatorname{CMLSI}}(\mathcal{L}^{(k)})^{-1}\Big(\ln(4kn)+\ln\ln (d) +2\ln \frac{1}{\eps}\Big)$,
\begin{align*}
    \norm{e^{t\mathcal{L}^{(k)}}-\cD_{\operatorname{Haar}}^{(k)}}{\diamond} \le \eps
  \end{align*}
\end{enumerate}
% The similar assertion also holds for $\mathcal{L}^{(k)}_{\operatorname{Haar}},  \mathcal{L}^{(k)}_{\mathcal{U}} $ and $\mathcal{L}^{(k)}_{\operatorname{Heat}}$
\end{lemma}
\begin{proof}
(i) follows from the definition of the CMSLI constant, whereas (ii) is a direct consequence of the definition of the gap; (iii) follows from \eqref{eq:new} in Lemma \ref{lemma:eps} together with the bound:
\begin{align*}
    \binom{k+2d^n-1}{k}\le \frac{(k+2d^n-1)^k}{k!}\le (2d^n)^k\,.
\end{align*}
\end{proof}

Using the above estimates together with the bounds we found on the CMLSI constant in \Cref{thm:designs}, we get e.g. that the subsystem Lindbladian $\cL_{\operatorname{Heat}}^{(k)}$ with local Brownian motion generates an $\eps$-approximate unitary $k$-design in diamond norm for any
\begin{align*}
    t> \al_{\operatorname{CMLSI}}(\mathcal{L}^{(k)}_{\operatorname{Heat}})^{-1}\Big(\ln(4kn)+\ln\ln (d) +2\ln \frac{1}{\eps}\Big)= \frac{1}{a}\,\widetilde{\mathcal{O}}\big(n\operatorname{poly}(k) \big)\,.
\end{align*}

The bounds derived in \Cref{thm:designs} (as well as their analogues for the spectral gap as derived in \cite{onorati2017mixing}) can be thought of as worse case bounds, where we only considered the mixing arising from a unique nearest neighbour subgraph of $G$. However, the CMLSI constant as well as the gap should improve when increasing the connectivity of $G$. In the extreme case of a complete graph, in analogy with the classical Kac master equation \cite{carlen2011kinetic}, we would even expect the gap to increase linearly with the size $n$ of the graph, whereas $\alpha_{\operatorname{CMLSI}}$ should remain constant. We answer these questions in the next subsection.

\subsection{Quantum Kac model}\label{sec:quantumkac}
The celebrated Boltzmann equation describes the out-of-equilibrium dynamics of a gas of $n$ colliding particles and its evolution to equilibrium. In his attempt to derive the Boltzmann equation based on basic probabilistic assumptions, Kac \cite{Kac} introduced an equation, known today as the Kac master equation, whose simplest form can be stated as follows \cite{carlen2011kinetic}: consider $n$ particles moving on the line, and denote by $\mathbf{v}:=(v_1,\ldots, v_n)\in \mathbb{S}^{n-1}$ the velocity vector, where the coordinate $v_i$ is the velocity of the particle carrying label $i$. Define the operator on functions of the sphere $\mathbb{S}^{n-1}$
\begin{align*}
Q_n(f)(\mathbf{v}):=\frac{1}{\binom{n}{2}}\,\sum_{i<j}\,P_{i,j}(f)(\mathbf{v})
\,,\qquad \text{ where }\qquad \, P_{i,j}(f)(\mathbf{v}):=\int_{-\pi}^{\pi} f(R_{i,j}(\theta)\mathbf{v})\,\frac{d\theta}{2\pi}
\end{align*}
where the matrix $R_{i,j}(\theta)$ models a random change of velocities of particles $i$ and $j$ after an elastic collision occurred between them :
\begin{align*}
    (v_i,v_j)\mapsto (v_i^*(\theta),v_j^*(\theta))=(\cos(\theta)v_i-\sin(\theta)v_j,\sin(\theta)v_i+\cos(\theta)v_j)\,.
\end{align*}
In other words, the random interaction described above preserves the total kinetic energy $\sum_{i=1}^n v_i^2$. In the continuous time Markovian setting, the generator of the evolution is defined on functions of $\mathbb{S}^{n-1}$ as
\begin{align*}
    L_n^{\operatorname{Kac}}:=n\,(Q_n-\id)\,.
\end{align*}
Fundamental questions about the speed of convergence to the equilibrium $$e^{t L_n^{\operatorname{Kac}}}(f)\to \big(\int f d\sigma^{(n)}\big) \mathbb{1}_{\mathbb{S}^{n-1}}\, , \ \ \ \ t\to\infty $$ where $\sigma^{(n)}$ denotes the uniform probability measure on the sphere, were left open for decades after Kac's original paper (see also \cite{carlen2011kinetic} for a short review on the subject). In particular, \cite{Kac} had conjectured that the spectral gap satisfies
$\lambda(L_n^{\operatorname{Kac}})\ge c>0$ for a constant $c$ that is independent of $n$. This conjecture was proved by Janvresse \cite{janvresse2001spectral} who adapted Lu and Yau's martingale method \cite{lu1993spectral}. In \cite{carlen2000many}, the exact expression for the spectral gap was derived:
\begin{align*}
    \lambda(L_n^{\operatorname{Kac}})=\frac{1}{2}\,\frac{n+2}{n-1}\,.
\end{align*}
Later, \cite{villani2003cercignani} proved the following bound on the MLSI constant using the exact tensorization of the relative entropy with respect to the family $\{P_{i,j}\}_{i<j}$ of local conditional expectations:
\begin{align*}
    \alpha_{\operatorname{MLSI}}(L_n^{\operatorname{Kac}})\ge \frac{2}{n-1}\,.
\end{align*}

More recently, a quantum extension of the Kac model was proposed in \cite{carlen2019chaos}, a locally finite dimensional version of which can be described as follows: for a system of $n$-particle, the total Hilbert space is $\cH_V=\otimes_{v\in V} \cH_v$ and the local Hilbert spaces $\cH_v\equiv \cH$ are of dimension $\operatorname{dim}(\cH)=d$. Denoting by $h$ the one-particle Hamiltonian, the global free $n$-particle Hamiltonian of the system is
\begin{align*}
H_n:=\sum_{i=1}^n\,h_i\,\otimes\,\Id_{\{i\}^c}\,,
\end{align*}
This Hamiltonian plays the role of the kinetic energy $\sum_{i=1}^nv_i^2$ in the quantum setting. The collision between two particles is modeled as follows:
\begin{Def}
A collision specification $\mu$ is a probability measure on $\operatorname{U}(\cH^{\otimes 2})$, the unitary group acting the bipartite Hilbert space $\cH^{\otimes 2}$, such that
\begin{itemize}
    \item[(i)] For all $U\in \operatorname{supp}(\mu)$, $[U,H_2]=0$;
    \item[(ii)] %For any $U\in \operatorname{supp}(\mu)$, $U^\dagger \in \operatorname{supp}(\mu)$, moreover
    The map $U\mapsto U^\dagger$ leaves the measure $\mu$ invariant;
    \item[(iii)] Given the swap gate $S(|\psi\rangle \otimes |\varphi\rangle)=|\varphi\rangle\otimes|\psi\rangle$, %for all $U\in \operatorname{supp}(\mu)$, $SUS\in \operatorname{supp}(\mu)$, moreover
    the map $U\mapsto SUS$ leaves the measure $\mu$ invariant.
\end{itemize}
Given a collision specification $\mu$, the (local) collision operator is defined as
\begin{align*}
    \Phi_\mu(\rho):=\mathbb{E}_\mu\big[U\,\rho \,U^\dagger\big]\,.
\end{align*}
\end{Def}
\begin{rem}
Property (i) ensures that the model is elastic: the energy is preserved under the action of any ``allowed'' unitary. Property (ii) is a time reversibility condition: it ensures that the map $\Phi_\mu$ is self-adjoint. Finally, property (iii) ensures that
the two particles enter the collision specification in a symmetric way.
\end{rem}

Like in the classical case, we allow any two particles to interact, which means the system is governed by a subsystem Lindbladian associated to the complete graph $K_n:=(V,E)$ over $n$ vertices. From now on, we fix the edge set to be $E=\{(i,j)\ | i\neq j\}$. We define the quantum Kac generator corresponding to a collision specification $\mu$ \cite{carlen2019chaos} as follows,
\begin{align*}
    \cL_n^{\mu}:=n\,(\mathcal{Q}_{n,\mu}-\id) \,,\qquad\text{ where }\qquad \mathcal{Q}_{n,\mu}:=\frac{1}{\binom{n}{2}}\sum_{e\in E} (\Phi_{\mu})_{e}\,.
\end{align*}
 The Kac semigroup satisfies $ e^{t\cL^\mu_n}{\to} E^\mu_{K_n}$ as ${t\to\infty}$, where $E^\mu_{K_n}$ is the conditional expectation that projects onto the subalgebra
\begin{align*}
    \mathcal{F}^\mu_n:=\{U_e\otimes \Id_{e^c}\ |\ \,U\in\operatorname{supp}(\mu)\ , \forall \ e\}'\,.
\end{align*}
In the general setting, the channel $\Phi_\mu$ may not be a conditional expectation. Therefore, for each edge $e\in E$, we denote by
$\cF_e^\mu$
the kernel of the generator $(\Phi_\mu)_e-\id$, and by $E_e^\mu$ the conditional expectation projecting onto $\cF_e^\mu$. As in \Cref{sec:designs}, we also consider the subsystem Lindbladian
\begin{align*}
    \widetilde{\cL}^\mu_n:=\frac{n-1}{2}\,\sum_{e\in E}(E_e^\mu-\id)\,.
\end{align*}

\begin{example}[(Quantum Villani's theorem)]
Consider the case of $\mu$ being the Haar measure. This corresponds to the case of a trivial Hamiltonian. Then, the (local) collision operator is a $1$-design, i.e.~it reduces to the replacement of the state on edge $e$ by the maximally mixed state: for any $\rho\in \cD(\cH_V)$,
\begin{align*}
    (\Phi_{\operatorname{Haar}})_{e}(\rho)=\rho_{e^c}\otimes \frac{\Id_{e}}{d^2}\,.
\end{align*}
Namely, $(\Phi_{\operatorname{Haar}})_{e}$ is the partial trace over $\cH_e$.
Then the semigroup $(e^{t\cL_n^{\operatorname{Haar}}})_{t\ge 0}$ on the total system converges to the maximally mixed state $d^{-n}\Id$. Moreover, since partial traces over different (possibly adjacent) edges commute $\prod_{e}\Phi_{\operatorname{Haar}}=E^{\operatorname{Haar}}_{K_n}$, it follows from \cref{theo:lambdaindexcontrolAT} that
\begin{align*}
    D\Big(\rho\,\Big\|\,\frac{\Id}{d^n}\Big)\le \sum_{e\in E}\,D(\rho\|(\Phi_{\operatorname{Haar}})_e(\rho))\,.
\end{align*}
Since $\alpha_{\operatorname{CMLSI}}(\Phi_{\operatorname{Haar}}-\id)\ge 1$, we conclude that $\alpha_{\operatorname{CMLSI}}(\cL_n^{\operatorname{Haar}})\ge \frac{2}{n-1}$. Note however that in this simple case the approximate tensorization constant can be easily tightened. For instance, in the case $|V|=4$ one can chose any edge $e$ uniformly at random and use the chain rule followed by the data processing inequality to get
\begin{align}
    D\Big(\rho\Big\|\frac{\Id}{d^4}\Big)\le D(\rho\|(\Phi_{\operatorname{Haar}})_e(\rho))+D(\rho\|(\Phi_{\operatorname{Haar}})_{e'}(\rho))
\end{align}
where $e'$ stands for the edge that is complementary to $e$. Then an averaging procedure leads to
\begin{align}
      D\Big(\rho\Big\|\frac{\Id}{d^4}\Big)\le \frac{1}{3}\sum_{e\in E}D(\rho\|(\Phi_{\operatorname{Haar}})_e(\rho))\,.
\end{align}
A similar argument for general even $n$ gives $(n-1)D\Big(\rho\,\Big\|\,\frac{\Id}{d^n}\Big)\le \sum_{e\in E}\,D(\rho\|(\Phi_{\operatorname{Haar}})_e(\rho))$.
\end{example}

In the next theorem, we prove a similar behavior in the case of a general collision specification $\mu$, by combining the techniques that we used in \Cref{thm:graphs}. We need the following lemma on dividing the complete graph into subgraphs.

\begin{lemma}\label{lemma:graph}
Let $K_n=(V,E)$ be the $n$-complete graph. There exists at least
$\lfloor n/4\rfloor$ connected subgraphs $G_l=(V,E_l)\subset K_n$ with common vertex set $V$ and maximal degree $3$ such that the edge sets $E_1,\cdots, E_{\lfloor n/4\rfloor}$ are mutually disjoint.
\end{lemma}
\begin{proof}We first argue the case when $n=p$ is a prime number. We label the vertices as $\{1,\ldots ,p\}$. We choose the edge sets as
\[ E_l=\{(i,j)\ |\ (j-i) \operatorname{mod } p = l \} \ ,\  1\le l\le \lfloor p/2\rfloor\ .\]
It is clear that $E_l$'s are mutually disjoint. Moreover, it follows from the fact that $p$ is a prime number that for any integer $1\le l< p$, $l$ is co-prime with $p$ and hence $$\{1,(1+l)\operatorname{mod } p, (1+2l)\operatorname{mod } p,\cdots, (1+(p-1)l)\operatorname{mod } p\}=\{1,\ldots ,p\}.$$ This means that
$E_l$ connects the vertex set $\{1,\ldots ,p\}$ as a length $p$ cyclic graph (except for $p=2$ which has only one edge). Note that we only have $\lfloor p/2\rfloor $ many such subgraphs because $E_l=E_{p-l}$ for $l> \lfloor p/2\rfloor$. This proves the case for prime numbers.

For general $n$, we recall by Bertrand's postulate that there exists at least one prime number $\lfloor n/2\rfloor<p< n$. Let $E_l$ be the edge sets defined above connecting the vertex set $\{1,\ldots ,p\}$. We now connect the first $p$ vertices with $\{p+1,\cdots, n\}$ by enlarging the edge sets as follows
\[E_l'=E_l\cup \tilde{E}_l\ ,\ \tilde{E}_l:=\{(i,j)\ |\ j-i=p-l+1\ ,  p+1\le j\le n \}\ . \]
One can see that the added edge sets $\tilde{E}_l$ connect each vertex in $\{p+1,\cdots, n\}$ to a vertex in $\{1,\cdots, p\}$ and are also mutually disjoint. Hence $E_l'$ are mutually disjoint and connect the total vertex sets $\{1,\cdots, n\}$. That completes proof.
\end{proof}

We say a collision specification $\mu$ is \textit{universal} if for each connected subgraph $G=(V,E')$ of $K_n$, $\cap_{e\in E'}\cF_e^\mu=\cF_n^\mu$. This is in analogy with the notion of universal measure in \Cref{sec:designs}. In fact, defining the pushforward $\tilde{\mu}$ of $\mu$ with respect to the map $U\mapsto U^{\otimes k}$, one can easily see that the universality of $\tilde{\mu}$ is a consequence of the universality of $\mu$.
%At this point, we also require that for each connected subgraph $G=(V,E')$ of $K_n$, $\cap_{e\in E'}\cF_e^\mu=\cF_n^\mu$. In analogy with the notion of universal measure in \Cref{sec:designs}, we call a collision specification $\mu$ with this property \textit{universal}.
We will also need the following condition of $L_2$ decay of correlations:
\begin{Def}Let $\mu$ be collision specification and $\{E_e^\mu\}_{e\in E}$ be the associated conditional expectations $\{E_e^\mu\}_{e\in E}$. For any subgraph $G':=(V',E')\subseteq K_n$, denote by $E^\mu_{G'}$ the conditional expectation projecting onto $\cap_{e\in E'}\cF_e^\mu$. Then, the specification $\mu$ is said to satisfy the decay of correlations property if there exists $n_0\in\mathbb{N}$ such that for any $n\ge n_0$, any connected linear subgraph $G'=(V',E')$ with $n\ge n_0$, and any partitioning of $G'=G_A\cup G_B\cup G_C$ into three consecutive subgraphs, where $G_B$ separates $G_A$ from $G_C$, $|G_B|=\ell$ and $|G_C|=1$,
\begin{align*}
    \|E_{BC}^\mu\circ E_{AB}^\mu-E_{ABC}^\mu:\,L_2\to L_2\|\le \frac{1}{\sqrt{\ell}}\,,
\end{align*}
 where by a slight abuse of notations we denote e.g.~by $E^\mu_{BC}$ the conditional expectation $E^\mu_{G_B\cup G_C}$.
\end{Def}
%  Under the condition of decay of correlations, for any connected linear subgraph $G'=(V,E')$, the generator $\cL_{E'}^\mu:=\sum_{e\in E'}E^\mu_e-\id$ satisfies a lower bound on its gap which is independent of the system size $|V|$ (see \cite[Theorem 3]{nachtergaele1996spectral}).

\begin{theorem}\label{thm:Kac}
Let $\mu$ be a universal specification satisfying the decay of correlations property defined as above. Then, in the notations of \Cref{thm:graphs}, for all $m\in\mathbb{N}$ and any $\rho\in\cD(\cH_V\otimes  \mathbb{C}^{ m})$,
   \begin{align}
     D(\rho\|E^\mu_{K_n}(\rho)) \le 4\Big\lceil\frac{ (1+\ln(c))}{\ln\Big({\lambda_0^{-\lfloor n/4\rfloor}}\Big)}\Big\rceil\,\sum_{e\in E}\,D(\rho\|E_{e}^\mu(\rho))\,,
     \end{align}
 where $c:=\min\{C_{K_n},c_{K_n}^{-1}\}$ and for some constant $\lambda_0$ that is independent of $|V|=n$. Moreover, the $\operatorname{CMLSI}$ constants for the subsystem generators $\cL_n^\mu$ and $\widetilde{\cL}_n^\mu$ satisfy
\begin{align}
\frac{\lambda(\Phi^\mu-\id)}{4(n-1)C_e}\Big\lceil\frac{ (1+\ln(c))}{\ln\Big({\lambda_0^{-\lfloor n/4\rfloor}}\Big)}\Big\rceil^{-1}
\le    \alpha_{\operatorname{CMLSI}}(\mathcal{L}_n^\mu)\,,\qquad \frac{1}{4(n-1)}\Big\lceil\frac{ (1+\ln(c))}{\ln\Big({\lambda_0^{-\lfloor n/4\rfloor}}\Big)}\Big\rceil^{-1}\le    \alpha_{\operatorname{CMLSI}}(\widetilde{\mathcal{L}}_n^\mu)\,,
\end{align}
where $C_e=C_{\operatorname{cb}}(\cB(\cH_e):\mathcal{F}_e^\mu)$.

\end{theorem}

\begin{proof}
 Under the condition of decay of correlations, for any connected linear subgraph $G'=(V,E')$, the spectral gap of the generator $\cL_{E'}^\mu:=\sum_{e\in E'}E^\mu_e-\id$ admits a uniform lower bound $\lambda'$  which is independent of the system size $|V|$ (see \cite[Theorem 3]{nachtergaele1996spectral}). Therefore, by Lemma \ref{thm:detectability}, we have that for any subgraphs with maximal degree $3$,
\begin{align*}
    \lambda_0:=\big\|\prod_{e\in E'}E^\mu_e-E^\mu_{K_n}:L_2\to L_2\big\|=\big\|\prod_{e\in E'}E_e^\mu-E^\mu_{G'}:L_2\to L_2\big\|\le \frac{1}{\lambda'/16+1}<1
\end{align*}
 where the first identity follows by the condition of universality of $\mu$ that $E^\mu_{G'}=E^\mu_{K_n}$. Then by Lemma \ref{lemma:graph}, the complete graph $K_n$ contains $\lfloor n/4\rfloor$ connected subgraphs $G_1 ,\cdots, G_{\lfloor n/4\rfloor}$ with common vertex set $V$ and maximal degree $3$ such that theirs edge sets are mutually commuting. %Indeed, by Bertrand's postulate, there exists at least one prime number $p>\lfloor n/2\rfloor$. Then, choosing any $p$ vertices from $V$, one can generate at least $\lfloor p/2\rfloor$ linear connected graphs out of these $p$ vertices with non intersecting edge sets: one way to see this is to index the vertices by $\{1,\ldots ,p\}$. Then, we connect all the vertices that are $1 \operatorname{mod}(p)$, $2\operatorname{mod}(p)$, up to $\lfloor p/2\rfloor\operatorname{mod}(p)$ vertices away from each other. Finally, the $n-p<\lceil n/2\rceil$ remaining vertices can be connected to each of the $p$-vertex graphs that we just constructed using unused edges from $K_n$: for instance, connect vertex $p+1$ to vertex $1$, $p+2$ to $2$, etc. for the first subgraph, then $p+1$ to $2$, $p+2$ to $3$ etc. for the second subgraph etc.
 Then, for each $G_\ell:=(V,E_\ell)$, we denote by $\Pi_\ell$ the operator $\prod_{e\in E_\ell}E^\mu_e$ and $R:= \prod_{e\in E\backslash \cup_\ell E_\ell}E^\mu_e$, where the products are taken in an arbitrary order. Thus,
 \begin{align*}
     \big\| \prod_{e\in E}E_e^\mu-E^\mu_{K_n}:L_2\to L_2\big\|=\big\|R\, \prod_{\ell}\Pi_\ell(\id-E^\mu_{K_n}):L_2\to L_2\big\|\le \lambda_0^{\lfloor n/4\rfloor}\,.
 \end{align*}
  We conclude by using the same argument as in the proof of \Cref{thm:graphs}, with $\lambda$ being given by $\lambda_0^{\lfloor n/4\rfloor}$.
 \end{proof}

It turns out that $k$-designs on the complete graph also satisfy the conditions of \Cref{thm:Kac}, which results in a tightening of the bounds found in \Cref{thm:designs}. For instance:
\begin{corollary}
Let $\mu$ be a universal measure of $\operatorname{U}(d^2)$ in the sense of definition \ref{def:universalmeasure}.
For any $k\in \mathbb{N}$, the subsystem Lindbladians $\cL_{\operatorname{Haar}}^{(k)}$ defined in \Cref{def:Lhaar} in the case of the complete graph satisfies
\begin{align*}
  &\frac{1}{2}\Big\lceil\frac{({2kn}\ln(d)+1)}{{\lfloor n/4\rfloor}\ln\big((\lambda(k)/16+1)\big)} \Big\rceil^{-1}  \le \alpha_{\operatorname{CMLSI}}(\cL_{\operatorname{Haar}}^{(k)})\, ,
%   &\frac{\lambda(\Phi_\mu^{(k)}-\id)}{4\binom{k+2d^2+1}{k}}\Big\lceil\frac{({2kn}\ln(d)+1)}{\ln\big((\lambda(k)/16+1)^{\lfloor n/4\rfloor}\big)} \Big\rceil^{-1}  \le \alpha_{\operatorname{CMLSI}}(\cL_{\operatorname{\mu}}^{(k)})\
\end{align*}
where $\lambda(k):=\Big(42500\left\lceil \frac{\ln(4k)}{\ln(d)}\right\rceil^2\,d^2\,k^5\,k^{\frac{3,1}{\ln(d)}} \Big)^{-1}$.
In particular, $\alpha_{\operatorname{CMLSI}}(\cL_{\operatorname{Haar}}^{(k)})=\Omega(1)$ as a function of $n$.
\end{corollary}

\begin{proof}
It suffices to show that the pushforward $\tilde{\mu}$ of $\mu$ with respect to the map $U\mapsto U^{\otimes k}$ satisfies the assumptions of \Cref{thm:Kac}. First, universality of $\tilde{\mu}$ is a consequence of the universality of $\mu$. Also, the decay of correlations for $\tilde{\mu}$ was proved in \cite[Lemma 18]{brandao2016local}, which resulted in the bound \eqref{haar} on the gap of $\cL_{\operatorname{Haar}}$. Therefore, we conclude by replacing $\lambda'$ in the proof of \Cref{thm:Kac} by \eqref{haar}, by using the bound of the constant $c_G \ge d^{-2kn}$ from \cite[Lemma 30]{brandao2016local}, and adapting the normalization of the generator to fit the definitions of \Cref{sec:designs}.
\end{proof}

\begin{rem}
By the same method as above, we also deduce e.g. that $\lambda(\cL_{\operatorname{Heat}}^{(k)}) =\Omega(n)$ on the complete graph using the reverse detectability lemma \cite{anshu2016simple}. This is an improvement over the result of \cite{onorati2017mixing}.
\end{rem}

\section{Discussion and open problems}
We end our paper by discussing some open problems. \Cref{theo:CLSI} proves that any GNS-symmetric quantum Markov semigroup $(\cP_{t}:\cB(\cH)\to \cB(\cH))_{t\ge 0}$ on a finite dimensional Hilbert space $\cH$ satisfies the complete modified log-Sobolev inequality (CMLSI) with constant
\begin{align} \label{eq:bound}\al_{\operatorname{CMLSI}}(\cL)\ge \frac{\la(\cL)}{C_{\operatorname{cb},\tau}(E_*)}\end{align}
where $\cL$ is the generator and $E_*=\lim_{t} \cP_t$. In the primitive case (unique invariant state), $C_{\operatorname{cb},\tau}(E_*)\sim d^2$ with $d=\dim(\cH)$. On the other hand, it was proven that for a primitive semigroup with invariant state $\si$,
\begin{align} \label{eq:bound2} \frac{\la(\cL)}{\ln(\mu_{\min}(\si)^{-1})+2}\le \frac{\al_{\operatorname{LSI}}(\cL)}{2}\le \al_{\operatorname{MLSI}}(\cL)\end{align}
where $\al_{\operatorname{MLSI}}$ is the optimal constant for the modified log-Sobolev inequality and $\al_{\operatorname{LSI}}$ is the optimal constant for the $L_2$-log-Sobolev inequality (LSI), which is known to be equivalent to hypercontractivity \cite{olkiewicz1999hypercontractivity}. Here $\mu_{\min}(\si)$ is the minimal eigenvalue of $\si$ and $\ln(\mu_{\min}(\si)^{-1})\sim \ln d$. Our lower bound on $\al_{\operatorname{CMLSI}}$ controls any amplification $\cP_{t}\ten \id_{\mathbb{M}_n}$, in contrast with $L_2$-log-Sobolev inequality/hypercontractivity bound which fails for $\cP_{t}\ten \id_{\mathbb{M}_n}$ for any $n>1$. It remains open whether the bound \ref{eq:bound} can be improved asymptotically.
\begin{problem}\label{prob:optimalasymptotic}
Does there exist a general lower bound on the $\operatorname{CMLSI}$ constant of the form $\al_{\operatorname{CMLSI}}(\cL)\ge \la(\cL) o(d^2)^{-1}$?
\end{problem}
\noindent Here $o(d^2)$ stands for any function that would be asymptotically smaller than the square function.

Our second question concerns the strong data processing inequality (SDPI). It was proven in \cite{muller2016entropy} that for a primitive unital quantum channel $\Phi$,
\begin{align}\label{eq:bound3}\al_{\operatorname{SDPI}}(\Phi)\ge 1-\frac{1}{2}\al_{\operatorname{LSI}}(\Phi^*\Phi-\id)\end{align}
where $\al_{\operatorname{LSI}}(T^*T-\id)$ is the LSI constant of the map $T^*T-\id$ seen as the generator of a QMS. This combined with \eqref{eq:bound} gives upper bounds on SDPI constant for primitive unital channel. Nevertheless, since LSI generally fails for non-primitive semigroups, this approach does not apply to CSDPI. In order to find better (C)SDPI constant, we propose the following question:
\begin{problem}\label{prob:MLSISDPI}
 Can we find a lower bound on $\al_{\operatorname{SDPI}}(\Phi)$ in terms of the modified log-Sobolev constant $\al_{\operatorname{MLSI}}(\Phi^*\Phi-\id)$ for any non-primitive quantum channel $\Phi$?
\end{problem}
\noindent Note that in general $2\al_{\operatorname{MLSI}}(\cL)\ge  \al_{\operatorname{LSI}}(\cL)$, so a positive answer to Problem \ref{prob:MLSISDPI} would be potentially even stronger than \eqref{eq:bound3}. Moreover, combined with our Theorem \ref{theo:CLSI}, such a positive solution would also give a lower estimate on the CSDPI constant in terms of the spectral gap and the index.

Recall that our SDPI constant $\al_{\operatorname{SDPI}}$ is defined as the optimal constant $0\le \al\le 1$ such that
\[ D(\Phi(\rho)\|\Phi\circ E_*(\rho))\le \al D(\rho\|E_*(\rho))\pl,\]
for any state $\rho$.
Here $E_*(\rho)$ is the decoherence free part of the state $\rho$ in the sense that for a GNS symmetric channel $\Phi$, $\Phi^2\circ E_*(\rho)=E_*(\rho)$ and $\displaystyle\lim_{n\to \infty}\norm{\Phi^n(\rho)-\Phi^n\circ E_*(\rho)}{}=0$. This is a natural choice analogous to MLSI. Nevertheless, the data processing inequality asserts that $D(\Phi(\rho)\|\Phi (\si))\le  D(\rho\|\si)$ for any two states $\rho$ and $\si$. Indeed, in Theorem \ref{thm:local}, we prove that for a state $\omega$, the best local constant $\al'(\si)$, which satisfies
\[ D(\Phi(\rho)\|\Phi(\si))\le \al'(\sigma) D(\rho\|\si),\]
for all state $\rho$ with $E_*(\rho)=E_*(\omega)$, can be two-sided controlled by the corresponding $\chi_2$ contraction coefficient. (Here the restriction $E_*(\rho)=E_*(\omega)$ is needed because without it the constant $\al(\omega)$ would often be trivially equal to $1$.) It is then natural to ask whether for a quantum channel $\Phi$, there is a non-trivial upper bound on $\al(\omega)$ uniformly in $\omega$. Namely,
\begin{problem}\label{prob:generalSDPI}
For a finite dimensional quantum channel $\Phi$, does there exist a constant $\al'(\Phi)<1$ such that
\[ D(\Phi(\rho)\|\Phi (\si))\le  \al'(\Phi) D(\rho\|\si)\pl,\]
for all states $\rho,\si$ with $E_*(\rho)=E_*(\omega)$?
\end{problem}
\noindent Such a constant $\al'(\Phi)$ leads to a stronger notion of contraction than our definition of $\al_{\operatorname{SDPI}}$. It is closer to the classical strong data processing inequality studied in \cite{polyanskiy2017strong,Raginsky16}, which was proven to be equivalent to the contraction coefficient of (classical) $\chi_2$-divergence. Note that by our Theorem \ref{thm:local}, it also suffices to show that there is a global contraction coefficient on the quantum $\chi_2$ divergence $$\chi_2(\rho,\si)=\norm{\rho-\si}{\si^{-1}}^2=\int_0^\infty \tr\Big((\rho-\si)\,\frac{1}{\si+s}\,(\rho-\si)\,\frac{1}{\si+s}\Big)\,ds$$ for all states $\sigma$.

\begin{appendix}
\section{Detectability lemma}\label{sec:detectabilitylemma}

Let $\mathcal{H}$ be a finite dimensional Hilbert space, and denote by $\langle .,.\rangle$ its inner product. Let $\{E_i\}_{i\in J }$ be a family of orthogonal projections on $\mathcal{H}$ and define
% Therefore, these projections can be partitioned into $\ell$ subsets $L_1\ldots L_{\ell}$, called layers, such that each layer consists of pairwise commuting elements. For each layer $k\in \{1,\ldots,\ell\}$, we define the \textit{layer projector} $\Phi_k^*:=\prod_{i\in L_k} E_i $, as well as the \textit{detectability lemma operator}
\begin{align*}
  &\Phi^*:= \prod_{i\in J}\,E_{|J|-i+1}\ , \ \mathcal{L}:=\sum_{i\in J}E_i-\id
\end{align*}
where the product is of arbitrary ordering. Finally, we denote by $E_J$ the orthogonal projection onto the intersection $\cap_{i\in J}\operatorname{Ran}(E_i)$ of the ranges of the projections $E_i$.

\begin{theorem}\label{thm:detectability}
 Suppose each $E_i$ commute with all but at most $g$ other $E_j$'s. Then
 \begin{align*}
 \|\Phi^*-E_J:\mathcal{H}\,\to\, \mathcal{H}\|^2\le\frac{1}{\lambda(\mathcal{L})/g^2+1}\,,
 \end{align*}
 where $\lambda(\mathcal{L})$ denotes the spectral gap of $\mathcal{L}$.
\end{theorem}

\begin{proof}
We follow the proof provided in \cite{anshu2016simple}. Without loss of generality, we label the ordering of projections in $\Phi^*$ as $\Phi^*=E_{|J|}\ldots E_1$. Let $X\in \mathcal{H}$ with $\|X\|_{\mathcal{H}}\le 1$ and $E_J(X)=0$. For $Y:=\Phi^*(X)$, we have $$-\langle Y,\mathcal{L}(Y)\rangle\equiv \sum_{i\in J}\|(\id-E_i)(Y)\|_{\mathcal{H}}^2\,.$$
Denoting by $N_i$ the subset of indices of projectors that do not commute with $E_i$, we have by triangle inequality that for all $j\in N_i$
\begin{align*}
    \|(\id-E_i)E_j \ldots E_1(X)\|_{\mathcal{H}}\le \|(\id-E_i)E_{j-1}\ldots E_1(X)\|_{\mathcal{H}}+\|(\id-E_i)(\id-E_j)E_{j-1}\ldots E_1(X)\|_{\mathcal{H}}\,.
\end{align*}
Iterating this procedure while moving $\id-E_i$ to the right each time $E_i$ reaches a projector $E_j$ for $j\in N_i$ until $\id-E_i$ reaches $E_i$ and vanishes, we get
\begin{align*}
    \|(\id-E_i)(Y)\|_{\mathcal{H}}&\le \sum_{j\in N_i} \|(\id-E_j)E_{j-1}\ldots E_{1}(X)\|_{\mathcal{H}}\,,
\end{align*}
where we also used that $\|(\id-E_i)\|_{\mathcal{K}}\le 1$. Now, since $|N_i|\le g$ by assumption, we have by Jensen's inequality that
\begin{align*}
    \|(\id-E_i)(Y)\|_{\mathcal{H}}^2\le g\sum_{j\in N_i} \|(\id-E_j)E_{j-1}\ldots E_1(X)\|_{\mathcal{H}}^2 \,.
\end{align*}
Then we sum over $i\in\{1,
\ldots, |J|\}$ so that each of the terms in the above sum appears at most $g$ times,
\begin{align*}
    -\langle Y,\mathcal{L}(Y)\rangle& \le g^2\sum_{j=2}^{|J|}\,\|(\id-E_j)E_{j-1}\ldots E_1(X)\|_{\mathcal{K}}^2\\ &\overset{(1)}{=} g^2\sum_{j=2}^{|J|}\,\|E_{j-1}\ldots E_1(X)\|_{\mathcal{K}}^2-\|E_jE_{j-1}\ldots E_1(X)\|_{\mathcal{K}}^2\\
    &\overset{(2)}{=} g^2 \big(\|E_1(X)\|_{\mathcal{K}}^2-\|E_{|J|}\ldots E_1(X)\|_{\mathcal{K}}^2 \big)
    \le g^2\,(1-\|Y\|_{\mathcal{K}}^2)\,,
    \end{align*}
    where $(1)$ follows from the orthogonality of the projections $E_j$ and $\id-E_j$, whereas $(2)$ follows by a telescopic sum argument. Equivalently, we have found that
    \begin{align*}
   \|Y\|_{\mathcal{K}}^2\le \frac{1}{1-\frac{\langle Y,\mathcal{L}(Y)\rangle}{g^2\|Y\|_{\mathcal{K}}^2} }\le \frac{1}{1+\frac{\lambda(\mathcal{L})}{g^2}}\,.
    \end{align*}
    which completes the proof.
\end{proof}
\begin{rem}
We observe that the result of \cite{anshu2016simple} allow for an arbitrary ordering of the projectors in $\Phi^*$.
\end{rem}

\section{Completely positive order relations from norm estimates}\label{sec:normstoorder}

In this appendix, we provide two generic strategies to derive the order relations needed in \Cref{theorem:meta} in the case of symmetric generators. Our first method is a variant of \cite[Lemma 3.15]{gao2020fisher} (see also \cite[Corollary 1.8]{laracuente2019quasi}). Let $\cM$ be a finite dimensional von Neumann algebra equipped with trace $\tau$ and let $\cN\subset \cM$ be a von Neumann subalgebra. Note that we do not fix the normalization of $\tau$ at this point. For $p\ge 1$, the space of $p$-integrable operators in $\cM$ is denoted by $L_p(\cM,\tau)\equiv L_p(\tau)$, with associated norm
\begin{align*}
    \|x\|_{L_p(\cM)}\equiv \|x\|_p:=\tau\big(|x|^p\big)^{\frac{1}{p}}\,.
\end{align*}
We also need the notion of an amalgamated $L_p$ norm \cite{junge2010mixed}: for $1\le p\le \infty$ the $L_\infty^p(\cN\subset\cM)$ norm is given by
\[ \norm{X}{L_p^\infty(\cN\subset\cM)}:=\sup_{ X=aYb}\norm{aYb}{L_p(\cM)}\ ,\]
where the supremum is taken over $a,b\in \cN$ with $\norm{a}{2p}=\norm{b}{2p}=1$. Its operator space structure is given by (see \cite[Appendix]{gaoindex})
\[ \mathbb{M}_n(L_\infty^p(\cN\subset\cM))=L_\infty^p(\mathbb{M}_n(\cN)\subset \mathbb{M}_n(\cM))\ . \]
It was proved in
 \cite[Theorem 3.9]{gaoindex}
that
\begin{align}\label{eqidindex}
\norm{\id:L_\infty^2(\cN\subset\cM)\to \cM}{\operatorname{cb}}^2
=\norm{\id:L_\infty^1(\cN\subset\cM)\to \cM}{\operatorname{cb}}=C_{\operatorname{cb}}(\cM:\cN)\pl.
\end{align}
Let $E_\cN:\cM\to \cN$ be the trace preserving conditional expectation onto $\cN$. Recall that
there exists a module basis $\{\xi_i\}_{i=1}^n\in \cM$ satisfying (\cite[Theorem 3.15]{paschke1973inner}, see also \cite[Cons\'{e}quence 1.8]{baillet1988indice}):
\[E_{\cN}(\xi_i^\dagger\xi_j)=\delta_{ij}p_i,\]
where $p_i\in \cN$ are some projections. Also recall that $\Phi:\cM\to \cM$ is a $\cN$-bimodule map if
$\Phi(aXb)=a\Phi(X)b$ for all $a,b\in \cN$ and $X\in \cM$.
We have $\Phi\circ E_\cN=E_{\cN}$ if $\Phi$ is unital and $E_{\cN}\circ \Phi= E_{\cN}$ if $\Phi$ is trace preserving. For a $\cN$-bimodule map $\Phi$,
we have by \cite[Lemma 3.12]{gao2020fisher}
\begin{align}\label{eqgapcb}
\la=\norm{\Phi:L_2(\cM)\to L_2(\cM)}{}=\norm{\Phi:L_2(\cM)\to L_2(\cM)}{\operatorname{cb}} = \norm{\Phi:L_\infty^2(\cM)\to L_\infty^2(\cM)}{\operatorname{cb}} \,.
\end{align}
Next, we define the \textit{module Choi operator} of a bi-modular map $\Phi$ as
\[ \chi_\Phi=\sum_{i,j=1}^n\ket{i}\bra{j}\ten \Phi(\xi_i^\dagger\xi_j)\in \cB(l_2^n)\ten \cM\pl,\]
where $l_2^n$ denotes the space of $n$-dimensional vectors with associated norm
\begin{align*}
    \|u\|_{l_2^n}=\Big(\sum_{i}|u_i|^2\Big)^{\frac{1}{2}}
\end{align*}
and $\{\ket{i}\}_{i=1}^n$ is a fixed orthonormal basis in $l_2^n$. Thus $\Phi$ and $\chi_\Phi$ determine each other because for each $x\in \cM$, we have a unique decomposition $x=\sum_{i}\xi_i x_i $ with $x_i\in\cN$ satisfying $p_ix_i=x_i$. Indeed, we have $x_i=E_\cN(\xi_i^\dagger x)$.
Moreover, $\Phi$ is completely positive if and only if $\chi_\Phi$ is a positive operator in $\cB(l_2^n)\ten \cM$. Indeed, for any finite family $y_1,\cdots,y_m\in \cM$, we assume the decomposition $y_j=\sum_{l}\xi_l x_{jl}$ with $x_{jl}\in \cN$. Then
\begin{align*}
(\id\ten \Phi)\big(\sum_{i,j}\ket{i}\bra{j}\ten y_i^\dagger y_j\big)
=& \sum_{i,j}\ket{i}\bra{j}\ten \Phi(y_i^\dagger y_j)
\\
=& \sum_{i,j,k,l}\ket{i}\bra{j}\ten x_{ik}^\dagger \Phi(\xi_k^\dagger \xi_l)x_{jl}
\\
=& \big(\sum_{i,k}\ket{i}\bra{k}\ten x_{ik}^\dagger \big)\chi_\Phi
\big(\sum_{j,l}\ket{l}\bra{j}\ten x_{jl}\big)\ ,
\end{align*}
from which the equivalence claimed directly follows. We also recall \cite[Lemma 3.14]{gao2020fisher} that for a $\cN$-bimodule map $\Phi$,
\begin{align}\label{eq:L1inftytochi}
\norm{\chi_\Phi}{\cB(l_2^n)\ten \cM}=\norm{\Phi:L_\infty^1(\cN\subset\cM)\to \cM}{\operatorname{cb}}.
\end{align}

We are now ready to state and prove the first main Lemma of this section:

\begin{lemma}\label[Lemma]{lemmageneral}Let $\cN\subset \cM$ be a von Neumann subalgebra and let $E_\cN:\cM\to \cN$ be the corresponding trace preserving conditional expectation. Let $\Phi:\cM\to \cM$ be a $\cN$-bimodule map.
\begin{enumerate}
\item[$\operatorname{i)}$] if $\norm{\chi_\Phi-\chi_{E_\cN}}{\cB((l_2^n))\ten \M}{}\le \eps< 1$, then
\[ (1-\eps)E_{\cN}\le_{\operatorname{cp}} \Phi\le_{\operatorname{cp}} (1+\eps) E_{\cN} \ .\]
\item[$\operatorname{ii)}$] if $\Phi$ is a $\tau$-symmetric quantum channel and $\la:=\norm{\Phi-E_\cN:L_2(\cM,\tau)\to L_2(\cM,\tau)}{}<1$,
then for $k>\frac{\ln C_{\operatorname{cb}}(\cM:\cN)}{-\ln \lambda}$,
\[ (1-\epsilon)E_\cN\le_{\operatorname{cp}} \Phi^k\le_{\operatorname{cp}} (1+\epsilon)E_\cN  \]
for $\epsilon=\la^kC_{\operatorname{cb}}(\cM:\cN)<1 $.
\end{enumerate}
\end{lemma}
\begin{proof} For i) we first observe that the Choi operator of $E_{\cN}$
\[\chi_{E_\cN}=\sum_{i,j=1}^n\ket{i}\bra{j}\ten E_{\cN}(\xi_i^\dagger \xi_j)=\sum_{i}\ket{i}\bra{i}\ten p_i\]
is a projection.
By the assumption $\norm{\chi_\Phi-\chi_{E_\cN}}{\cB((l_2^n))\ten \M}{}\le \eps<1$, we may write $\chi_{E_\cN}-\chi_\Phi =\alpha-\beta$ with $\displaystyle 0\le \alpha,\beta \le \eps $. %Write $\al=\sum_{i,j}\ket{i}\bra{j}\ten \al_{i,j}$ and $\beta=\sum_{i,j}\ket{i}\bra{j}\ten \beta_{i,j}$. We have $\al_{i,j}-\beta_{i,j}=E(\xi_i^*\xi_j)-T(\xi_i^*\xi_j)$.
Let $\{y_1,\cdots,y_m\}$ be a finite family of elements in $ \cM$ with decomposition $y_j=\sum_{l}\xi_l x_{jl}$ with $x_{jl}=p_jx_{jl}\in \cN$, and denote $Y=\sum_{j}\ket{1}\bra{j}\ten y_j$. Then,
\begin{align*}
(\id\ten \Phi)(Y^\dagger Y)-(\id\ten E_{\cN})(Y^\dagger Y)
=& \sum_{i,j,k,l}\ket{i}\bra{j}\ten x_{ik}^\dagger p_k(\Phi-E_{\cN})(\xi_k^\dagger \xi_l)p_jx_{jl}\\
=& \Big(\sum_{i,k}\ket{i}\bra{k}\ten x_{ik}^\dagger \Big)\Big(\chi_\Phi-\chi_{E_{\cN}}\Big)
\Big(\sum_{j,l}\ket{l}\bra{j}\ten x_{jl}\Big)
\\=& \Big(\sum_{i,k}\ket{i}\bra{k}\ten x_{ik}^\dagger p_k\Big)\Big(\al-\beta\Big)
\Big(\sum_{j,l}\ket{l}\bra{j}\ten p_jx_{jl}\Big)\,.\end{align*}
Since $0\le \al\le\eps$,
\begin{align*}\Big(\sum_{i,k}\ket{i}\bra{k}\ten x_{ik}^\dagger p_k\Big)\al\Big(\sum_{j,l}\ket{l}\bra{j}\ten p_jx_{jl}\Big) \le &\epsilon \, \Big(\sum_{i,k}\ket{i}\bra{k}\ten x_{ik}^\dagger p_k\Big)\Big(\sum_{j,l}\ket{l}\bra{j}\ten p_jx_{jl}\Big)\\ =& \,\eps\sum_{i,j,k}
\ket{i}\bra{j}\ten x_{ik}^\dagger p_kx_{jk}\\ =&\eps\sum_{i,j,l,k}
\ket{i}\bra{j}\ten x_{ik}^\dagger E_\cN(\xi_k\xi_l)x_{jl}
=\,\eps\,(\id \ten E_\cN)(Y^\dagger Y)
\end{align*}
and similarly for $\beta$.
Thus we showed that
\[-\eps \,(\id\ten E_{\cN})(Y^\dagger Y)\le (\id\ten \Phi)(Y^\dagger Y)-(\id\ten E_{\cN})(Y^\dagger Y)\le \eps \,(\id\ten E_{\cN})(Y^\dagger Y) \]
for any $Y=\sum_{j}\ket{1}\bra{j}\ten y_j$, which proves i) because any $X\ge 0$ in $\mathbb{M}_m\ten \cM$ can be written as a sum $X=\sum_{j}Y_j^\dagger Y_j$ of such $Y$'s.

For ii), we first observe that since $\Phi$ is trace perserving and unital, we have by bimodule property that $\Phi\circ E_\cN=E_\cN\circ \Phi= E_{\cN}$ and hence $(\Phi-E_{\cN})^k=\Phi^k-E_{\cN}$. Then, we can use identity \eqref{eq:L1inftytochi} between the $L_\infty$ norm of the module Choi operator of $\Phi^k$ and its $L_\infty^1$ norm, so that:
\begin{align*}
&\norm{\chi_{\Phi^k}-\chi_{E_{\cN}}}{\cB(l_2^n)\ten \cM}\\&~~~~~=\norm{\Phi^k-E_{\cN}:L_\infty^1(\cN\subset \cM)\to \cM }{\operatorname{cb}}\\
&~~~~~=\norm{(\Phi-E_{\cN})^k:L_\infty^1(\cN\subset \cM)\to \cM }{\operatorname{cb}}
\\
&~~~~~=\norm{\id:L_\infty^1(\cN\subset \cM)\to L_\infty^2(\cN\subset \cM)}{\operatorname{cb}}
\cdot
\norm{(\Phi-E_{\cN})^k:L_\infty^2(\cN\subset \cM)\to L_\infty^2(\cN\subset \cM)}{\operatorname{cb}}\\ &~~~~~\qquad\qquad\qquad\qquad\qquad\qquad\qquad \qquad\, \qquad\cdot
\norm{\id:L_\infty^2(\cN\subset \cM)\to  \cM}{\operatorname{cb}}
\\&~~~~~= C_{\operatorname{cb}}(\cM:\cN)^{1/2}\cdot\la^k\cdot  C_{\operatorname{cb}}(\cM:\cN)^{1/2}=C_{\operatorname{cb}}(\cM:\cN)\la^k\,,
\end{align*}
where we also used \eqref{eqidindex} and \eqref{eqgapcb}. Then, the assertion ii) follows from i).
\end{proof}

The next main result of the present appendix is an approach to obtain $\operatorname{cp}$ orders which follows the idea of \cite[Appendix A]{brandao2016local}. Here, we restrict ourselves to the algebra $\cM=\cB(\cH)$ and take $\tau:=\tr$ to be the standard matrix trace. Given a map $\Phi:\cB(\cH)\to \cB(\cH)$, the standard (normalized)
Choi-Jamiolkowski matrix $J_\Phi$ is
\[J_\Phi:=(\id\ten \Phi)(\ket{\psi}\bra{\psi})\in \cB(\cH\ten \cH)\ ,\]
where $\ket{\psi}=\frac{1}{\sqrt{d_\cH}}\sum_{i}\ket{i}\ket{i}$ is the maximally entangled state on $\cH\ten \cH$. It is well-known that
$J_\Phi\ge 0$ if and only if $\Phi$ is completely positive, and
$J_\Phi$ is a state if  and only if $\Phi$ is a quantum channel.

\begin{lemma}\label{lemma:brandao}
Let $E_{\cN}:\cB(\cH)\to \cN$ be the trace preserving conditional expectation onto a subalgebra $\cN$.
Let $\Phi:\cT_1(\cH)\to \cT_1(\cH)$ be a quantum channel such that $\Phi\circ E_{\cN}=E_{\cN}\circ \Phi=E_{\cN}$. Suppose $J_{E_\cN}\ge C_E^{-1} P$ for some $C_E>0$ where $P$ is the support projection of $J_{E_\cN}$.
\begin{enumerate}
    \item [$\operatorname{i)}$] If $J_\Phi$
has support $P_{\Phi}\le P$ and for some $0\le \eps<1$
 \begin{align*}
   \|J_\Phi-J_{E_{\cN}}\|_{\infty}\le \eps\, C_E^{-1}\, ,
 \end{align*}
 then
\[ (1-\epsilon)E_\cN\le_{\operatorname{cp}} \Phi\le_{\operatorname{cp}} (1+\epsilon)E_\cN  \,.\]
\item [$\operatorname{ii)}$] If for each $k$, $J_{\Phi^k}$ has support projection $P$ and $\la:=\norm{\Phi-E_{\cN}:L_2(\cB(\cH),\tr)\to L_2(\cB(\cH),\tr)}{}<1$. Then for
  $k>
  \frac{\ln(C_E)}{-\ln(\lambda)}$,
 \[ (1-\epsilon)E_\cN\le_{\operatorname{cp}} \Phi^k\le_{\operatorname{cp}} (1+\epsilon)E_\cN  \,,\]
where $\epsilon =\lambda^k\,C_E<1$.
\end{enumerate}
\end{lemma}
\begin{proof}If $J_\Phi$ has support $P_\Phi\le P$, then
\begin{align*}
 J_\Phi-J_{E_{\cN}}\,\le\, \| J_\Phi-J_{E_{\cN}}\|_\infty\, P\,\le \,\|J_\Phi-J_{E_{\cN}}\|_\infty\, C_E\,J_{E_{\cN}}\le \eps\, J_{E_{\cN}}\ ,
 \end{align*}
Similarly,
\[ J_{E_{\cN}}-J_\Phi\le \eps \, J_{E_{\cN}}\ ,\]
Thus we have
\[ (1-\eps)J_{E_{\cN}}\le J_\Phi\le (1+\eps)J_{E_{\cN}}\]
which implies i). For ii), we note that the maximally entangled state $\ket{\psi}\bra{\psi}$ is a pure state, so that
 \begin{align*}
 \norm{J_{\Phi^k}-J_{E_{\cN}}}{\infty}\le& \norm{J_{\Phi^k}-J_{E_{\cN}}}{2}
=\norm{(\id \ten(\Phi^k-E_\cN))(\ket{\psi}\bra{\psi})}{2}
 =\norm{(\id \ten(\Phi-E_\cN)^k)(\ket{\psi}\bra{\psi})}{2}
\\ \le&  \norm{(\id \ten(\Phi-E_\cN)^k)(\ket{\psi}\bra{\psi})}{2}
\\ \le &\norm{\Phi-E_\cN:L_2(\cB(\cH),\tr)\to L_2(\cB(\cH),\tr)}{}^k=\la^k
 \end{align*}
Then assertion ii) follows from i).
\end{proof}
We end this appendix by comparing the approaches proposed in Lemmas \ref{lemmageneral} and \ref{lemma:brandao} in two cases.

\begin{example}[(Trivial subalgebra)] Consider the trivial subalgebra $\mathbb{C}1\subset \cB(\cH)$. Then for $\dim(\cH)=d$, $$C_{\operatorname{cb}}(\cB(\cH),\mathbb{C})=C_E=d^2\ . $$
\end{example}

\begin{example}[(Unitary Representations of compact groups)]\label{exam:urep}
{\rm Let $G$ be a compact group and let $\pi:G\to \cB(\cH)$ be a unitary representation. The fixed point space for the conjugation action $\al:G\curvearrowright \cB(\cH)$ defined as $\al_g(x)=\pi(g)x\pi(g)^\dagger $ is
$\cN=\pi(G)'=\{x\in \cB(\cH)\ |\ x\pi(g)=\pi(g)x\ \}$, that is the commutant of the representation. Suppose $\pi$ admits the decomposition $\pi=\oplus_{i\in \operatorname{Irr}(G)}\big(\pi_{i}\ten \id_{n_i}\big)$ as a direct sum of irreducible representations $\pi_i:G\to \mathbb{M}_{m_i}$ with multiplicity $n_i$. Then $\cN= \oplus_{i\in \operatorname{Irr}(G)}\big(\mathbb{C}\Id_{m_i}\ten \mathbb{M}_{n_i}\big)$ is a direct sum of matrix algebras $\mathbb{M}_{n_i}$ with multiplicity $m_i$.
The trace preserving conditional expectation is $$E_{\cN}=\oplus_{i}\\ \big(\tr_{m_i}\ten \id_{n_i}\big)\,,$$
where $\tr_{m_i}(x)=\tr(x)\frac{\Id}{m_i}$ is the partial trace map and $\id_{n_i}$ is the identity map on $\mathbb{M}_{n_i}$. The Choi-Jamiolkowski state of
$E_{\cN}$ is
\[J_{E_{\cN}}=\frac{1}{d_\cH}\oplus_{i}\Big( \frac{n_i}{m_i}(\Id_{m_i}\ten \Id_{m_i})\ten \ket{\psi_{n_i}}\bra{\psi_{n_i}}\Big)\]
where $\Id_{m_i}\in \mathbb{M}_{m_i}$ is the identity operator and $\ket{\psi_{n_i}}=\frac{1}{\sqrt{n_i}}\sum_{j=1}^{n_i}\ket{j}\ket{j}$ is the maximally entangled state in $\mathbb{M}_{n_i}\ten \mathbb{M}_{n_i}$.
The support projection is
\[ P=\oplus_{i}\Big( (\Id_{m_i}\ten \Id_{m_i})\ten \ket{\psi_{n_i}}\bra{\psi_{n_i}}\Big)\ .\]
Then $J_{E_{\cN}}\ge C_E^{-1} P$ for
\begin{align*}
    C_E=d_\cH\max_{i}\frac{m_i}{n_i}\,.
\end{align*}
 On the other hand, the cb-index is
\[C_{\operatorname{cb}}(\cB(\cH):\cN)=\sum_{i}m_i^2\pl.\]
 In particular, when $G$ is a finite group and $\pi_L: G\to \cB(l_2(G))$ is its the left regular representation, we have by Schur-Weyl Theorem that $n_i=m_i$ and
 \[C_E=d_\cH=\sum_{i}m_i^2=C_{\operatorname{cb}}(\cB(l_2(G)):\cN),\]
 where $\cN$ is isomorphic to the group algebra $\mathbb{C}G$.}
\end{example}

\end{appendix}

\end{document}